\documentclass{article}
\usepackage{amsfonts}
\usepackage{epsfig}
\usepackage{psfrag}

\usepackage{amssymb}
\usepackage{indentfirst}
\usepackage{bm}
\usepackage{amssymb,epsfig,amsmath,accents,setspace,amsthm}

\makeatletter

\newdimen\mainfontsize \mainfontsize=1\@ptsize pt
\topskip=\mainfontsize
   \@maxdepth=\maxdepth

\makeatother

\newtheorem{thm}{Theorem}[section]

\newtheorem{cor}[thm]{Corollary}

\newtheorem{lem}[thm]{Lemma}

\newcommand{\bnu}{{\bar{\nu}}}

\newcommand{\cf}{\mbox{$\cal F$}}
\newcommand{\cG}{\mbox{$\cal G$}}

\newcommand{\tcf}{\tilde {\mbox{$\cal F$}}}

\newcommand{\tX}{\tilde{X}}

\newcommand{\br}{\mbox{$\mathbb R$}}

\newcommand{\bp}{\mbox{$\mathbb P$}}
\newcommand{\be}{\mbox{$\mathbb E$}}
\newcommand{\tbp}{\mbox{$\tilde{\mathbb P}$}}
\newcommand{\tbe}{\mbox{$\tilde{\mathbb E}$}}

\newcommand{\tO}{\tilde{\Omega}}
\newcommand{\srho}{\sqrt{\rho}}

\makeatletter
\@addtoreset{equation}{section}

\makeatother

\begin{document}
\bibliographystyle{plain}
\begin{center}
{\LARGE Stochastic evolution equations in portfolio credit modelling with applications to exotic credit products}
\end{center}
\vspace*{.1 true in}
\begin{center}
{\Large N.~Bush}\footnote{Mathematical Institute, University of Oxford,
24-29 St Giles, Oxford OX1 3LB, UK.\par
~\mbox{ } E-mail:  nick.bush@maths.ox.ac.uk},
{\Large B.~M.~Hambly}\footnote{Mathematical Institute, University of Oxford,
24-29 St Giles, Oxford OX1 3LB, UK.\par
~\mbox{ } E-mail:  hambly@maths.ox.ac.uk},
{\Large H.~Haworth}\footnote{Credit Suisse\par
~\mbox{ } E-mail:  helen.haworth@linacre.oxon.org},
{\Large L. Jin}\footnote{Mathematical Institute, University of Oxford,
24-29 St Giles, Oxford OX1 3LB, UK.\par
~\mbox{ } E-mail:  jin@maths.ox.ac.uk}
and 
{\Large C.~Reisinger}\footnote{Mathematical Institute, University of Oxford,
24-29 St Giles, Oxford OX1 3LB, UK.\par
~\mbox{ } E-mail:  reisinge@maths.ox.ac.uk\par
~\mbox{ } The authors are grateful to Credit Suisse for the market data used for the calibration.\par  
~\mbox{ } The views expressed in this article are those of the authors and not those of Credit Suisse.}

\vspace*{.1 true in}
\large{\today}
\end{center}

\begin{abstract}
We consider a structural credit model for a large portfolio of credit risky assets where the correlation
is due to a market factor. By considering the large portfolio limit of this system we show the existence
of a density process for the asset values. This density evolves according to a stochastic partial differential
equation and we establish existence and uniqueness for the solution taking values in a suitable function space.
The loss function of the portfolio is then a function
of the evolution of this density at the default boundary. We develop numerical methods for pricing and 
calibration of the model to credit indices and consider its performance pre and post credit crunch. Finally, we give further examples illustrating the valuation of exotic credit products, specifically forward starting CDOs.
\end{abstract}

\section{Introduction}\label{section:intro}

The rapid growth of the credit derivatives market from 2000-2007 led to the development 
of increasingly complex credit instruments requiring new mathematical models for pricing and risk 
management. The subsequent contraction due to the credit crunch has placed even more emphasis 
on the importance of understanding the risks involved in dealing with complex credit products.
Our aim in this paper is to extend standard large portfolio credit models by introducing dynamics
and working with the infinite dimensional limit. Although this model has shortcomings (as inherent
in the underlying structural model), we provide a mathematical basis for the
development of more realistic extensions.

The two natural approaches to credit modelling that have been extensively developed are the structural
approach and the reduced form approach, and each has been extended to the portfolio 
setting in a variety of ways. We consider a dynamic large portfolio model obtained by taking 
the large portfolio limit of a multidimensional structural model. By taking this limit
we obtain a stochastic partial differential equation which models the
evolution of the value of a large basket of underlying assets. The key quantities for multiname credit are
then certain functions of the solution of this stochastic partial differential equation.

Our motivation for the development of our structural evolution model came originally from the 
lack of dynamics in the credit market's standard pricing methodology. This absence of dynamics made 
pricing some structured credit instruments very difficult and credit market developments since 
mid-2007 have further exposed the limitations of existing approaches. There is still a need for a new 
generation of models to enable a better fitting as well as understanding of the risks inherent in some of the more 
complex products. For instance the existence of 5, 7 and 10-year index and bespoke tranches requires 
a model that can fit the entire correlation skew term structure, not just the correlation skew for a given 
time horizon. Also the introduction of forward starting tranches, options on tranches and STCDOs with 
trigger features requires information on the dynamics of spreads and information on the timings of 
default for their pricing.
By investigating the behaviour of our simplified model, we are able to gain an insight into which aspects of 
dynamic models are important for the pricing of more exotic structured credit products. This information can 
then be used to help guide future model development. 

Our model follows a bottom-up approach in which the individual entities in a credit basket are modelled. 
This approach (whether structural or reduced form) has been widely used, primarily as a result of the 
introduction of copulas and the subsequent conditionally independent factor (CIF) models. These models 
allow the problem of specifying the marginal distributions and the market co-movements to be separated 
and through the choice of specific copulas has led to simple, easy to implement and computationally 
efficient techniques for pricing credit products.



However, as the portfolio credit market expanded it became clear that these models were
unable to cope with some of the new instruments. Copula and CIF models have no dynamics
to speak of; nowhere is it specified how their parameters or underlyings evolve.
Furthermore, they only model expected defaults within one time period (making copula 
parameters time dependent leads to prices that are not arbitrage free). For instruments
such as collateralised debt obligations (CDO) this is not an issue as they are
essentially one period instruments, but for those with stronger timing features this is
not acceptable. These two points make it impossible to price dynamic instruments such as
options on CDOs and very difficult to price multi-period instruments such as forward
starting CDOs. Thus our purpose is to develop a relatively simple dynamic extension
of a CIF to the large portfolio setting.

An alternative multi-asset route is a top-down approach where 
the joint default distribution is modeled directly without regard to the single name market. 
The correlation is an inherent property of the quantity being modeled and thus does not need to be 
specified. Using the top-down approach, frameworks similar to
that of the HJM interest rate models have been developed for the joint loss distribution.

Although many of the exotic credit instruments have traded infrequently, especially post the credit crunch, 
their introduction highlighted the need for a more sophisticated approach to portfolio credit modelling. 
There is a large and rapidly growing literature in this area, so we only mention a
few papers \cite{carmona}, \cite{sirzar}, \cite{filip}, \cite{bayrak}, \cite{papasir}.
Top-down approaches include the Markov chain model in \cite{schon2} and the 
models of \cite{brigo1}, \cite{errais1} and \cite{graz1}. Reduced-form approaches have 
been extended to more than one issuer via correlated stochastic parameters. A relatively 
tractable example is the intensity-gamma model by \cite{joshi1}; another is 
the affine jump diffusion model of \cite{mortensen1}. \cite{okane2} provides an overview 
of some of the main bottom-up approaches. 

\subsection{Structural models}

Our model falls into the class of multi-dimensional structural models and we take the approach of
modelling the empirical measure of the asset prices in the basket when the underlyings have dynamics
linked through a factor model. The pricing of CDOs is then a function of the limit of the empirical measure
of the large basket. 

Structural models are based on the premise that when a company's asset value falls below
a certain threshold barrier a default is triggered. The first model of this type was
introduced by \cite{merton1} and then extended by \cite{blackcox1}.  To date, there are many variants of 
this model but the basic type is as
follows. Let $A_t$ be the asset value of a company whose evolution is governed by
\begin{displaymath}
\frac{dA_t}{A_t} = \mu\, dt + \sigma\, dW_t,
\end{displaymath}
where $\mu$ is the mean rate of return on the assets, $\sigma$ is the asset volatility
and $W_t$ is a standard Brownian motion. If we denote the default threshold barrier by
$B_t$ we define the \emph{distance to default}, $X_t$, as
\begin{equation}
\label{distance-to-default} X_t = \frac{1}{\sigma}\left(\,\textrm{log}A_t -
\textrm{log}B_t\,\right).
\end{equation}
The event of a default by time $t$ is now expressed as the event that $X$ hits $0$ before time $t$. 

Structural models are appealing due to their intuitive economic interpretation and the
link they provide between the equity and credit markets. They introduce spread dynamics
and allow market participants to hedge spread risk with the underlying equity of the
reference entity. Defaults are endogenously generated within the model and recovery rates
do not need to be determined until after a default occurs.

There are however downsides that affect the practical applicability of structural models.
Due to the diffusive nature of the asset process, and the assumption of perfect
information regarding asset values and default thresholds, any credit event generated by
the model is predictable. The immediate consequence is short term credit spreads that are
near zero: a fact contradicted by empirical evidence. Extensions that try to address these
issues include CreditGrades$^{TM}$ described in \cite{finger1}, as well as \cite{duffie2}, \cite{zhou2},
\cite{zhou3}, \cite{hilberink1} and \cite{schoutens1}. As structural models are extended in these
ways their analytic complexity increases dramatically. Credit spread prices can then no
longer be expressed in closed form and numerical methods must be employed for pricing.
Another downside is that calibration of the model parameters is not a straightforward
exercise.

Due to the popularity enjoyed by CIF and copula models, multidimensional structural
models have typically received less attention; as a result, the literature on this
subject is relatively sparse. 
The first authors to incorporate default correlation into first passage models were
\cite{zhou1} and \cite{hull2}. The former extended the Black-Cox framework to include
correlated asset value processes, with hitting times being calculated from a time
dependent barrier in closed form for two risky assets. \cite{hull2} followed Zhou's
approach and moved to a higher dimensional space but had to sacrifice the analytic
results. In \cite{hull1} the asset value processes for a multi-dimensional structural model are
correlated via a set of common factors. In this setting piecewise default barriers are calibrated
to match market prices and Monte-Carlo simulation is used to value single tranche CDOs (STCDOs).
Other recent papers using a structural approach include \cite{fouque1}, \cite{haworth1}, 
\cite{carmona} and \cite{crepey}. We aim to develop a model which can allow pricing of
exotic options on CDO tranches and note that there has been some discussion of such products 
in \cite{hull3}, \cite{jackson1}. 

\subsection{The SPDE model}

The starting point for our model is very similar to that used in \cite{hull1}. We will develop a
simple model in this paper in which all assets have the same constant volatility and are correlated via a single 
market factor. A more general version, in which the volatility and correlation are functions, can be found in \cite{LJthesis}. 
Let $(\Omega^N,\cf^N,\bp^N)$ denote a probability space for a market consisting of $N$ different companies 
whose asset values $A_t$ at time $t$ evolve under the risk neutral measure $\bp^N$ according to a diffusion process
given by
\begin{equation}
dA_t^i = r A_t^i \,dt + \sigma\sqrt{1-\rho}A_t^i\,dW_t^i + \sigma\sqrt{\rho} A_t^i\,dM_t,\;\;i=1,\dots,N
\label{assetsde}
\end{equation}
up until the hitting time of a barrier $B^i$ or the horizon $T$. We assume $W_t^i$ and $M_t$ are Brownian 
motions satisfying
\begin{displaymath}
d\left[ W_t^i , M_t \right] = 0 \qquad \forall i\\
\end{displaymath}
and
\begin{displaymath}
d\left[ W_t^i , W_t^j \right] = \delta_{ij}\,dt,
\end{displaymath}
where we have written $[.,.]$ for the quadratic covariation and will use $[.]$ for the quadratic variation,
and $\sigma>0$ is a constant and $\rho \in [0,1)$ is the constant correlation. 
Note the co-dependence between the asset
processes is provided solely by the Brownian motion $M_t$ which can be thought of as a market wide factor
influencing all of the assets. 

Thus we can write (\ref{assetsde}) in terms of the distance to default process $X^i_t = (\ln A^i_t - \ln B^i)/\sigma$, 
with constant barrier $B^i$, as
\begin{equation}
\begin{array}{rcl}
dX_{t}^{i}&=&
\mu dt+\sqrt{1-\rho}dW_{t}^{i}+\sqrt{\rho} dM_{t}, \quad t<T_{0}^{i},\\
X_{t}^{i}&=&0, \quad t\geq T_{0}^{i},\\
X_{0}^{i}&=&x^{i}>0, \\
T_{0}^{i}&=&\inf\{t:X_{t}^{i}=0\},\\
\end{array}
\label{1}
\end{equation}
for $i=1, 2, ... , N$, where $\mu=(r-\frac12\sigma^2)/\sigma$.

It does not matter how we label our assets so make the following assumptions.
We will assume that $\{X_{0}^{1}, ... , X_{0}^{N}\}$ is a family of exchangeable, $[C_B,\infty)$-valued random 
variables with $\mathbb{E} (X_0^i) < \infty$, where the constant $C_B>0$. 
We assume that this initial distribution is independent of $\{W^{i}\}$ and $M$. 

By construction we see that our system extends to an infinite system as $N\to\infty$ and we will show that there
is a limit empirical measure whose density satisfies an SPDE. We will write $(\Omega, \cf,\bp)$ with associated expectation
operator $\be$ for the limit probability space containing the full infinite asset value model.

In order to state our main mathematical result we need some further notation. Let $(\Omega^M, \cf^M,\bp^M)$ be 
a probability space supporting a one-dimensional Brownian motion $(M_t,\cf_t)$. Let $\cG^M$ denote 
the $\sigma$-algebra of predictable sets on $\Omega^M\times (0,\infty)$ associated with the filtration $\cf^M_t$ and 
$H^1((0,\infty))=\{f: f\in L^2((0,\infty)), f'\in L^2((0,\infty))\}$, where $L^2((0,\infty)) = \{f: \int_0^{\infty}
f^2 dx<\infty\}$. We write $L^2(\Omega^M \times (0,T), {\cG}^M,H^1((0,\infty))) = \{f(\omega,t,.): 
f(\omega,t,.)\in H^1((0,\infty)),  f(\omega,t,.) \mbox{ is $\cf^M_t$-measurable},\\ 
\be^M \int_0^T \|f(\omega,t)\|_{H^1}^2 dt<\infty\}$. 
We also write $\delta_x$ for a Dirac measure at the point $x$.

Let $\bar{\nu}_{N,t}$ denote the equally weighted 
empirical measure for the entire portfolio given by
\begin{equation}
\label{empiricalmeasure} \bar{\nu}_{N,t} = \frac{1}{N}\sum_{i=1}^N \delta_{X_t^i}.
\end{equation}

\begin{thm}\label{thm:main}
The limit empirical measure $\bar{\nu}_t=\lim_{N\to\infty} \nu_{N,t}$ exists and is a probability 
measure with two components, $\bar{\nu}_t = L_t\delta_0 + \nu_t$. 
The measure $\nu_t$ is a measure on $(0,\infty)$ with density $v(t,x)$, which is 
the unique solution in $L^2(\Omega^M \times (0,T), {\cG^M},H^1((0,\infty)))$ of the SPDE
\begin{equation}\label{simplifiedspde}
\left\{ \begin{array}{ll} dv = -\frac{1}{\sigma}\left(r-\frac{1}{2}\sigma^2\right) v_x \,dt 
+ \frac{1}{2}v_{xx}\,dt - \sqrt{\rho} v_x\,dM(t), & \vspace{0.2cm} \\
v(0,x) = v_0(x), \quad v(t,0) = 0. &
\end{array} \right.
\end{equation}
The weight of the Dirac mass at 0 is the loss function
\[ L_t = 1-\int_0^{\infty} v(t,x) dx. \] 
\end{thm}

The price of the credit products that we consider are functions of the loss function $L_t$.
There is no analytic solution for this SPDE, though it can be viewed as the Zakai equation for a filtering problem, and 
thus we require numerical techniques for its solution. One natural approach is just to use a Monte Carlo technique to 
simulate the whole basket, and for small sizes of basket this would be a 
natural approach. However, as the basket size increases, the numerical solution of the limit SPDE becomes 
more computationally efficient and we discuss this in our simplified setting.

An outline of the paper is as follows. We begin with a description of the mechanics and basic valuation methods 
of synthetic collateralised debt obligations 
in Section \ref{section:cdo} in order to provide the necessary background for later 
sections. The mathematical core of the work
is in Section~\ref{section:derivation} where we develop our infinite dimensional model for portfolio credit starting 
from a multidimensional structural model and prove Theorem~\ref{thm:main}. 
We make strong assumptions with the aim of delivering a relatively simple, tractable model that encapsulates the
information required to calculate the loss distribution for a portfolio of risky assets.
The aim in Section \ref{section:numerics} is to develop a suitable numerical scheme for solving the SPDE.
Section \ref{section:pricing} discusses the calibration and performance of the model when pricing tranches
of the iTraxx before and after the credit crunch.

\section[Collateralised Debt Obligations]{Collateralised debt obligations}\label{section:cdo}

Collateralised Debt Obligations (CDOs) are securitized interests in pools of credit risky
assets. These assets can include mortgages, bonds, loans and credit derivatives. The CDO
repackages the credit risk of the reference portfolio into multiple tranches that are
then passed on to investors. 
Prior to the `credit crunch' the synthetic CDO, credit indices and single name Credit Default Swap (CDS) 
market together 
made up the majority of the total traded notional in the credit derivative market. However the index 
tranche market is currently the only area that is still active. The bespoke CDO business has yet to return
although there are a few signs of activity.

Although there are many different types of CDO, here we will 
be focussing on what is known as a synthetic CDO i.e. one whose collateral pool consists 
entirely of credit default swaps. It is possible to trade single
tranches within a synthetic CDO without the entire structure being constructed. In this
case the two parties of the transaction, the protection buyer and protection seller,
exchange payments as if the CDO had been set-up. The performance of this single tranche
CDO (STCDO) is dependent on the number of defaults that occur in the reference portfolio
during the lifetime of the contract.

Each tranche is defined by two points that determine its place within the capital
structure: the attachment point and the higher valued detachment point. These are usually
expressed as a percentage of the total portfolio notional. The tranche notional is
defined as the difference between the attachment and detachment points. When losses are
incurred (the loss is the notional of the defaulted entity corrected
for recovery), and the cumulative loss in the collateral pool is between the attachment and
detachment point, the seller pays the buyer an amount equal to the loss incurred within
the tranche. The tranche notional is then reduced by this amount. This means that when
the cumulative loss exceeds the detachment point the tranche notional is zero. In return
for this protection, the buyer pays a quarterly premium based off a fixed spread and the
outstanding tranche notional. 
Say we have $N$ entities in our reference credit portfolio each with notional $N_0$. We
define the total loss $L_t$ on the portfolio as
\begin{equation}
\label{lossvariable} L_t = \sum_{i=1}^{N} L_i 1_{\lbrace \tau_{i} \leq t \rbrace},
\end{equation}
where $L_{i} = N_0(1-R_{i})$, $R_{i}$ and $\tau_{i}$ are the recovery rate and default time
of the $i$-th entity respectively. If we assume the recovery rate is the same across all credit
entities and equal to a value $R$ then we can write
\begin{equation}
\label{lossvariable2} L_t = N_0(1-R)\sum_{i=1}^{N} 1_{\lbrace \tau_{i} \leq t \rbrace}.
\end{equation}
The outstanding tranche notional, $Z_t$, of a single tranche within a synthetic CDO is
given by
\begin{equation}
\label{outstandingnotional} Z_t = [d - L_t]^{+} - [a - L_t]^{+},
\end{equation}
and the tranche loss $Y_t$ as
\begin{equation}
\label{trancheloss}
Y_t = [L_t-a]^{+} - [L_t-d]^{+},
\end{equation}
where $a$ is the tranche attachment point and $d$ is the tranche detachment point.

As for a Credit Default Swap (CDS) the value of a STCDO is given by the difference between 
the fee leg and
the protection leg. The protection buyer pays a regular fixed spread on the outstanding
notional of the tranche. We denote the payment dates by $T_{i}$, $1 \leq i
\leq n$, the intervals by $\delta_i = T_i-T_{i-1}$ and the value of a bank account at time $t$ by
$b(t)$. Then the value of the fee leg is given by
\begin{equation}
\label{cdofeeleg} s V^{fee} = s \sum_{i=1}^{n} \frac{\delta_{i}}{b(T_i)}
\mathbb{E}[Z_{T_{i}}],
\end{equation}
where the expectation is with respect to a suitable pricing measure.
The protection seller only makes payments to the buyer when the tranche incurs losses,
and the value of this payment is equal to the change in the tranche loss $Y_t$. However,
we can express the value of the protection leg in terms of the outstanding tranche
notional $X_t$ as follows
\begin{equation} \label{cdoprotectionleg} V^{prot} = \sum_{i=1}^{n}
\frac{1}{b(T_{i})}\mathbb{E}[Z_{T_{i-1}} - Z_{T_{i}}],
\end{equation}
assuming that the losses are paid at the coupon dates.
As in a CDS contract the par spread $s$ of the tranche is chosen
to make the initial value zero hence is calculated as
\begin{equation}
\label{cdospread}
s = \frac{V^{prot}}{V^{fee}}.
\end{equation}
{From} (\ref{cdofeeleg}) and (\ref{cdoprotectionleg}) we see that the key to finding the
par spread is obtaining the distribution of the outstanding tranche notional; from
(\ref{outstandingnotional}), this is equivalent to finding the distribution of the loss
$L_t$. As all portfolio credit derivatives are essentially options on this loss variable
the heart of every multiname credit model is determining its distribution.

\section{An infinite dimensional structural model}\label{section:derivation}

Our aim in this section is to establish Theorem~\ref{thm:main}.
We will begin by describing the system (\ref{1}) by a measure valued process and showing that
there is a limit empirical measure for the infinite system. We then proceed to establish its behaviour near 0 before proving that
its evolution can be captured by an SPDE. 

\subsection{The limit empirical density}

Recall the equally weighted 
empirical measure for the entire portfolio is given by
\begin{equation*}
 \bar{\nu}_{N,t} = \frac{1}{N}\sum_{i=1}^N \delta_{X_t^i}.
\end{equation*}
We can write this as
\[ \bar{\nu}_{N,t} = L_{N,t} \delta_0 + \nu_{N,t}, \]
where 
\[ \nu_{N,t} = \frac1N \sum_{i=1}^N \delta_{X_t^i}1_{\{t < T_i^0\}}, \;\; L_{N,t} = \frac1N \sum_{i=1}^N 
1_{\{t \geq T_i^0\}}. \]
Note that $L_{N,t}$ is a loss function in that it is the proportion of companies that have 
defaulted by time $t$.

Let $\mathbb{R}_+=[0,\infty)$. We write $\mathcal{P}(\mathbb{R}_+)$ for the set of probability measures 
on $\mathbb{R}_+$ and $\mathcal{P}(C_{\mathbb{R}_+}[0,\infty))$ for the set of probability measures on 
$C_{\mathbb{R}_+}[0,\infty)$ where the topology is always that of weak convergence. 
We write $C_{\mathcal{P}(\mathbb{R}_+)}[0,\infty)$ for the continuous $\mathcal{P}(\mathcal{R}_+)$-valued 
functions on $[0,\infty)$. 

\begin{thm}
There exists a
$C_{\mathcal{P}(\mathbb{R}_+)}[0,\infty)$-valued random variable $\bar{\nu}$ such that
\[ \bar{\nu}_{t}=\lim_{N\rightarrow\infty}\bar{\nu}_{N,t}=\lim_{N\rightarrow\infty}
    \frac{1}{N}\sum_{i=1}^{N}\delta_{X_{t}^{i}}, \; \text{$\bp$- a.s. }.\]
We also have a decomposition for the limit into two subprobability measures
\[\bar{\nu}_{t}=L_t \delta_0 + \nu_t. \]
\label{Th2}
\end{thm}

\begin{proof}
Let us denote the system with the same dynamics but without default by $\{\tilde{X}_{t}^{i}\}$. Then
\[X_{t}^{i}=\tilde{X}_{t}^{i}1_{\{\min\limits_{0\leq s\leq t}\tilde{X}_{s}^{i}>0\}}:=
F\left(\tilde{X}_{s}^{i}, 0\leq s\leq t\right).\]
Since $F$ is independent of $i$, in order to show that $\{X^{i}\}$ is exchangeable in $C_{\mathbb{R}}[0, \infty)$ 
we only need to show $\{\tilde{X}^{i}\}$ is exchangeable in $C_{\mathbb{R}}[0, \infty)$.

Since $\tX_{t}^{i}=X_{0}^{i}+\mu t+\sqrt{1-\rho} W_{t}^{i}+\sqrt{\rho}M_{t}$ for all $t$, and 
$\{X_{0}^{i}\}$ is an exchangeable family, we have that $\{\tX_{t}^{1}, ... , \tX_{t}^{N}\}$ is 
exchangeable at any time $t$. 

We prove that for any $N$, $\{\tX_{\cdot}^{1}, ... , \tX_{\cdot}^{N}\}$ is exchangeable in 
$C_{\mathbb{R}_+}[0, \infty)$, the continuous non-negative functions on $[0,\infty)$. 
In fact, for any Borel sets $A_{1}, ... , A_{N}\in C_{\mathbb{R}_+}[0, \infty)$, we need to prove that for any permutation 
$\sigma$, we have
\[\mathbb{P}\left\{\tX_{\cdot}^{1}\in A_{1}, ... , \tX_{\cdot}^{N}\in A_{N}\right\}=\mathbb{P}\left\{\tX_{\cdot}^{\sigma(1)}\in A_{1},
\dots, \tX_{\cdot}^{\sigma(N)}\in A_{N}\right\}.\]
It suffices to choose the following $A_{i}$'s: for any $n\in\mathbb{N}$, take $A_{i,1}, ... , A_{i,n}\in\mathcal{B}(\mathbb{R}_+)$ 
for $i=1,\dots,N$ with a time set $0=t_0<t_1< \dots < t_n$, and set
\[A_{i}=\left\{\tX_{t_{1}}^{i}\in A_{i, 1}, \tX_{t_{2}}^{i}-\tX_{t_{1}}^{i}\in A_{i,2}, ... , \tX_{t_{n}}^{i}-\tX_{t_{n-1}}^i
\in A_{i,n}\right\},\]
Therefore we have,
\begin{eqnarray*}
\lefteqn{\mathbb{P}\left(\tX_{\cdot}^{1}\in A_{1}, ... , \tX_{\cdot}^{N}\in A_{N}\right)\allowdisplaybreaks}\\
&=& \mathbb{P}\left(\bigcup_{i=1}^N \{\tX_{t_{1}}^i\in A_{i, 1}, \tX_{t_{2}}^i-\tX_{t_{1}}^i\in A_{i,2}, ... , 
\tX_{t_{n}}^i-\tX_{t_{n-1}}^i \in A_{i,n}\}\right) \\
&=&\prod_{j=1}^n \mathbb{P}\left(\bigcup_{i=1}^N\left\{\tX_{t_j}^i-\tX_{t_{j-1}}^i\in A_{i,j}\right\}\right) \allowdisplaybreaks\\
&=&\prod_{j=1}^n \mathbb{P}\left(\bigcup_{i=1}^N\left\{\tX_{t_j}^{\sigma(i)}-\tX_{t_{j-1}}^{\sigma(i)}\in A_{i,j}\right\}\right) 
\allowdisplaybreaks\\
&=&\mathbb{P}\left\{\tX_{\cdot}^{\sigma(1)}\in A_{1}, ... , \tX_{\cdot}^{\sigma(N)}\in A_{N}\right\},
\end{eqnarray*}
by the exchangeability of the increments of $\{\tX_{t}^{1}, ... , \tX_{t}^{N}\}$ at any time $t$. 
Hence $\{\tX_{\cdot}^{1}, ... , \tX_{\cdot}^{N}\}$ 
is exchangeable in $C_{\mathbb{R}}[0, \infty)$. As a consequence we have $\{X_{\cdot}^{1}, ... , X_{\cdot}^{N}\}$
is exchangeable in $C_{\mathbb{R}_+}[0, \infty)$ and for a fixed $t$, $\{X_t^{1}, ... , X_t^{N}\}$
is exchangeable in $\mathbb{R}_+$.
As the system (\ref{1}) is easily extended to an infinite particle system, by de Finetti's 
theorem, see for example, \cite{Aldous1985}, 
\[\bar{\nu}_{\cdot}=\lim\limits_{N\rightarrow+\infty}\frac{1}{N}\sum\limits_{i=1}^{N}\delta_{X_{\cdot}^{i}}\]
exists almost surely in $\mathcal{P}(C_{\mathbb{R}_+}[0,\infty))$.

We now need to show that the $\{\nu_t,t\in [0,\infty)\}$ is a continuous process in the space of probability measures.
We define a projection mapping
\[P_{t}: C_{\mathbb{R}}[0, \infty)\rightarrow\mathbb{R}\]
by setting, for any $Y_{\cdot}\in C_{\mathbb{R}}[0, \infty)$,
\[P_{t}(Y_{\cdot})=Y_{t}.\]
Then define $\bnu_{t}:=\bnu\circ P_{t}^{-1}\in\mathcal{P}(\mathbb{R})$. We first show that
\[\bnu_{t}=\lim_{N\rightarrow\infty}\frac{1}{N}\sum_{i=1}^N\delta_{X_{t}^{i}}.\]
To establish this we denote
\[\theta_N=\frac{1}{N}\sum_{i=1}^{N}\delta_{X_{t}^{i}}, \;
 \theta=\lim_{N\rightarrow\infty}\frac{1}{N}\sum_{i=1}^N\delta_{X_{t}^{i}},\]
where $\theta_N$ converges weakly to $\theta$ and $\theta$ exists in $\mathcal{P}(\mathbb{R})$ almost surely by the 
exchangeability of $\{X^{i}_{t}\}$ at any time $t$. For any $h\in C_{b}(\mathbb{R})$, the collection of all the bounded 
and continuous functions on $\mathbb{R}$, we have
\[
\int h(x) \theta(dx)= \lim\limits_{N\rightarrow\infty}\int h(x)\theta_N(dx).
\]
Define
\[\alpha_N=\frac{1}{N}\sum_{i=1}^{N}\delta_{X_{\cdot}^{i}}\in\mathcal{P}(C_{\mathbb{R}}[0,\infty)),
\; k=h\circ P_{t} \in C_b(C_{\mathbb{R}}[0,\infty)),\]
then
\[\theta_N=\alpha_N\circ P_{t}^{-1},\]
and $\alpha_N$ converges weakly to $\bnu$ in $\mathcal{P}(C_{\mathbb{R}}[0,\infty))$. Thus
\begin{align*}
\int h(x) \theta(dx)=& \lim_{N\rightarrow\infty}\int k\circ P_{t}^{-1}(x) (\alpha_N \circ P_{t}^{-1})(dx)\allowdisplaybreaks\\
=&\int k\circ P_{t}^{-1}(x)\bnu\circ P_{t}^{-1}(dx)\allowdisplaybreaks\\
=&\int h(x) \bnu_{t} (dx).
\end{align*}
Therefore $\bnu_{t}=\theta=\lim_{N\rightarrow+\infty}\frac{1}{N}\sum_{i=1}^N\delta_{X_{t}^{i}}\in
\mathcal{P}(\mathbb{R})$.
Next we want to show that $\bnu_{t}\in C_{\mathcal{P}(\mathbb{R})}[0,\infty)$. 
By definition it suffices prove that when $t_{n}\rightarrow t_{0}$, we have $\bnu_{t_{n}}\rightarrow \bnu_{t_{0}}$ 
weakly in
$\mathcal{P}(\mathbb{R})$, i.e., we want to show that for any open set $U\in\mathcal{B}(\mathbb{R})$, 
$\liminf_{n\rightarrow\infty} \bnu_{t_{n}}(U)\geq \bnu_{t_{0}}(U)$ [\cite{Ethier1986}, Theorem 3.3.1]. 
In fact we have, by continuity of $Y$ and Fatou's Lemma for sets, that
\begin{align*}
\bnu_{t_{0}}(U)=&\bnu\circ P_{t_{0}}^{-1}(U)=\bnu\left(\{Y_{\cdot}|Y_{t_{0}}\in U\}\right) \\
=& \bnu\left(\bigcup_{n=1}^{\infty}\bigcap_{k=n}^{\infty}\{Y_{\cdot}|Y_{t_{k}}\in U\}\right) \\
\leq&\lim_{n\rightarrow\infty}\inf_{k\geq n}\bnu\left(\{Y_{\cdot}|Y_{t_{k}}\in U\}\right)=
\liminf_{n\rightarrow\infty} \bnu_{t_{n}}(U).
\end{align*}
Therefore, the process $\{ \bnu_{t}: t\in [0,\infty)\}$ 
exists almost surely in $C_{\mathcal{P}(\mathbb{R})}[0,\infty)$.

The decomposition follows from the decomposition for $N$ companies. We then define $L_t = \bar{\nu}_t(\{0\})$ 
and $\nu_t$ to be $\bar{\nu}_t$ restricted to $(0,\infty)$.
\end{proof}

For a measure $\zeta_t$ and integrable function $\phi$ we write
\begin{equation}
\label{scalarproduct} \left\langle\phi,\zeta_t\right\rangle = \int \phi(x)\zeta_t(dx).
\end{equation}
Let $\bar{C}:=\{f: f\in C^2_b(0,\infty), f(0)=0,\lim_{x\to\infty}f(x)=0\}$.
Using the empirical measure (\ref{empiricalmeasure}) we define a
family of processes $F_t^{N,\phi}$ for $\phi\in \bar{C}$ by
\begin{equation}
F_{t}^{N,\phi} = \left\langle\phi,\bar{\nu}_{N,t}\right\rangle = \frac{1}{N}\sum_{i=1}^N
\phi(X_t^i) = \left\langle\phi,\nu_{N,t}\right\rangle
\end{equation}

As $X^i_t=0$ for $t>T_0^i$, and hence $\phi(X_t^i)=0$ for $t>T_0^i$, in order to apply It\^{o}'s formula to
$F^{N,\phi}_t$ we write $F^{N,\phi}_t=\dfrac{1}{N}\sum_{i=1}^{N}\phi(X_{t}^{i})
1_{\{t<T_{0}^{i}\}}$. Thus we have

\begin{eqnarray*}
F^{N,\phi}_t - F^{N,\phi}_0 &=& \dfrac{1}{N}\sum_{i=1}^{N}\int_{0}^{t}1_{\{s\leq T_{0}^{i}\}} 
\left(\phi^{\prime}(X_{s}^{i})dX_{s}^{i}+\dfrac{1}{2}\phi^{\prime\prime}(X_{s}^{i})
d[X^{i}_{s}] \right)\nonumber \allowdisplaybreaks\\
&=&\dfrac{1}{N}\sum_{i=1}^{N}\int_{0}^{t}1_{\{s\leq T_{0}^{i}\}}\left[\phi^{\prime}(X_{s}^{i}) \mu ds+
\phi^{\prime}(X_{s}^{i})\sqrt{1-\rho}dW_{s}^{i} 
+\phi^{\prime}(X_{s}^{i})\sqrt{\rho}dM_{s}+\dfrac{1}{2}\phi^{\prime\prime}(X_{s}^{i})ds\right]\nonumber \allowdisplaybreaks\\
&=&\int_{0}^{t}\dfrac{1}{N}\sum_{i=1}^{N}(\mu \phi^{\prime}(X_{s}^{i})+\dfrac{1}{2}
\phi^{\prime\prime}(X_{s}^{i}))1_{\{s<T_{0}^{i}\}}ds +\int_{0}^{t}\dfrac{1}{N}\sum_{i=1}^{N}\sqrt{1-\rho}
\phi^{\prime}(X_{s}^{i})1_{\{s< T_{0}^{i}\}}dW_{s}^{i}\allowdisplaybreaks\nonumber\\
&&\qquad +\int_{0}^{t}\dfrac{1}{N}\sum_{i=1}^{N}\sqrt{\rho}\phi^{\prime}(X_{s}^{i})1_{\{s< 
T_{0}^{i}\}}dM_{s}
\end{eqnarray*}

If we define the second order linear operator $\mathcal{A}$ by
\begin{displaymath}
\mathcal{A} = \mu\frac{\partial}{\partial x} + \frac{1}{2} \frac{\partial^2}{\partial x^2},
\end{displaymath}
we have
\begin{eqnarray}
F_t^{N,\phi} &=& F_0^{N,\phi} + \int_0^t \left\langle \mathcal{A}\phi,\nu_{N,s} \right\rangle ds 
+ \int_0^t \left\langle\sqrt{\rho}\phi^{\prime},\nu_{N,s} \right\rangle dM_s \nonumber \\ 
&&\qquad + \int_0^t \frac{1}{N}\sum_{i=1}^N \phi^{\prime}(X_s^i)\,\sqrt{1-\rho}\,dW_s^i.\label{3.6}
\end{eqnarray}
We now pass to the limit by letting $N\rightarrow\infty$.

In order to determine what happens we first focus on the idiosyncratic term in
(\ref{3.6}) 
\begin{equation} \label{idiosyncraticterm}
I^{\phi}_{t,N} = \int_0^t \frac{1}{N}\sum_{i=1}^N \sqrt{1-\rho} \phi^{\prime}(X_s^i)\,dW_s^i.
\end{equation}
As $\phi'$ is bounded $I^{\phi}_{t,N}$ is a martingale and, by the independence of the $W_t^i$ it has quadratic variation
\begin{displaymath}
[I_N^{\phi}]_t = \int_0^t \frac{1}{N^2}\sum_{i=1}^N (1-\rho)\,(\phi^{\prime}(X_s^i))^2\,ds.
\end{displaymath}
As $\phi\in \bar{C}$ there exists a constant $K_{\phi}$ such that $|\phi'|\leq K_{\phi}$. Thus
\[ \lim_{N\rightarrow\infty}\frac{1}{N}\sum_{i=1}^N
\int_0^t (1-\rho) \mid\phi^{\prime}(X_s^i)\mid^2\,ds \leq K_{\phi}^2 t, \]
and hence we have for any such $\phi$
\[ \lim_{N\rightarrow\infty}\frac{1}{N^2}\sum_{i=1}^N
\int_0^t (1-\rho) \mid\phi^{\prime}(X_s^i)\,\mid^2 \,ds \leq \lim_{N\rightarrow\infty}
\frac{1}{N} K_{\phi}^2t = 0, \;\;\forall t\in [0,T]. \]
Thus the random term due to the idiosyncratic component of the asset values has
become deterministic in the infinite dimensional limit and must vanish almost surely. 

We also note that as $\phi',\phi''$ are bounded and $\nu_{N,s}$ is a probability measure, we can apply the dominated 
convergence theorem to take the limit under the integrals in the other terms in (\ref{3.6}). 
We summarize in the following

\begin{thm}
The sequence of empirical measures $\nu_{N,t}$ on $(0,\infty)$ satisfies for all $\phi\in\bar{C}$,
\begin{displaymath}
F_t^{N,\phi} \rightarrow F_t^{\phi} = \left\langle \phi,\nu_t \right\rangle \quad
\textrm{as } N\rightarrow\infty,\;\;a.s.  \end{displaymath} 
The evolution of the limit empirical measure in the weak sense is given by
\begin{equation}
\left\langle\phi,\nu_t\right\rangle = \left\langle\phi,\nu_0\right\rangle +
\int_0^t\left\langle \mathcal{A}\phi,\nu_s \right\rangle ds +
\int_0^t\left\langle \srho\phi^{\prime},\nu_s\right\rangle
dM_s, \;\;\forall \phi\in\bar{C}. \label{martingaleproblem}
\end{equation}
\end{thm}

\subsection{The boundary condition}

The behaviour of $\nu_{t}$, the limit empirical measure on $(0,\infty)$, at the boundary zero is given in the 
following theorem:

\begin{thm} 
We have
\[ \lim_{\varepsilon\downarrow 0}\frac{\nu_{t}((0,\varepsilon))}{\varepsilon}=0, \; \text{ a.s.}.\]
\label{Th3}
\end{thm}

\begin{proof}
By the definition of $\nu_{t}$, properties of weak convergence and an application of Fatou's Lemma, we have
\begin{align}
\mathbb{E}[\nu_{t}((0,\varepsilon))]\leq &\mathbb{E}\left[\liminf_{N\rightarrow+\infty}
\frac{1}{N}\sum\limits_{i=1}^{N}1_{\{0<X_{t}^{i}<\varepsilon\}}\right]\allowdisplaybreaks\nonumber\\
\leq&\liminf_{N\rightarrow+\infty}
\frac{1}{N}\sum\limits_{i=1}^{N}\mathbb{P}\left\{X_{t}^{i}\leq \varepsilon, \inf\limits_{0\leq s\leq t}X_{s}^{i}>0\right\}.
\label{2}
\end{align}
For $t<T_{0}^{i}$, integrating the system (\ref{1}) from time $0$ to $t$, we have:
\[X_{t}^{i}=x^{i}+\mu t+\sqrt{1-\rho}W_{t}^{i}+\sqrt{\rho}M_{t}.\]
Since we know that
\[\sqrt{1-\rho}W_{t}^{i}+\sqrt{\rho}M_{t}\stackrel{d}=B_{t},\]
where $B_{t}$ is a standard Brownian motion on the same probability space, we have
\begin{align}
\mathbb{P}\left\{X_{t}^{i} < \varepsilon, \inf_{0\leq s\leq t}X_{s}^{i}>0\right\}=& 
\mathbb{P}\left\{x^{i}+\mu t+B_{t}<\varepsilon,
\inf_{0\leq s\leq t}(x^{i}+\mu s+B_{s})>0\right\}\allowdisplaybreaks\nonumber\\
=&\mathbb{P}^{x^{i}}\left\{\mu t+B_{t}<\varepsilon\right\}-\mathbb{P}^{x^{i}}\left\{\mu t+B_{t} <\varepsilon, 
\inf_{0\leq s\leq t}(\mu s+B_{s})\leq 0\right\}\allowdisplaybreaks\nonumber\\
=&\int_{-\infty}^{\varepsilon}\frac{1}{\sqrt{2\pi t}}e^{-(z-\mu t-x^{i})^{2}/2t}dz 
-\int_{-\infty}^{\varepsilon}\frac{1}{\sqrt{2\pi t}}e^{\mu (z-x^{i})-\mu^{2}t/2-(|z|+x^{i})^{2}/2t}dz 
\nonumber\allowdisplaybreaks\\
=&\frac{1}{\sqrt{2\pi t}}\int_{0}^{\varepsilon}\left(e^{-(z-\mu t-x^{i})^{2}/2t}-
e^{\mu (z-x^{i})-\mu^{2}t/2-(z+x^{i})^{2}/2t}\right)dz\nonumber\allowdisplaybreaks\\
=&\frac{1}{\sqrt{2\pi t}}\int_{0}^{\varepsilon}e^{-(z-\mu t-x^{i})^{2}/2t} 
\left(1-e^{-\frac{2zx^{i}}{t}}\right)dz\allowdisplaybreaks\nonumber\\
\leq& \frac{1}{\sqrt{2\pi t}}\left(1-e^{-\frac{2\varepsilon x^{i}}{t}}\right)\int_{0}^{\varepsilon}
e^{-(z-\mu t-x^{i})^{2}/2t}dz\allowdisplaybreaks\nonumber\\
\leq& \frac{1}{\sqrt{2\pi t}}\frac{2\varepsilon x^{i}}{t}\int_{0}^{\varepsilon}e^{-(z-\mu t-x^{i})^{2}/2t}dz.
\label{3}
\end{align}

Assume $\varepsilon<\frac{1}{2}C_{B}$. Since we have $x^{i}\geq C_{B}$, if $t<\frac{C_{B}-\varepsilon}{|\mu|}$, then
$|z-\mu t-x^{i}|>0,\; \forall 0<z<\varepsilon$
and there exists $C^{1}_{T}>0$ only depending on $T$ such that
\[\frac{1}{t^{\frac{3}{2}}}e^{-(z-\mu t-x^{i})^{2}/2t}\leq C^{1}_{T}, \;\; \forall t<\frac{C_{B}-\varepsilon}{|\mu|}.\]
If $t\geq\frac{C_{B}-\varepsilon}{|\mu|}$, then
\[\frac{1}{t^{\frac{3}{2}}}e^{-(z-\mu t-x^{i})^{2}/2t}\leq \frac{1}{\left(\frac{C_{B}-\varepsilon}{|\mu|}\right)^{\frac{3}{2}}}
\leq \frac{1}{\left(\frac{C_{B}}{2|\mu|}\right)^{\frac{3}{2}}}.\] 

Letting $C_{T}^{\prime}:=\max\left\{C_{T}^{1}, \frac{1}{\left(\frac{C_{B}}{2|\mu|}\right)^{\frac{3}{2}}}\right\}$, 
(\ref{3}) becomes
\begin{equation}
\mathbb{P}\left\{X_{t}^{i}< \varepsilon, \inf_{0\leq s\leq t}X_{s}^{i}>0\right\}\leq \frac{2}{\sqrt{2\pi}}
\varepsilon x^{i}\varepsilon C_{T}^{\prime}:= x^{i}C_{T}\varepsilon^{2},
\label{*4}
\end{equation}
where $C_{T}$ is a positive constant only depending on $T$. Thus by (\ref{2}) and (\ref{*4}) we have
\begin{equation}
\frac{\mathbb{E}[\nu_{t}((0,\varepsilon))]}{\varepsilon}\leq C_{T}
\varepsilon (\liminf_{N\rightarrow+\infty}\frac{1}{N}\sum_{i=1}^{N} x^{i}).
\label{5}
\end{equation}
Since $\{X_{0}^{1}, ... , X_{0}^{N}\}$ is an exchangeable family of integrable random variables, 
$\lim_{N\rightarrow+\infty}\frac{1}{N} \sum_{i=1}^{N} X^{i}_{0}$ exists and is finite almost surely. 
Let $K=\lim_{N\rightarrow\infty}\frac{1}{N}\sum_{i=1}^{N} x^{i}$. Now
(\ref{5}) becomes
\[\mathbb{E}\left[\frac{\nu_{t}((0,\varepsilon))}{\varepsilon}\right]\leq KC_{T}\varepsilon.\]
By Markov's inequality, for any $\lambda>0$ we have
\[\mathbb{P}\left\{\frac{\nu_{t}((0,\varepsilon))}{\varepsilon}>\lambda\right\}\leq \frac{KC_{T}\varepsilon}{\lambda},\]
therefore, for the subsequence $\varepsilon=\frac{1}{n^{2}}$,
\[\mathbb{P}\left\{\frac{\nu_{t}((0,\frac{1}{n^{2}}))}{\frac{1}{n^{2}}}>\lambda\right\}\leq \frac{KC_{T}}{\lambda n^{2}}.\]
Thus by the first Borel-Cantelli Lemma, as $\lambda>0$ is arbitrary and also $\frac{\nu_{t}((0,\frac{1}{n^{2}}))}{\frac{1}{n^{2}}}\geq 0$, 
we must have
\[\limsup_{n\rightarrow\infty}\frac{\nu_{t}((0,\frac{1}{n^{2}}))}{\frac{1}{n^{2}}}=0, \; \text{ a.s.}.\]
Now for any $\varepsilon>0$, there exists a $n$ such that $\frac{1}{(n+1)^{2}}\leq\varepsilon\leq\frac{1}{n^{2}}$ and hence
\[\limsup_{\varepsilon\downarrow 0}\frac{\nu_{t}((0,\varepsilon))}{\varepsilon}\leq 
\limsup_{n\rightarrow\infty}\frac{\nu_{t}((0,\frac{1}{n^{2}}))}{\frac{1}{(n+1)^{2}}}=
\limsup_{n\rightarrow\infty}\frac{\nu_{t}((0,\frac{1}{n^{2}}))}{\frac{1}{n^{2}}}\frac{\frac{1}{n^{2}}}
{\frac{1}{(n+1)^{2}}}=0 , \; \text{ a.s.}.\]
Since $\frac{\nu_{t}((0,\varepsilon))}{\varepsilon}\geq 0$, therefore
\[v(t,0):=\lim_{\varepsilon\downarrow 0}\frac{\nu_{t}((0,\varepsilon))}{\varepsilon}=0, \;\text{ a.s.}.\]
\end{proof}

Therefore, if there is a density for the empirical measure, it will satisfy a Dirichlet boundary condition.

Next we give an estimate on $\mathbb{E}[(\nu_{t}((0,\varepsilon)))^{2}]$ which will be needed later.
In order to do this we require an estimate for the distribution of the first passage times of two correlated Brownian motions, 
and the Brownian motions themselves. We use a transformation to independence and the formula derived in \cite{Iyengar1985}.

\begin{lem}
Let $B^{1}_{t}$ and $B^{2}_{t}$ be two correlated Brownian motions with constant correlation $|\varrho|<1$, $B^{1}_{0}=a_{1}>0$,
$B^{2}_{0}=a_{2}>0$ and law $\bp_B$. Then there exists
 $\varepsilon_{0}=\frac{1}{3}\sqrt{\frac{1-\varrho}{2}}\sqrt{\frac{a_{1}^{2}+a_{2}^{2}-
2\varrho a_{1}a_{2}}{1-\varrho^{2}}}$ such that for all $\varepsilon<\varepsilon_{0}$,
\[\mathbb{P}_B\left\{0<B^{1}_{t}<\varepsilon,\; \inf_{0\leq s \leq t}B^{1}_{s}>0, \; 0<B^{2}_{t}<\varepsilon, 
\; \inf_{0\leq s \leq t}B^{2}_{s}>0\right\}\leq C_{T} \varepsilon^{2+\frac{\pi}{\alpha}},\]
where $C_{T}=2^{1-\frac{\pi}{\alpha}}\left(\sqrt{\frac{a_{1}^{2}+a_{2}^{2}-2\varrho a_{1}a_{2}}{1-\varrho^{2}}} 
\right)^{\frac{\pi}{\alpha}}K_{T}\left(\sqrt{\frac{2}{1-\varrho}}\right)^{2+\frac{\pi}{\alpha}}$ 
and $K_{T}$ is a constant only depending on $T$; and
\begin{equation}
\alpha=\left\{
\begin{array}{lll}
&\pi+\tan^{-1}\left(-\frac{\sqrt{1-\varrho^{2}}}{\varrho}\right), &\varrho>0,\\
&\frac{\pi}{2}, &\varrho=0,\\
&\tan^{-1}\left(-\frac{\sqrt{1-\varrho^{2}}}{\varrho}\right),&\varrho<0.
\end{array}
\right.
\label{7*}
\end{equation}
Therefore, if $\varrho\geq 0$, we have $\frac{\pi}{2}\leq\alpha<\pi$ and $3<2+\frac{\pi}{\alpha}\leq 4$.
\label{Lemma1*}
\end{lem}

\begin{proof}
We begin by making a transformation to obtain 
a two-dimensional Brownian motion with independent components. We follow the setup and statements in \cite{Metzler2010}. 
Let $B_{t}=(B^{1}_{t}, B^{2}_{t})$ and consider the process $Z=\sigma^{-1}B$, where
\[\sigma=
\begin{bmatrix}
  \sqrt{1-\varrho^{2}} & \varrho \\
  0 & 1
 \end{bmatrix}.
\]
We know that $Z$ has independent components. It is easily seen that the horizontal axis is invariant under the transformation 
$T: \mathbb{R}^{2}\rightarrow\mathbb{R}^{2}$ defined by $T(\textbf{x})=\sigma^{-1}\textbf{x}$, while the vertical axis is mapped 
to the line $z_{1}=-\frac{\varrho}{\sqrt{1-\varrho^{2}}}z_{2}$. 

Now the time that the first Brownian motion $B^{1}$ hits zero is transformed to the time $\tau_{1}$ which is the first passage time 
of $Z_{t}$ to the horizontal axis; and the time that the second Brownian motion $B^{2}$ hits zero is transformed to the time 
$\tau_{2}$ which is the first passage time of $Z_{t}$ to the line $z_{2}=z_{1}\tan\alpha$, where $0<\alpha<\pi$ is given in 
(\ref{7*}). Moreover, in polar coordinates $Z_{t}=(R_t,\Theta_t)$ starts at the point $z_{0}$ given by
\[r_{0}=\sqrt{\frac{a_{1}^{2}+a_{2}^{2}-2\varrho a_{1}a_{2}}{1-\varrho^{2}}};\]
and
\begin{equation*}
\theta_{0}=\left\{
\begin{array}{lll}
&\pi+\tan^{-1}\left(\frac{a_{2}\sqrt{1-\varrho^{2}}}{a_{1}-\varrho a_{2}}\right), &a_{1}<\varrho a_{2},\\
&\frac{\pi}{2}, &a_{1}=\varrho a_{2},\\
&\tan^{-1}\left(\frac{a_{2}\sqrt{1-\varrho^{2}}}{a_{1}-\varrho a_{2}}\right),&a_{1}>\varrho a_{2}.
\end{array}
\right.
\label{7}
\end{equation*}
It is easily verified that $0<\theta_{0}<\alpha$. We denote by $\tau=\min(\tau_{1}, \tau_{2})$ the first exit time of $Z$ from 
the wedge
\[C_{\alpha}=\{(r\cos\theta, r\sin\theta): r>0, 0<\theta<\alpha\}\subset\mathbb{R}^{2}.\]  

 If $z=(r\cos\theta, r\sin\theta)$ is a point in $C_{\alpha}$ we have, by \cite{Iyengar1985},
\begin{equation}
\mathbb{P}_B^{z_{0}}\{\tau>t, Z_t\in dz\}=\frac{2r}{t\alpha}e^{-(r^{2}+r_{0}^{2})/2t}\sum_{n=1}^{\infty}\sin\frac{n\pi\theta} 
{\alpha}\sin\frac{n\pi\theta_{0}}{\alpha}I_{n\pi/\alpha}\left(\frac{rr_{0}}{t}\right)drd\theta,
\label{8}
\end{equation}
where $I_{v}$ denotes the modified Bessel function of the first kind of order $v$ 

Using this transformation and the formula (\ref{8}) we have
\begin{eqnarray}
\lefteqn{\mathbb{P}_B\left\{0<B^{1}_{t}<\varepsilon,\; \inf_{0\leq s \leq t}B^{1}_{s}>0, \; 0<B^{2}_{t}< \varepsilon, 
\; \inf_{0\leq s \leq t}B^{2}_{s}>0\right\}} \nonumber \\
&\leq&\mathbb{P}_B\left\{\tau>t, 0<\Theta_t<\alpha, 0<R_t< \sqrt{\frac{2}{1-\varrho}}\varepsilon\right\}\nonumber\\
&=&\int_{0}^{\sqrt{\frac{2}{1-\varrho}}\varepsilon}\int_{0}^{\alpha}\frac{2r}{t\alpha}e^{-(r^{2}+r_{0}^{2})/2t}
\sum_{n=1}^{\infty}\sin\frac{n\pi\theta}{\alpha}\sin\frac{n\pi\theta_{0}}{\alpha}I_{n\pi/\alpha}\left(\frac{rr_{0}}{t}\right)
drd\theta \nonumber\\
&\leq&\int_{0}^{\sqrt{\frac{2}{1-\varrho}}\varepsilon}\frac{2r}{t\alpha}e^{-(r^{2}+r_{0}^{2})/2t}\int_{0}^{\alpha}
\sum_{n=1}^{\infty}I_{n\pi/\alpha}\left(\frac{rr_{0}}{t}\right)drd\theta. \label{eq:8a}
\end{eqnarray}

By the definition of the modified Bessel function, we have
\begin{eqnarray*}
I_{n\pi/\alpha}\left(\frac{rr_{0}}{t}\right)&=&\sum\limits_{m=0}^{\infty}\frac{1}{m!\Gamma(m+\frac{n\pi}{\alpha}+1)}
\left(\frac{rr_{0}}{2t}\right)^{2m+\frac{n\pi}{\alpha}}\allowdisplaybreaks\\
&\leq&\sum_{m=0}^{\infty}\frac{1}{(m!)^{2}\left[\frac{n\pi}{\alpha}\right]!}\left(\frac{rr_{0}}{2t}\right)^{2m+
\frac{n\pi}{\alpha}}\allowdisplaybreaks\\
&=&\frac{1}{\left[\frac{n\pi}{\alpha}\right]!}\left(\frac{rr_{0}}{2t}\right)^{\frac{n\pi}{\alpha}}
\sum_{m=0}^{\infty}\frac{1}{(m!)^{2}}\left(\frac{rr_{0}}{2t}\right)^{2m}\allowdisplaybreaks\\
&\leq&\frac{1}{\left[\frac{n\pi}{\alpha}\right]!}\left(\frac{rr_{0}}{2t}\right)^{\frac{n\pi}{\alpha}}
\left[\sum_{m=0}^{\infty}\frac{1}{m!}\left(\frac{rr_{0}}{2t}\right)^{m}\right]^{2}\allowdisplaybreaks\\
&=&e^{rr_{0}/t}\frac{1}{\left[\frac{n\pi}{\alpha}\right]!}\left(\frac{rr_{0}}{2t}\right)^{\frac{n\pi}{\alpha}},
\end{eqnarray*}
where $[x]$ denotes the integer part of $x$. Using this in (\ref{eq:8a}) we have
\begin{eqnarray*}
\lefteqn{\mathbb{P}_B\left\{0<B^{1}_{t}<\varepsilon,\; \inf_{0\leq s \leq t}B^{1}_{s}>0, \; 0<B^{2}_{t}<\varepsilon, 
\; \inf_{0\leq s \leq t}B^{2}_{s}>0\right\}\allowdisplaybreaks}\\
&\leq&\int_{0}^{\sqrt{\frac{2}{1-\varrho}}\varepsilon}\frac{2r}{t\alpha}e^{-(r^{2}+r_{0}^{2})/2t}\int_{0}^{\alpha}e^{rr_{0}/t} 
\sum_{n=1}^{\infty}\frac{1}{\left[\frac{n\pi}{\alpha}\right]!}\left(\frac{rr_{0}}{2t}\right)^{\frac{n\pi}{\alpha}}drd\theta
\allowdisplaybreaks\\
&\leq&\int_{0}^{\sqrt{\frac{2}{1-\varrho}}\varepsilon}\frac{2r}{t}e^{-(r^{2}+r_{0}^{2})/2t}e^{rr_{0}/t}
(\frac{rr_{0}}{2t})^{\frac{\pi}{\alpha}}e^{rr_{0}/2t}dr\allowdisplaybreaks\\
&=& 2^{1-\frac{\pi}{\alpha}}r_{0}^{\frac{\pi}{\alpha}}\int_{0}^{\sqrt{\frac{2}{1-\varrho}}\varepsilon}r^{1+\frac{\pi}{\alpha}}
\frac{1}{t^{1+\frac{\pi}{\alpha}}}e^{-\frac{r^{2}+r_{0}^{2}-3rr_{0}}{2t}}dr.
\end{eqnarray*}

If we choose $\varepsilon_{0}=\frac{r_{0}\sqrt{\frac{1-\varrho}{2}}}{3}$, then for any $\varepsilon<\varepsilon_{0}$ we have
$r^{2}+r_{0}^{2}-3rr_{0}>0$. Therefore we can find a constant $K_{T}$ only depending on $T$ such that
\[\frac{1}{t^{1+\frac{\pi}{\alpha}}}e^{-\frac{r^{2}+r_{0}^{2}-3rr_{0}}{2t}}\leq K_{T}.\]
Thus
\begin{eqnarray*}
\lefteqn{\mathbb{P}_B\left\{0<B^{1}_{t}<\varepsilon,\; \inf_{0\leq s \leq t}B^{1}_{s}>0, \; 0<B^{2}_{t}<\varepsilon, 
\; \inf_{0\leq s \leq t}B^{2}_{s}>0\right\}\allowdisplaybreaks}\\
&\leq&2^{1-\frac{\pi}{\alpha}}r_{0}^{\frac{\pi}{\alpha}}\int_{0}^{\sqrt{\frac{2}{1-\varrho}}\varepsilon}r^{1+\frac{\pi}{\alpha}}
K_{T}dr\allowdisplaybreaks\\
&\leq&2^{1-\frac{\pi}{\alpha}}r_{0}^{\frac{\pi}{\alpha}}K_{T}\left(\sqrt{\frac{2}{1-\varrho}}\varepsilon\right)^{1+
\frac{\pi}{\alpha}} \sqrt{\frac{2}{1-\varrho}}\varepsilon 
= C_{T}\varepsilon^{2+\frac{\pi}{\alpha}},
\end{eqnarray*}
where $C_{T}=2^{1-\frac{\pi}{\alpha}}r_{0}^{\frac{\pi}{\alpha}}K_{T}\left(\sqrt{\frac{2}{1-\varrho}}\right)^{2+\frac{\pi}{\alpha}}$ 
is a constant only depending on $\varrho, a_{1}, a_{2}$ and $T$.

Moreover, it is obvious that $0<\alpha<\pi$ and $\frac{\pi}{2}\leq\alpha<\pi$ if $\varrho\geq 0$. In the latter case we have
$3<2+\frac{\pi}{\alpha}\leq 4$.
\end{proof}

\begin{lem}
There exists $\tilde{\varepsilon}_{0}>0$ only depending on $\rho$ and the lower bound $C_{B}$ for the $\{X_0^i\}$, such that for any 
$\eta>0$, for all $\varepsilon<\tilde{\varepsilon}_{0}$ we have
\[
\mathbb{E}[(\nu_{t}((0,\varepsilon)))^{2}]\leq K_{T}\varepsilon^{2+\pi/\alpha-\eta},
\]
where $K_{T}$ is a positive constant depending on $T$ and $\alpha$ is given in (\ref{7}).
\label{Lemma1}
\end{lem}

\begin{proof}
By definition of $\nu_{t}$, properties of weak convergence and Fatou's Lemma
\begin{align}
\mathbb{E}[(\nu_{t}((0,\varepsilon)))^{2}]\leq &\mathbb{E}\left[\liminf_{N\rightarrow\infty}\frac{1}{N}
\sum_{i=1}^{N}1_{\{0<X_{t}^{i}< \varepsilon, \inf_{0\leq s\leq t}X_{s}^{i}>0\}} \liminf_{M\rightarrow\infty}
\frac{1}{M}\sum_{j=1}^{M}1_{\{0<X_{t}^{j}< \varepsilon, \inf_{0\leq s\leq t}X_{s}^{j}>0\}}\right]
\allowdisplaybreaks\nonumber\\
\leq&\liminf_{N\rightarrow\infty, M\rightarrow\infty }\frac{1}{NM}\sum_{i=1}^{N}\sum_{j=1}^{M}\mathbb{E}\left[1_{\{0<X_{t}^{i}
< \varepsilon, \inf\limits_{0\leq s\leq t}X_{s}^{i}>0, \;0<X_{t}^{j}< \varepsilon, 
\inf_{0\leq s\leq t}X_{s}^{j}>0\}}\right]\allowdisplaybreaks\nonumber\\
=&\liminf_{N\rightarrow\infty, M\rightarrow\infty }\frac{1}{NM}\sum_{i=1}^{N}\sum_{j\neq i, j=1}^{M}\mathbb{E} 
\left[1_{\{0<X_{t}^{i}< \varepsilon, \inf\limits_{0\leq s\leq t}X_{s}^{i}>0, \; 0<X_{t}^{j}< \varepsilon, 
\inf_{0\leq s\leq t}X_{s}^{j}>0\}}\right].
\label{*9}
\end{align}
Since neither of the firms $i$ or $j$ has defaulted by time $t$, we have
\[X_{t}^{i}=x^{i}+\mu t +\sqrt{1-\rho}W_{t}^{i}+\sqrt{\rho}M_{t}\stackrel{d}=x^{i}+\mu t+B^{1}_{t}; \]
\[ X_{t}^{j}=x^{j}+\mu t +\sqrt{1-\rho}W_{t}^{j}+\sqrt{\rho}M_{t}\stackrel{d}=x^{j}+\mu t+B^{2}_{t},\]
where $B^{1}_{t}$ and $B^{2}_{t}$ are correlated Brownian motions with correlation $\rho$.

We use the Girsanov theorem (e.g. \cite{Revuz2005}) to change the measure and set
\[ Z_{t}(\mu) =\exp\left(-\frac{\mu}{1+\rho}\left(B^1_t+B^2_t +\mu t\right)\right),\]
which is easily seen to be a true martingale by Novikov's condition.
We write $\tilde{\mathbb{P}}$ for the probability measure on $\mathcal{F}_{T}$ given by
\begin{equation}
\tilde{\mathbb{P}}(A):=\mathbb{E}[1_{A}Z_{T}(\mu)]; \;A\in\mathcal{F}_{T},
\label{3.31}
\end{equation}
and $\tilde{\be}$ for expectation with respect to $\tilde{\bp}$.
Thus for each fixed $T\in[0, \infty)$, the process
\[\{(\tilde{B}^{1}_{t}, \tilde{B}^{2}_{t}):=(B^{1}_{t}+\mu t, B^{2}_{t}+\mu t), \mathcal{F}_{t}, 
0\leq t\leq T\}\]
is a two-dimensional Brownian motion on $(\Omega, \mathcal{F}_{T}, \tilde{\mathbb{P}})$, where $\tilde{B}^{1}_{t}$ 
and $\tilde{B}^{2}_{t}$ have correlation $\rho$.

We now calculate the term $\mathbb{E}\left[1_{\{0<X_{t}^{i}< \varepsilon, \inf_{0\leq s\leq t}X_{s}^{i}>0, 
\; 0<X_{t}^{j}< \varepsilon, \inf_{0\leq s\leq t}X_{s}^{j}>0\}}\right]$ in (\ref{*9}). We have
\begin{eqnarray*}
\lefteqn{\mathbb{E}\left[1_{\{0<X_{t}^{i}< \varepsilon, \inf_{0\leq s\leq t}X_{s}^{i}>0, 
\; 0<X_{t}^{j}< \varepsilon, \inf_{0\leq s\leq t}X_{s}^{j}>0\}}\right]\allowdisplaybreaks}\\
&=&\tilde{\mathbb{E}}\left[1_{\{0<X_{t}^{i}< \varepsilon, \inf_{0\leq s\leq t}X_{s}^{i}>0, \; 
0<X_{t}^{j}< \varepsilon, \inf_{0\leq s\leq t}X_{s}^{j}>0\}}\frac{1}{Z_{T}(\mu)}\right]\allowdisplaybreaks\\
&\leq&\left\{\tilde{\be}\left[1_{\{0<X_{t}^{i}< \varepsilon, \inf_{0\leq s\leq t}X_{s}^{i}>0, 
\; 0<X_{t}^{j}< \varepsilon, \inf_{0\leq s\leq t}X_{s}^{j}>0\}}\right]\right\}^{1/a} 
\cdot\left\{\tilde{\be}\left[\left(\frac{1}{Z_{T}(\mu)}\right)^{b}\right]\right\}^{1/b}
\allowdisplaybreaks\\
&=&J_{1}\cdot J_{2},
\end{eqnarray*}
where
\[J_{1}=\left\{\tilde{\mathbb{E}}\left[1_{\{0<X_{t}^{i}< \varepsilon, \inf_{0\leq s\leq t}X_{s}^{i}>0, 
\; 0<X_{t}^{j}< \varepsilon, \inf_{0\leq s\leq t}X_{s}^{j}>0\}}\right]\right\}^{1/a},\]
\[J_{2}=\left\{\tilde{\mathbb{E}}\left[\left(\frac{1}{Z_{T}(\mu)}\right)^{b}\right]\right\}^{1/b},\]
\[1/a+1/b=1, a>1, b>1,\]
by H\"olders inequality.

For $J_{1}$, 
we have
\begin{align*}
J_{1}=&\left(\mathbb{\tilde{P}}\left\{0<X_{t}^{i}< \varepsilon, \inf_{0\leq s\leq t}X_{s}^{i}>0, \; 0<X_{t}^{j}< \varepsilon,
\inf_{0\leq s\leq t}X_{s}^{j}>0\right\}\right)^{1/a}\allowdisplaybreaks\\
=&\left(\mathbb{\tilde{P}}\left\{0<x^{i}+\tilde{B}_{t}^{1}< \varepsilon, \inf_{0\leq s\leq t}x^{i}+\tilde{B}_{s}^{1}>0, \; 
0<x^{j}+\tilde{B}_{t}^{2}< \varepsilon, \inf_{0\leq s\leq t}x^{j}+\tilde{B}_{s}^{2}>0\right\}\right)^{1/a}.
\end{align*}
By Lemma \ref{Lemma1*} with $\varrho=\rho, a_1=x^i, a_2=x^j$ we know that there exists 
$\varepsilon_{0}=\frac{1}{3}\sqrt{\frac{1-\rho}{2}}\frac{\sqrt{(x^{i})^{2}+(x^{j})^{2}-2\rho x^{i}x^{j}}}{\sqrt{1-\rho^{2}}}$ and 
$\alpha$ as in (\ref{7}), such that for all $\varepsilon<\varepsilon_{0}$ we have
\[\mathbb{\tilde{P}}\left\{0<x^{i}+\tilde{B}_{t}^{1}< \varepsilon, \inf_{0\leq s\leq t}x^{i}+\tilde{B}_{s}^{1}>0, \;
0<x^{j}+\tilde{B}_{t}^{2}\leq <, \inf_{0\leq s\leq t}x^{j}+\tilde{B}_{s}^{2}>0\right\}\leq 
C_{T}\varepsilon^{2+\frac{\pi}{\alpha}}. \]
As $x^{i}\geq C_{B}$ and $x^{j}\geq C_{B}$, we have
\[\sqrt{(x^{i})^{2}+(x^{j})^{2}-2\rho x^{i}x^{j}}\geq \sqrt{2(1-\rho)}C_{B}.\]
Thus we can choose a new $\tilde{\varepsilon}_{0}:=\frac{1}{3}\sqrt{\frac{1-\rho}{1+\rho}}
C_{B} \leq C_B$, such that for all $\varepsilon<\tilde{\varepsilon}_{0}$ we have, for all $i,j$,
\[J_{1}\leq C_{T}\varepsilon^{\frac{2+\frac{\pi}{\alpha}}{a}}.\]

For $J_{2}$ we have
\begin{align*}
J_{2}=&\left\{\tilde{\mathbb{E}}\left[\left(\frac{1}{Z_{T}(\mu)}\right)^{b}\right]\right\}^{1/b}\allowdisplaybreaks\\
=&\left\{\tilde{\mathbb{E}}\left[\exp\left(\frac{b\mu}{1+\rho}\left(B^1_T+B^2_T +\mu T\right)\right) 
\right]\right\}^{1/b}\allowdisplaybreaks\\
=&\left\{\tilde{\mathbb{E}}\left[\exp\left(\frac{b\mu}{1+\rho}\left(\tilde{B}^1_T+\tilde{B}^2_T -\mu T\right)\right) 
\right]\right\}^{1/b}\allowdisplaybreaks\\
=&\exp\left(-\frac{\mu^2 T}{1+\rho}\right) \left\{ \tilde{\mathbb{E}} 
\left[\exp\left(\frac{b\mu}{1+\rho}\left(\tilde{B}^1_T+\tilde{B}^2_T\right)\right) 
\right]\right\}^{1/b}\allowdisplaybreaks\\
=&\exp\left(\frac{(b-1)\mu^{2}T}{1+\rho}\right):=J_{T}<\infty, \;\; \forall b, i,j\in\mathbb{N}.
\end{align*}

Now we have
\[
\mathbb{E}[\nu_{t}((0,\varepsilon))^{2}]\leq J_{1}\cdot J_{2}\leq C_{T}J_{T}\varepsilon^{\frac{2+\frac{\pi}{\alpha}}{a}}, 
\; \forall \varepsilon<\tilde{\varepsilon}_{0}.
\]
Now for any $0<\eta< \frac{\pi}{\alpha}-1$ we can choose $1<a=(2+\frac{\pi}{\alpha})/(2+\frac{\pi}{\alpha}-\eta) 
<(2+\frac{\pi}{\alpha})/3$ and hence
\[\mathbb{E}[\nu_{t}((0,\varepsilon))^{2}]\leq K_{T}\varepsilon^{2+\pi/\alpha-\eta}, \; 
\forall \varepsilon<\tilde{\varepsilon}_{0},\]
where $K_{T}$ is a positive constant only depending on $T$.
\end{proof}

We will write $\beta=\pi/\alpha-\eta-1>0$ so that $2+\pi/\alpha-\eta=3+\beta$.

\subsection{The existence and uniqueness of the density}

In order to prove our main Theorem we need to recharacterise the evolution obtain in (\ref{martingaleproblem}) as 
the stochastic PDE. Thus we need the measure $\nu_t$ to be absolutely
continuous with respect to the Lebesgue measure to write $\nu_t(dx) = v(t,x)dx$ for some density $v$.


We introduce some notation first. Let $H^0=L^{2}((0, \infty))$ be the usual Hilbert space with $L^{2}$-norm 
$||\cdot||_{0}$ and inner product $\langle\cdot, \cdot\rangle_{0}$ given by
$||\phi||_{0}^{2}=\int_{0}^{\infty}|\phi(x)|^{2}dx$ and $\langle\phi,\psi\rangle_{0}= 
\int_0^{\infty}\phi(x)\psi(x)dx$. In the following we adapt the approach in \cite{ku1999} to our setting. 
The idea to prove the existence of an $L^{2}((0,\infty))$-density is to transform our
$\mathcal{M}((0,\infty))$-valued process to an $H^{0}$-valued process, by convolving the measure 
with the absorbing heat kernel, where $\mathcal{M}((0,\infty))$ denotes the set of finite Borel measures 
on $(0,\infty)$.

For any $\varrho\in\mathcal{M}((0,\infty))$ and $\delta>0$, we write
\begin{equation}
(T_{\delta}\varrho)(x)=\int_0^{\infty}G_{\delta}(x,y)\varrho(dy),
\end{equation}
where $G_{\delta}$ is the absorbing heat kernel in $\mathbb{R}^{+}$ given by
\[G_{\delta}(x,y)=\frac{1}{\sqrt{2\pi\delta}}\left(e^{-\frac{(x-y)^{2}}{2\delta}}- 
e^{-\frac{(x+y)^{2}}{2\delta}}\right), \;\forall x, y>0.\]
We use the same notation for the Brownian semigroup on $C_{b}(\mathbb{R}_{+})$, the bounded and continuous functions 
on $\mathbb{R}_{+}$, i.e.,
\[T_{t}\phi(x)=\int_0^{\infty}G_{t}(x,y)\phi(y)dy, \quad \forall \phi\in C_{b}(\mathbb{R}_{+}).\]

We will also need to consider the reflecting heat kernel $G_{\delta}^{r}(x, y)$, defined by
\[G_{\delta}^{r}(x,y)=\frac{1}{\sqrt{2\pi\delta}}\left(e^{-\frac{(x-y)^{2}}{2\delta}}+
e^{-\frac{(x+y)^{2}}{2\delta}}\right), \;\forall x, y>0. \]
We write the associated semigroup as 
\[T_{\delta}^{r}\nu_{t}(x)=\int_0^{\infty}G_{\delta}^{r}(x,y)\nu_{t}(dy).\]
Then it is an easy calculation to see that
\begin{equation}
\partial_{x}G_{\delta}(x,y)=-\partial_{y}G^{r}_{\delta}(x,y).\label{eq:diffrel}
\end{equation}

It is not difficult to prove the following lemma.
\begin{lem}
If $\varrho\in\mathcal{M}((0,\infty))$ and $\delta>0$, then $T_{\delta}\varrho\in H^0$.
\end{lem}

We will write $\nu_t\in H^0$ if the measure $\nu_t$ has a density which is in $H^0$.
Let $Z_{\delta}(s)=T_{\delta}\nu_{s}$, where $\nu$ is an $\mathcal{M}((0,\infty))$-valued solution 
to (\ref{martingaleproblem}).
Our aim is to obtain an estimate for the $H^{0}$-norm of the process $Z_{\delta}$. 

It is easy to see that $T_{\delta}\phi\in\bar{C}$ 
for any $\phi\in \bar{C}$. Thus, replacing $\phi\in\bar{C}$ by $T_{\delta}\phi$ in (\ref{martingaleproblem}) and using Fubini, 
we have
\begin{align}
\langle Z_{\delta}(t), \phi\rangle_{0} =& \langle T_{\delta}\phi, \nu_{t}\rangle\nonumber\\
=&\langle T_{\delta}\phi, \nu_{0}\rangle+\int_{0}^{t}\langle\mu(T_{\delta}\phi)^{\prime}(x)+ 
\frac{1}{2}(T_{\delta}\phi)^{\prime\prime}(x), \nu_{s} \rangle ds+ 
\int_{0}^{t}\langle\sqrt{\rho}(T_{\delta}\phi)^{\prime}(x), \nu_{s}\rangle dM_{s}.
\label{4}
\end{align}

The integrands can be rewritten as
\begin{align*}
\langle\mu(T_{\delta}\phi)^{\prime}(x), \nu_{s}\rangle=& 
\mu\int_0^{\infty}(T_{\delta}\phi)^{\prime}(x)\nu_{s}(dx)\allowdisplaybreaks\\
=&\mu\int_0^{\infty}\partial_{x}\left(\int_0^{\infty}G_{\delta}(x,y)\phi(y)dy\right)\nu_{s}(dx)
\allowdisplaybreaks\\
=&\mu\int_0^{\infty}\left(\int_0^{\infty}(\partial_{x}G_{\delta}(x,y))\phi(y)dy\right)\nu_{s}(dx)
\allowdisplaybreaks
\end{align*}
Applying (\ref{eq:diffrel}) and Fubini we have
\begin{align*}
\langle\mu(T_{\delta}\phi)^{\prime}(x), \nu_{s}\rangle
=&\mu\int_0^{\infty}\left(\int_0^{\infty}(-\partial_{y}G_{\delta}^{r}(x,y))\phi(y)dy\right)\nu_{s}(dx)
\allowdisplaybreaks\\
=&\mu\int_0^{\infty}\left(\int_0^{\infty}G_{\delta}^{r}(x,y)\phi^{\prime}(y)dy\right)\nu_{s}(dx)
\allowdisplaybreaks\\
=&\mu\int_0^{\infty} (T_{\delta}^{r}\nu_{s})(y)\phi^{\prime}(y)dy\allowdisplaybreaks\\
=&-\mu\int_0^{\infty} \phi(y)\partial_{y}(T_{\delta}^{r}\nu_{s}(y))dy\allowdisplaybreaks\\
=&-\mu\langle\phi, \partial_{x}T_{\delta}^{r}(\nu_{s})\rangle_{0}.
\end{align*}

Similarly, for the term $\langle\sqrt{\rho}(T_{\delta}\phi)^{\prime}(x), \nu_{s}\rangle$ we have
\[\langle\sqrt{\rho}(T_{\delta}\phi)^{\prime}(x), \nu_{s}\rangle=-\sqrt{\rho}\langle\phi, 
\partial_{x}T_{\delta}^{r}(\nu_{s})\rangle_{0}.\]

For the term $\langle\frac{1}{2}(T_{\delta}\phi)^{\prime\prime}(x), \nu_{s} \rangle$ we can perform the same type of calculation
to see
\[ \langle\frac{1}{2}(T_{\delta}\phi)^{\prime\prime}(x), \nu_{s} \rangle = 
\frac{1}{2}\langle\phi,\partial_{x}^{2}T_{\delta}(\nu_{s})\rangle_{0}. \]

Therefore (\ref{4}) becomes
\begin{align}
\langle Z_{\delta}(t), \phi\rangle_{0}=&\langle T_{\delta}\nu_{0}, \phi\rangle_{0}-\mu\int_{0}^{t}\langle\phi,
\partial_{x}T_{\delta}^{r}(\nu_{s})\rangle_{0}ds+\frac{1}{2}\int_{0}^{t}\langle\phi, 
\partial^{2}_{x}T_{\delta}(\nu_{s})\rangle_{0}ds-\sqrt{\rho}\int_{0}^{t}\langle\phi, 
\partial_{x}T_{\delta}^{r}(\nu_{s})\rangle_{0} dM_{s}.
\label{12}
\end{align}

By using It\^{o}'s formula on $\langle Z_{\delta}(s),\phi\rangle^2$ we have
\begin{align*}
\langle Z_{\delta}(t), \phi\rangle_{0}^{2}=&\langle Z_{\delta}(0), \phi\rangle_{0}^{2}+\int_{0}^{t}d\langle 
Z_{\delta}(s), \phi\rangle_{0}^{2}\allowdisplaybreaks\\
=&\langle Z_{\delta}(0), \phi\rangle_{0}^{2}+\int_{0}^{t}2\langle Z_{\delta}(s), \phi\rangle_{0} d\langle 
Z_{\delta}(s), \phi\rangle_{0}+\int_{0}^{t}d\Big\langle\langle Z_{\delta}(s), \phi\rangle_{0},\langle Z_{\delta}(s),
\phi\rangle_{0}\Big\rangle\allowdisplaybreaks\\
=&\langle Z_{\delta}(0), \phi\rangle_{0}^{2}-2\mu\int_{0}^{t}\langle Z_{\delta}(s), \phi\rangle_{0}\langle\phi,
 \partial_{x}T_{\delta}^{r}(\nu_{s})\rangle_{0}ds\allowdisplaybreaks\\
&+\int_{0}^{t}\langle Z_{\delta}(s), \phi\rangle_{0}\langle\phi, \partial^{2}_{x}T_{\delta}(\nu_{s})\rangle_{0}ds 
-2\sqrt{\rho}\int_{0}^{t}\langle Z_{\delta}(s), \phi\rangle_{0}\langle\phi, \partial_{x}T_{\delta}^{r}(\nu_{s})
\rangle_{0} dM_{s}\allowdisplaybreaks\\
&+\rho\int_{0}^{t}|\langle\phi, \partial_{x}T_{\delta}^{r}(\nu_{s})\rangle_{0}|^{2}ds.
\end{align*}

We can choose a set of $\phi\in\bar{C}$ to be a complete, orthonormal basis of $H^{0}$ and taking expectations, we have
\begin{align}
\mathbb{E}||Z_{\delta}(t)||_{0}^{2}=&||Z_{\delta}(0)||_{0}^{2}-2\mu\mathbb{E}\int_{0}^{t}\langle
Z_{\delta}(s),\partial_{x}T_{\delta}^{r}(\nu_{s})\rangle_{0}ds+\mathbb{E}\int_{0}^{t}\langle Z_{\delta}(s), 
 \partial^{2}_{x}T_{\delta}(\nu_{s})\rangle_{0}ds\nonumber\allowdisplaybreaks\\
&+\rho\mathbb{E}\int_{0}^{t}||\partial_{x}T_{\delta}^{r}(\nu_{s})||_{0}^{2}ds\nonumber\allowdisplaybreaks\\
=&||Z_{\delta}(0)||_{0}^{2}-2\mu\mathbb{E}\int_{0}^{t}\langle T_{\delta}(\nu_{s}), 
\partial_{x}T_{\delta}^{r}(\nu_{s})\rangle_{0}ds +\mathbb{E}\int_{0}^{t}\langle T_{\delta}(\nu_{s}),  
\partial^{2}_{x}T_{\delta}(\nu_{s})\rangle_{0}ds\nonumber\allowdisplaybreaks\\
&+\rho\mathbb{E}\int_{0}^{t}||\partial_{x}T_{\delta}^{r}(\nu_{s})||_{0}^{2}ds.
\label{13}
\end{align}

We now control the integral terms on the right-hand side of (\ref{13}) in terms of the integral of 
$\mathbb{E}||T_{\delta}(\nu_{s})||^{2}_{0}$ plus some constant which goes to 0 as $\delta\to 0$.

\begin{lem}
There exist constants $C_T^1,C_2$ such that for $\delta<\tilde{\varepsilon}_{0}^2/2$ we have
\begin{equation}
\mathbb{E}[|-2\mu\langle T_{\delta}(\nu_{s}),\partial_{x}T_{\delta}^{r}(\nu_{s})\rangle_{0}|]\leq 
|\mu|\cdot\mathbb{E}[||T_{\delta}(|\nu_{s}|)||_{0}^{2}]+C^{1}_{T}\delta^{\frac{\beta}{2}}+\frac{C_{2}\tilde{\varepsilon}_{0}}
{\delta^2} e^{-\tilde{\varepsilon}_{0}^2/2\delta}.
\end{equation}
\label{Lemma2}
\end{lem}

\begin{proof}
\begin{align*}
\langle T_{\delta}(\nu_{s}),\partial_{x}T_{\delta}^{r}(\nu_{s})\rangle_{0}=&\int_0^{\infty}T_{\delta}(\nu_{s})(x) 
\partial_{x}T_{\delta}^{r}(\nu_{s})(x)dx\allowdisplaybreaks\\
=&\int_0^{\infty}T_{\delta}(\nu_{s})(x)\left(\int_0^{\infty}\partial_{x}G_{\delta}^{r}(x,y)\nu_{s}(dy)\right)dx
\allowdisplaybreaks\\
=&\int_0^{\infty}T_{\delta}(\nu_{s})(x)\left(\int_0^{\infty}(\partial_{x}G_{\delta}(x,y)-
\frac{2}{\sqrt{2\pi\delta}}e^{-\frac{(x+y)^{2}}{2\delta}}\frac{x+y}{\delta})\nu_{s}(dy)\right)dx\allowdisplaybreaks\\
=&\int_0^{\infty}T_{\delta}(\nu_{s})(x)\int_0^{\infty}\partial_{x}G_{\delta}(x,y)\nu_{s}(dy)dx\allowdisplaybreaks\\
&-\int_0^{\infty}T_{\delta}(\nu_{s})(x)\int_0^{\infty}\frac{2}{\sqrt{2\pi\delta}}e^{-\frac{(x+y)^{2}}{2\delta}}
\frac{x+y}{\delta}\nu_{s}(dy)dx\allowdisplaybreaks\\
=&\int_0^{\infty}T_{\delta}(\nu_{s})(x)\partial_{x}T_{\delta}(\nu_{s})(x)dx-\int_0^{\infty}T_{\delta}(\nu_{s})(x)
\int_0^{\infty}\frac{2}{\sqrt{2\pi\delta}}e^{-\frac{(x+y)^{2}}{2\delta}}\frac{x+y}{\delta}\nu_{s}(dy)dx\allowdisplaybreaks\\
=&\frac{1}{2}\int_0^{\infty}\partial_{x}[(T_{\delta}(\nu_{s})(x))^{2}]dx-\int_0^{\infty}T_{\delta}(\nu_{s})(x)
\int_0^{\infty}\frac{2}{\sqrt{2\pi\delta}}e^{-\frac{(x+y)^{2}}{2\delta}}\frac{x+y}{\delta}\nu_{s}(dy)dx\allowdisplaybreaks\\
=&-\int_0^{\infty}T_{\delta}(\nu_{s})(x)\left(\int_0^{\infty}\frac{2}{\sqrt{2\pi\delta}}e^{-\frac{(x+y)^{2}}{2\delta}}
\frac{x+y}{\delta}\nu_{s}(dy)\right)dx.
\end{align*}
Therefore,
\begin{align*}
|-2\mu\langle T_{\delta}(\nu_{s}),\partial_{x}T_{\delta}^{r}(\nu_{s})\rangle_{0}|=&\left|2\mu\int_0^{\infty}T_{\delta}(\nu_{s})(x) 
\left(\int_0^{\infty}\frac{2}{\sqrt{2\pi\delta}}e^{-\frac{(x+y)^{2}}{2\delta}}\frac{x+y}{\delta}\nu_{s}(dy)\right)dx\right|
\allowdisplaybreaks\\
\leq&\left|\mu\int_0^{\infty}(T_{\delta}(\nu_{s})(x))^{2}dx+\mu\int_0^{\infty}\left(\int_0^{\infty}
\frac{2}{\sqrt{2\pi\delta}}e^{-\frac{(x+y)^{2}}{2\delta}}\frac{x+y}{\delta}\nu_{s}(dy)\right)^{2}dx\right|\allowdisplaybreaks\\
\leq&|\mu|\cdot||T_{\delta}(\nu_{s})||_{0}^{2}+|\mu|\int_0^{\infty}\left(\int_0^{\infty}\frac{2}{\sqrt{2\pi\delta}}
e^{-\frac{(x+y)^{2}}{2\delta}}\frac{x+y}{\delta}\nu_{s}(dy)\right)^{2}dx.
\end{align*}
Now let us denote
\[P_{1}:=\int_0^{\infty}\left(\int_0^{\infty}\frac{2}{\sqrt{2\pi\delta}}e^{-\frac{(x+y)^{2}}{2\delta}}
\frac{x+y}{\delta}\nu_{s}(dy)\right)^{2}dx\]
and investigate the bound for $P_{1}$.
\begin{align*}
P_{1}=&\int_0^{\infty}\left(\int_0^{\infty}\frac{2}{\sqrt{2\pi\delta}}e^{-\frac{(x+y)^{2}}{2\delta}}\frac{x+y}{\delta}
\nu_{s}(dy)\right)^{2}dx\allowdisplaybreaks\\
=&\int_0^{\infty}\int_0^{\infty}\int_0^{\infty}(\frac{2}{\sqrt{2\pi\delta}})^{2}
e^{-\frac{(x+y_{1})^{2}+(x+y_{2})^{2}}{2\delta}}\frac{(x+y_{1})(x+y_{2})}{\delta^{2}}\nu_{s}(dy_{1})\nu_{s}(dy_{2})dx\allowdisplaybreaks\\
=&\int_0^{\infty}\int_0^{\infty}\nu_{s}(dy_{1})\nu_{s}(dy_{2})\int_0^{\infty}(\frac{2}{\sqrt{2\pi\delta}})^{2}
e^{-\frac{(x+y_{1})^{2}+(x+y_{2})^{2}}{2\delta}}\frac{(x+y_{1})(x+y_{2})}{\delta^{2}}dx\allowdisplaybreaks\\
=&\int_0^{\infty}\int_0^{\infty}\nu_{s}(dy_{1})\nu_{s}(dy_{2})\int_0^{\infty}(\frac{2}{\sqrt{2\pi\delta}})^{2}
e^{-\frac{1}{\delta}\left[(x+\frac{y_{1}+y_{2}}{2})^{2}+(\frac{y_{1}-y_{2}}{2})^{2}\right]}\frac{(x+\frac{y_{1}+y_{2}}{2})^{2}
-(\frac{y_{1}-y_{2}}{2})^{2}}{\delta^{2}}dx.
\end{align*}
By changing variables using
\[z^{2}=(x+\frac{y_{1}+y_{2}}{2})^{2}+(\frac{y_{1}-y_{2}}{2})^{2},\]
we have
\[ P_{1}=\int_0^{\infty}\int_0^{\infty} \nu_{s}(dy_{1})\nu_{s}(dy_{2})\int_{\sqrt{\frac{y_{1}^{2}+y_{2}^{2}}{2}}}^{+\infty}
(\frac{2}{\sqrt{2\pi\delta}})^{2}e^{-\frac{z^{2}}{\delta}}\frac{z^{2}-\frac{(y_{1}-y_{2})^{2}}{2}}{\delta^{2}}
\frac{z}{\sqrt{z^{2}-\left(\frac{y_{1}-y_{2}}{2}\right)^{2}}}dz.\]

Since
\[1\leq\frac{z}{\sqrt{z^{2}-\left(\frac{y_{1}-y_{2}}{2}\right)^{2}}}\leq\sqrt{2}\]
when $z\geq \sqrt{\frac{y_{1}^{2}+y_{2}^{2}}{2}}$, we have
\begin{align*}
P_{1}\leq&\int_0^{\infty}\int_0^{\infty}\nu_{s}(dy_{1})\nu_{s}(dy_{2})
\int_{\sqrt{\frac{y_{1}^{2}+y_{2}^{2}}{2}}}^{\infty}\sqrt{2}(\frac{2}{\sqrt{2\pi\delta}})^{2}e^{-\frac{z^{2}}{\delta}}
\frac{z^{2}-\frac{(y_{1}-y_{2})^{2}}{2}}{\delta^{2}}dz\allowdisplaybreaks\\
=&\int_0^{\infty}\int_0^{\infty}\nu_{s}(dy_{1})\nu_{s}(dy_{2})
\int_0^{\infty}1_{\{y_{1}^{2}+y_{2}^{2}<2z^2\}}\sqrt{2}(\frac{2}{\sqrt{2\pi\delta}})^{2}e^{-\frac{z^{2}}{\delta}}
\frac{z^{2}-\frac{(y_{1}-y_{2})^{2}}{2}}{\delta^{2}}dz\allowdisplaybreaks\\
=&\int_{0}^{\infty}\sqrt{2}(\frac{2}{\sqrt{2\pi\delta}})^{2}e^{-\frac{z^{2}}{\delta}}dz\int_0^{\infty}\int_0^{\infty}
\frac{z^{2}-\frac{(y_{1}-y_{2})^{2}}{2}}{\delta^{2}}1_{\{y_{1}^{2}+y_{2}^{2}<
2z^{2}\}}\nu_{s}(dy_{1})\nu_{s}(dy_{2})\allowdisplaybreaks\\
\leq&\int_{0}^{\infty}\sqrt{2}(\frac{2}{\sqrt{2\pi\delta}})^{2}e^{-\frac{z^{2}}{\delta}}\frac{z^{2}}{\delta^{2}}dz
\int_0^{\infty}\int_0^{\infty}1_{\{y_{1}^{2}+y_{2}^{2}< 2z^{2}\}}\nu_{s}(dy_{1})\nu_{s}(dy_{2})\allowdisplaybreaks\\
\leq&\int_{0}^{+\infty}\sqrt{2}(\frac{2}{\sqrt{2\pi\delta}})^{2}e^{-\frac{z^{2}}{\delta}}\frac{z^{2}}{\delta^{2}}dz\int_0^{\infty}
\int_0^{\infty}1_{\{y_{1}< \sqrt{2}z,y_{2}<\sqrt{2}z}\}\nu_{s}(dy_{1})\nu_{s}(dy_{2})\allowdisplaybreaks\\
\leq&\int_{0}^{\infty}\sqrt{2}(\frac{2}{\sqrt{2\pi\delta}})^{2}e^{-\frac{z^{2}}{\delta}}\frac{z^{2}}{\delta^{2}}
(\nu_{s}((0, \sqrt{2}z)))^{2}dz\allowdisplaybreaks\\
=&\int_{0}^{\infty}\frac{2}{\pi\delta}e^{-\frac{z^{2}}{2\delta}}\frac{z^{2}}{2\delta^{2}}(\nu_{s}((0,z)))^{2}dz.
\end{align*}
Therefore,
\begin{eqnarray*}
\lefteqn{\mathbb{E}[|-2\mu\langle T_{\delta}(\nu_{s}),\partial_{x}T_{\delta}^{r}(\nu_{s})\rangle_{0}|]\allowdisplaybreaks}\\
&\leq&|\mu|\cdot\mathbb{E}[||T_{\delta}(\nu_{s})||_{0}^{2}]+|\mu|\mathbb{E}[P_{1}]\allowdisplaybreaks\\
&\leq&|\mu|\cdot\mathbb{E}[||T_{\delta}(\nu_{s})||_{0}^{2}]+|\mu|\mathbb{E}\left[\int_{0}^{\infty}\frac{2}{\pi\delta}
e^{-\frac{z^{2}}{2\delta}}\frac{z^{2}}{2\delta^{2}}(\nu_{s}((0,z)))^{2}dz\right]\allowdisplaybreaks\\
&=&|\mu|\cdot\mathbb{E}[||T_{\delta}(\nu_{s})||_{0}^{2}]+|\mu|\int_{0}^{\infty}\frac{2}{\pi\delta}e^{-\frac{z^{2}}{2\delta}}
\frac{z^{2}}{2\delta^{2}}\mathbb{E}[\nu_{s}((0,z))^{2}]dz.
\end{eqnarray*}
By Lemma \ref{Lemma1} in the last section we know that for the measure-valued solution $\nu_{s}^{+}$ of (\ref{martingaleproblem}), 
there exists $\tilde{\varepsilon}_{0}>0$ such that for all $z<\tilde{\varepsilon}_{0}$ we have
\[\mathbb{E}[\nu_{s}^{+}((0,z))^{2}]\leq K_{T}z^{3+\beta}.\]

Hence we have
\begin{align*}
&\mathbb{E}[|-2\mu\langle T_{\delta}(\nu_{s}),\partial_{x}T_{\delta}^{r}(\nu_{s})\rangle_{0}|]\allowdisplaybreaks\\
\leq&|\mu|\cdot\mathbb{E}[||T_{\delta}(\nu_{s})||_{0}^{2}]+|\mu|\int_{0}^{\tilde{\varepsilon}_{0}}\frac{2}{\pi\delta}
e^{-\frac{z^{2}}{2\delta}}\frac{z^{2}}{2\delta^{2}}\mathbb{E}[\nu_{s}((0,z))^{2}]dz+|\mu|\int_{\tilde{\varepsilon}_{0}}^{\infty}
\frac{2}{\pi\delta}e^{-\frac{z^{2}}{2\delta}}\frac{z^{2}}{2\delta^{2}}\mathbb{E}[(\nu_{s}((0,z]))^{2}]dz\allowdisplaybreaks\\
\leq&|\mu|\cdot\mathbb{E}[||T_{\delta}(\nu_{s})||_{0}^{2}]+|\mu|\int_{0}^{\tilde{\varepsilon}_{0}}\frac{2}{\pi\delta}
e^{-\frac{z^{2}}{2\delta}}\frac{z^{2}}{2\delta^{2}}K_{T}z^{3+\beta}dz+4|\mu|\int_{\tilde{\varepsilon}_{0}}^{\infty}
\frac{2}{\pi\delta}e^{-\frac{z^{2}}{2\delta}}\frac{z^{2}}{2\delta^{2}}dz\allowdisplaybreaks\\
\leq&|\mu|\cdot\mathbb{E}[||T_{\delta}(\nu_{s})||_{0}^{2}]+|\mu|\int_{0}^{\infty}\frac{2}{\pi\delta}e^{-\frac{z^{2}}{2\delta}}
\frac{z^{2}}{2\delta^{2}}K_{T}z^{3+\beta}dz+4|\mu|\int_{\tilde{\varepsilon}_{0}}^{\infty}\frac{2}{\pi\delta}e^{-\frac{z^{2}}{2\delta}}
\frac{z^{2}}{2\delta^{2}}dz\allowdisplaybreaks\\
=&|\mu|\cdot\mathbb{E}[||T_{\delta}(\nu_{s})||_{0}^{2}]+|\mu|K_{T}\frac{2^{3+\frac{\beta}{2}}}{\pi}\delta^{\frac{\beta}{2}}
\int_{0}^{\infty}e^{-x^{2}}x^{5+\beta}dx+4|\mu|\frac{2\sqrt{2}}{\pi\delta^{\frac{3}{2}}}\int_{\frac{\tilde{\varepsilon}_{0}}
{\sqrt{2\delta}}}^{\infty}e^{-x^{2}}x^{2}dx
\end{align*}
Finally we observe that for $\eta>0$
\[ \int_{\eta}^{\infty} x^2 e^{-x^2} dx \leq \frac12(\eta+\frac1\eta) e^{-\eta^2}, \]
and hence setting $\eta=\frac{\tilde{\varepsilon}_{0}}{\sqrt{2\delta}}$, so that by assumption $\eta>1$ we have
\[ \mathbb{E}[|-2\mu\langle T_{\delta}(\nu_{s}),\partial_{x}T_{\delta}^{r}(\nu_{s})\rangle_{0}|] \leq
 |\mu|\cdot\mathbb{E}[||T_{\delta}(\nu_{s})||_{0}^{2}]+C^{1}_{T}\delta^{\frac{\beta}{2}}+\frac{C_{2}
\tilde{\varepsilon}_{0}}{\delta^2} e^{-\tilde{\varepsilon}^2_{0}/2\delta}
\]
where 
\[C_{T}^{1}=|\mu|K_{T}\frac{2^{3+\frac{\beta}{2}}}{\pi}\int_{0}^{+\infty}e^{-x^{2}}x^{5+\beta}dx\]
is a constant and
\[C_{2}=|\mu|\frac{8\sqrt{2}}{\pi}.\]
\end{proof}

\begin{lem}
For $\delta<\tilde{\varepsilon}_0^2/2$ we have
\begin{equation}
\mathbb{E}[\langle T_{\delta}(\nu_{s}),  \partial^{2}_{x}T_{\delta}(\nu_{s})\rangle_{0}+\rho||\partial_{x} 
T_{\delta}^{r}(\nu_{s})||_{0}^{2}]\leq \frac{\rho}{1-\rho}\left(C^{1}_{T}\delta^{\frac{\beta}{2}}+ \frac{C_{2}\tilde{\varepsilon}_{0}}
{\delta^2} e^{-\tilde{\varepsilon}^2_{0}/2\delta} \right),
\end{equation}
where $C_{T}^{1}, C_{2}$ and $\tilde{\varepsilon}_{0}$ are the same as in Lemma \ref{Lemma2}.
\label{Lemma4}
\end{lem}

\begin{proof}
We have
\begin{align*}
||\partial_{x}T_{\delta}^{r}(\nu_{s})||_{0}^{2}=&\int_0^{\infty}(\partial_{x}T_{\delta}^{r}(\nu_{s})(x))^{2}dx\allowdisplaybreaks\\
=&\int_0^{\infty}\left(\int_0^{\infty}\partial_{x}G_{\delta}^{r}(x,y)\nu_{s}(dy)\right)^{2}dx\allowdisplaybreaks\\
=&\int_0^{\infty}\left(\int_0^{\infty}\left(\partial_{x}G_{\delta}(x,y)-\frac{2}{\sqrt{2\pi\delta}}
e^{-\frac{(x+y)^{2}}{2\delta}}\frac{x+y}{\delta}\right)\nu_{s}(dy)\right)^{2}dx\allowdisplaybreaks\\
=&\int_0^{\infty}\left(\int_0^{\infty}\partial_{x}G_{\delta}(x,y)\nu_{s}(dy)-\int_0^{\infty}\frac{2}{\sqrt{2\pi\delta}}
e^{-\frac{(x+y)^{2}}{2\delta}}\frac{x+y}{\delta}\nu_{s}(dy)\right)^{2}dx\allowdisplaybreaks\\
=&\int_0^{\infty}\left(\partial_{x}T_{\delta}\nu_{s}(x)-\int_0^{\infty}\frac{2}{\sqrt{2\pi\delta}}
e^{-\frac{(x+y)^{2}}{2\delta}}\frac{x+y}{\delta}\nu_{s}(dy)\right)^{2}dx\allowdisplaybreaks\\
=&||\partial_{x}T_{\delta}\nu_{s}||_{0}^{2}-2\int_0^{\infty}\partial_{x}T_{\delta}\nu_{s}(x)\int_0^{\infty}
\frac{2}{\sqrt{2\pi\delta}}e^{-\frac{(x+y)^{2}}{2\delta}}\frac{x+y}{\delta}\nu_{s}(dy)dx\allowdisplaybreaks\\
&+\int_0^{\infty}\left(\int_0^{\infty}\frac{2}{\sqrt{2\pi\delta}}e^{-\frac{(x+y)^{2}}{2\delta}}\frac{x+y}{\delta}
\nu_{s}(dy)\right)^{2}dx.
\end{align*}
Moreover,
\begin{eqnarray*}
\lefteqn{\rho\left(\langle T_{\delta}(\nu_{s}),  \partial^{2}_{x}T_{\delta}(\nu_{s})\rangle_{0}+||\partial_{x}T_{\delta}\nu_{s}||_{0}^{2}\right) 
\allowdisplaybreaks}\\
&=&\rho\left(\int_0^{\infty}T_{\delta}(\nu_{s})(x)\partial^{2}_{x}T_{\delta}(\nu_{s})(x)dx+\int_0^{\infty} 
\left(\partial_{x}T_{\delta}(\nu_{s})(x)\right)^{2}dx\right)\allowdisplaybreaks\\
&=&\rho\left(\int_0^{\infty}T_{\delta}(\nu_{s})(x)d(\partial_{x}T_{\delta}(\nu_{s})(x))+\int_0^{\infty}
\left(\partial_{x}T_{\delta}(\nu_{s})(x)\right)^{2}dx\right)\allowdisplaybreaks\\
&=&\rho\left(-\int_0^{\infty}(\partial_{x}T_{\delta}(\nu_{s})(x))^{2}dx+\int_0^{\infty}\left(\partial_{x}
T_{\delta}(\nu_{s})(x)\right)^{2}dx\right)\allowdisplaybreaks\\
&=&0.
\end{eqnarray*}
Also we know
\begin{align*}
(1-\rho)\langle T_{\delta}(\nu_{s}), \partial^{2}_{x}T_{\delta}(\nu_{s})\rangle_{0}=(1-\rho)\left(-\int_0^{\infty} 
(\partial_{x}T_{\delta}(\nu_{s})(x))^{2}dx\right).
\end{align*}
Therefore we have
\begin{eqnarray*}
\lefteqn{\langle T_{\delta}(\nu_{s}),  \partial^{2}_{x}T_{\delta}(\nu_{s})\rangle_{0}+\rho||\partial_{x}T_{\delta}^{r}(\nu_{s})||_{0}^{2} 
\allowdisplaybreaks}\\ 
&=& -2\rho\int_0^{\infty}\partial_{x}T_{\delta}\nu_{s}(x)\int_0^{\infty}\frac{2}{\sqrt{2\pi\delta}} 
e^{-\frac{(x+y)^{2}}{2\delta}}\frac{x+y}{\delta}\nu_{s}(dy)dx\allowdisplaybreaks\\
&&+\rho\int_0^{\infty}\left(\int_0^{\infty}\frac{2}{\sqrt{2\pi\delta}}e^{-\frac{(x+y)^{2}}{2\delta}}\frac{x+y}{\delta}
\nu_{s}(dy)\right)^{2}dx-(1-\rho)\int_0^{\infty}(\partial_{x}T_{\delta}(\nu_{s})(x))^{2}dx\allowdisplaybreaks\\
&\leq& \left|2\rho\int_0^{\infty}\sqrt{\frac{1-\rho}{\rho}}\partial_{x}T_{\delta}\nu_{s}(x)\int_0^{\infty} 
\sqrt{\frac{\rho}{1-\rho}}\frac{2}{\sqrt{2\pi\delta}}e^{-\frac{(x+y)^{2}}{2\delta}}\frac{x+y}{\delta}\nu_{s}(dy)dx\right|
\allowdisplaybreaks\\
&&+\rho\int_0^{\infty}\left(\int_0^{\infty}\frac{2}{\sqrt{2\pi\delta}}e^{-\frac{(x+y)^{2}}{2\delta}}\frac{x+y}{\delta}
\nu_{s}(dy)\right)^{2}dx-(1-\rho)\int_0^{\infty}(\partial_{x}T_{\delta}(\nu_{s})(x))^{2}dx\allowdisplaybreaks\\
&\leq&\rho\int_0^{\infty}\left(\sqrt{\frac{1-\rho}{\rho}}\partial_{x}T_{\delta}\nu_{s}(x)\right)^{2}dx+\rho\int_0^{\infty}
\left(\int_0^{\infty}\sqrt{\frac{\rho}{1-\rho}}\frac{2}{\sqrt{2\pi\delta}}e^{-\frac{(x+y)^{2}}{2\delta}}\frac{x+y}{\delta} 
\nu_{s}(dy)\right)^{2}dx\allowdisplaybreaks\\
&&+\rho\int_0^{\infty}\left(\int_0^{\infty}\frac{2}{\sqrt{2\pi\delta}}e^{-\frac{(x+y)^{2}}{2\delta}}\frac{x+y}{\delta}
\nu_{s}(dy)\right)^{2}dx-(1-\rho)\int_0^{\infty}(\partial_{x}T_{\delta}(\nu_{s})(x))^{2}dx\allowdisplaybreaks\\
&=&(1-\rho)\int_0^{\infty}\left(\partial_{x}T_{\delta}\nu_{s}(x)\right)^{2}dx+\frac{\rho}{1-\rho}\int_0^{\infty}
\left(\int_0^{\infty}\frac{2}{\sqrt{2\pi\delta}}e^{-\frac{(x+y)^{2}}{2\delta}}\frac{x+y}{\delta}\nu_{s}(dy)\right)^{2}dx
\allowdisplaybreaks\\
&&-(1-\rho)\int_0^{\infty}(\partial_{x}T_{\delta}(\nu_{s})(x))^{2}dx\allowdisplaybreaks\\
&=&\frac{\rho}{1-\rho}\int_0^{\infty}\left(\int_0^{\infty}\frac{2}{\sqrt{2\pi\delta}}e^{-\frac{(x+y)^{2}}{2\delta}}
\frac{x+y}{\delta}\nu_{s}(dy)\right)^{2}dx.
\end{eqnarray*}

By the estimate for $P_{1}$ obtained in Lemma \ref{Lemma2} we have
\[
\mathbb{E}[\langle T_{\delta}(\nu_{s}), \partial^{2}_{x}T_{\delta}(\nu_{s})\rangle_{0}+\rho||\partial_{x}T_{\delta}^{r}(\nu_{s})||_{0}^{2}]
\leq \frac{\rho}{1-\rho}\mathbb{E}[P_{1}] \leq \frac{\rho}{1-\rho}\left(C^{1}_{T}\delta^{\frac{\beta}{2}}+ 
\frac{C_{2}\tilde{\varepsilon}_{0}} {\delta^2} e^{-\tilde{\varepsilon}_{0}^2/2\delta} \right),
\]
where $C_{T}^{1}, C_{2}$ and $\tilde{\varepsilon}_{0}$ are the same as in Lemma \ref{Lemma2}.
\end{proof}

Now, combining Lemma \ref{Lemma2} and \ref{Lemma4} gives the following
\begin{thm}
If $\nu_{t}$ is an $\mathcal{M}(\mathbb{R}^{+})$-valued solution of ($\ref{martingaleproblem}$) 
and $Z_{\delta}(t)=T_{\delta}\nu_{t}$, we have for $\delta<\tilde{\varepsilon}_0^2/2$,
\begin{align}
\mathbb{E}||Z_{\delta}(t)||_{0}^{2}\leq& ||Z_{\delta}(0)||_{0}^{2}+|\mu|\int_{0}^{t}\mathbb{E} ||T_{\delta}(\nu_{s})||_{0}^{2}ds 
+\frac{1}{1-\rho}C^{1}_{T}\delta^{\frac{\beta}{2}}t\allowdisplaybreaks\nonumber\\
&+\frac{C_{2}t\tilde{\varepsilon}_{0}} {(1-\rho)\delta^2} e^{-\tilde{\varepsilon}^2_{0}/2\delta}.
\label{27}
\end{align}
\label{Th1}
\end{thm}

\begin{cor}
If $\nu_{t}$ is a measure-valued solution of ($\ref{martingaleproblem}$), then $\nu_{t}\in H^{0}$, $a.s.$ and $\mathbb{E}||\nu_{t}||_{0}^{2}<\infty$, $\forall t\geq 0$.
\label{Coro1}
\end{cor}

\begin{proof}
By ($\ref{27}$) we have for small $\delta$ that
\begin{align*}
\mathbb{E}||Z_{\delta}(t)||_{0}^{2}\leq& ||Z_{\delta}(0)||_{0}^{2}+|\mu|\int_{0}^{t}\mathbb{E} 
||T_{\delta}(\nu_{s})||_{0}^{2}ds+\frac{1}{1-\rho}C^{1}_{T}\delta^{\frac{\beta}{2}}T\allowdisplaybreaks\\
&+\frac{C_{2}T\tilde{\varepsilon}_{0}} {(1-\rho)\delta^2} e^{-\tilde{\varepsilon}_{0}^2/2\delta}\allowdisplaybreaks\\
:=&||Z_{\delta}(0)||_{0}^{2}+|\mu|\int_{0}^{t}\mathbb{E} ||Z_{\delta}(s)||_{0}^{2}ds+f(\delta,T),
\end{align*}
where
\[f(\delta,T)=\frac{1}{1-\rho}C^{1}_{T}\delta^{\frac{\beta}{2}}T+
\frac{C_{2}T\tilde{\varepsilon}_{0}} {(1-\rho)\delta^2} e^{-\tilde{\varepsilon}_{0}^2/2\delta}.\]

Applying Gronwall's inequality we have
\[\mathbb{E}||Z_{\delta}(t)||_{0}^{2}\leq (||Z_{\delta}(0)||_{0}^{2}+f(\delta,T)) e^{|\mu|t}.\]
It is clear that $\lim_{\delta\rightarrow 0}f(\delta,T)=0$. Now let $\{\phi_{j}\}$ be a complete, orthonormal system for $H^{0}$ 
such that $\phi_{j}\in C_{b}(\mathbb{R}^{+})$. Then by Fatou's lemma,
\begin{equation*}
\mathbb{E}\left[ \sum\limits_{j}\langle\phi_{j}, \nu_{t}\rangle^{2}\right]=\mathbb{E}\left[ \sum\limits_{j}
\lim_{\delta\rightarrow 0}\langle\phi_{j}, T_{\delta}\nu_{t}\rangle^{2}\right]\leq \liminf\limits_{\delta\rightarrow
 0}\mathbb{E}||Z_{\delta}(t)||_{0}^{2}\leq ||\nu_{0}||_{0}^{2} e^{|\mu|t},
\end{equation*}
Therefore $\nu_{t}\in H^{0}$ and $\mathbb{E}||\nu_{t}||_{0}^{2}<\infty$, $\forall t\geq 0$.
\end{proof}

Now we have proved the existence of an $L^{2}$-density for the limit empirical measure $\nu_{t}$, 
given that $\nu_{0}$ has an $L^{2}$-density.

\begin{thm}
Suppose that $\nu_{0}\in H^{0}$. Then (\ref{martingaleproblem}) has at most one measure-valued solution.
\end{thm}

\begin{proof}
Let $\nu^{1}_{t}$ and $\nu^{2}_{t}$ be two measure-valued solutions with the same initial value $\nu_{0}$, and both of them satisfy 
the boundary condition stated in Lemma \ref{Lemma1}. By Corollary \ref{Coro1}, $\nu^{1}_{t}, \nu^{2}_{t}\in H^{0}\; \text{ a.s.}$. 
Let $\nu_{t}=\nu^{1}_{t}-\nu^{2}_{t}$. Then $\nu_{t}\in H^{0}$ and also $\nu_{t}$ is a signed measure-valued solution to the 
equation (\ref{martingaleproblem}). It is straightforward to extend all the estimates we have obtained to the case of the
difference of two solutions as $|\nu_t| \leq \nu^1_t+\nu^2_t$ and the equations are linear.

Therefore by the appropriate extension of Theorem \ref{Th1} we have for $\delta<\tilde{\varepsilon}_0^2/2$
\[ \mathbb{E}||T_{\delta}\nu_{t}||_{0}^{2}\leq |\mu|\int_{0}^{t}\mathbb{E} ||T_{\delta}(|\nu_{s}|)||_{0}^{2}ds 
+\frac{2}{1-\rho}C^{1}_{T}\delta^{\frac{\beta}{2}}T+\frac{2C_{2}T\tilde{\varepsilon}_{0}} {(1-\rho)\delta^2}
 e^{-\tilde{\varepsilon}_{0}^2/2\delta}.\]
As before, taking $\delta\rightarrow 0$, we have
\[\mathbb{E}||\nu_{t}||_{0}^{2}\leq |\mu|\int_{0}^{t}\mathbb{E} |||\nu_{s}|||_{0}^{2}ds=|\mu|\int_{0}^{t}\mathbb{E} 
||\nu_{s}||_{0}^{2}ds,\]
and by Gronwall's inequality, we have $\nu_{t}\equiv 0$.
\end{proof}

This completes the proof of the uniqueness of the $L^2$-valued solution to the equation (\ref{martingaleproblem}).

\subsection{The limit SPDE}

Substituting the Lebesgue representation for the empirical measure into (\ref{martingaleproblem}),
integrating by parts 
and writing $\mathcal{A}^{\dag}$ for the adjoint operator of $\mathcal A$, we get
\begin{eqnarray*}
\int \phi(x)v(t,x)\,dx &=& \int \phi(x)v(0,x)\,dx + \int_0^t \int
\mathcal{A}\phi(x)v(s,x)\,dx\,ds
\\ &&\qquad + \int_0^t \int \srho\phi^{\prime}(x)v(s,x)\,dx\,dM_s \\
&=& \int \phi(x)v(0,x)\,dx + \int_0^t \int \phi(x)\mathcal{A}^{\dag}v(s,x)\,dx\,ds \\
& &\qquad - \int_0^t \int\phi(x) \frac{\partial}{\partial
x}\left(\srho v(s,x)\right)\,dx\,dM_s \\
&=& \int\phi(x)\left( v(0,x) + \int_0^t \mathcal{A}^{\dag}v(s,x)\,ds -
\int_0^t\frac{\partial}{\partial
x}\left(\srho v(s,x) \right)\,dM_s\right)\,dx.
\end{eqnarray*}
As this holds $\forall \phi \in \bar{C}$ we have shown that we have a weak solution to the SPDE
given by
\begin{equation}
\label{integralspde} v(t,x) =  v(0,x) + \int_0^t \mathcal{A}^{\dag}v(s,x)\,ds -
\int_0^t \frac{\partial}{\partial x}\left(\srho v(s,x)\right)\,dM_s,
\end{equation}
with $v(t,0)=0$ for all $t\in [0,T]$.
Alternatively, we can write this in differential form 
\begin{equation}
\label{spde} dv(t,x) = -\mu\frac{\partial v}{\partial x}(t,x) dt + \frac{1}{2}
\frac{\partial^2  v }{\partial x^2}(t,x) dt - \srho \frac{\partial v}{\partial
x}(t,x)\textrm{d}M_t,
\end{equation}
with $v(t,0)=0$ for all $t\in [0,T]$ and $v(0,x)=v_0(x)$.
This is a stochastic PDE that describes the evolution of the distance to default of 
an infinite portfolio of assets whose dynamics are given by (\ref{assetsde}). However
the derivatives are only defined in the weak sense.

We can now use the limiting empirical measure $\nu_t$ to approximate
the loss distribution for a portfolio of fixed size $N$ whose assets also follow
(\ref{assetsde}). We do this by matching the initial conditions, thus setting
\begin{equation}
v(0,x) = \frac{1}{N}\sum_{i=1}^N \delta_{X_0^i}(x),\label{initialcondition}
\end{equation}
where the $X_0^i>0$, $i=1,\ldots,N$ are the initial values for the distance to default of the 
assets in our fixed portfolio of size $N$.

\subsection{Solving the SPDE}

The SPDE (\ref{simplifiedspde}) without the boundary condition is easily solved as
\begin{equation}
 v(t,x) = u(t,x-\sqrt{\rho} M_t), \;\;\forall x\in \br, t>0,  \label{eq:fullspde}
\end{equation}
where $u(t,x)$ is the solution to the deterministic PDE
\begin{equation}
\label{eq:pde}
u_t = \frac12(1-\rho)u_{xx} - \frac{1}{\sigma}(r-\frac12\sigma^2) u_x, 
\end{equation}
with $u(0,x) = v_0(x)$.

The SPDE with the boundary condition has been treated in \cite{Kry94}. This allows us to complete the proof of
our existence and uniqueness theorem. 

\begin{thm}\label{thm:krylov}
Let $v_0(x)\in H^1((0,\infty))$.
The SPDE (\ref{simplifiedspde}) has a unique solution $u\in L^2(\Omega \times (0,T), {\cG},
H^1((0,\infty)))$ and is such that $xu_{xx} \in L^2(\Omega \times (0,T), {\cG}, L^2((0,\infty)))$.
\end{thm}

\begin{proof}
The result follows from Theorem~2.1 of \cite{Kry94}. Thus all we have to do is ensure 
that the conditions of that Theorem hold in our setting. The boundary of the domain $(0,\infty)$ is the
single point 0 and hence we can take the function $\psi(x)=\min(x,1)$ in the Theorem. The single point
boundary trivially satisfies the Hypothesis~2.1 of \cite{Kry94}. The coefficients of our SPDE 
are constants and hence satisfy the measurability requirement of Hypothesis~2.2 and the Lipschitz condition of 
Hypothesis~2.4. Hypothesis~2.3 also follows as the coefficients
are constants and the initial condition is in $H^1$.
\end{proof}

\begin{proof} \emph{(of Theorem~\ref{thm:main}}):
Our previous work has shown that the empirical measure satisfies (\ref{martingaleproblem}) and has a unique 
density in $L^2((0,\infty))$. By Theorem~\ref{thm:krylov} the SPDE with boundary condition has a unique solution
in $H^1((0,\infty))$. As this solution satisfies (\ref{martingaleproblem}), by the uniqueness of solutions, it must be
the density for our empirical measure. Thus our density satisfies the SPDE.
\end{proof}

We note that we can derive a formal expression for $L_t$ in terms of the density after integrating by parts.
\begin{align*}
L_{t}=&1-\int_{0}^{+\infty}v(t,x)dx\allowdisplaybreaks\\
=&1-\int_{0}^{+\infty}\left(v(0,x)-\int_{0}^{t}\frac{\partial}{\partial x}\mu 
v(s,x)ds+\int_{0}^{t}\frac{1}{2}v_{xx}(s,x)ds\right.\allowdisplaybreaks\\
&\left.-\int_{0}^{t}\frac{\partial}{\partial x}\sqrt{\rho}v(s,x)dM_{s}\right)dx\allowdisplaybreaks\\
=&1-\int_{0}^{+\infty}v(0,x)dx+\mu \int_{0}^{t} v(s,x)|_{x=0}^{x=\infty}ds-\int_{0}^{t}\frac{1}{2} 
v_{x}(s,x)|_{x=0+}^{x=\infty}ds\allowdisplaybreaks\\
&+\sqrt{\rho}\int_{0}^{t} v(s,x)|_{x=0}^{x=\infty}dM_{s}.
\end{align*}
Since $x^{i}>0, \forall i$ and $X_{t}^{i}$  is a continuous process, we can conclude that $T_{0}^{i}>0, \forall i$. Thus
\[ L_{0}=\lim_{N\rightarrow\infty}\frac{1}{N}\sum_{i=1}^{N}1_{\{0\geq T_{0}^{i}\}}=0,\]
therefore
\[ \bar{\nu}(\mathbb{R^{+}}\cup\{0\})=1=\int_{0}^{+\infty}v(0,x)dx.\]
Moreover we have $v(s,x)\rightarrow 0, v_{x}(s,x)\rightarrow 0$, as $x\rightarrow\infty$ and $v(s,0)=0$, $\forall s$. 
Therefore, provided that $v_x(s,0)$, the right derivative of $v(s,x)$ with respect to $x$ at the point $x=0$, exists we
would have
\[ L_{t}=\frac{1}{2}\int_{0}^{t}v_{x}(s,0)ds.\]

One issue that has not been addressed is the existence of $C^2$ solutions to this equation. We note that the work of 
Lototsky \cite{Lot99} shows that there is a classical $C^2$ solution to this SPDE over a bounded domain $(0,K)$, with Dirichlet
boundary conditions at 0 and $K$, provided that the initial condition is smooth enough.

\subsection{The portfolio loss} 

We would like to price portfolio credit derivatives whose values depend on the
cumulative defaults occurring within a reference basket of risky assets. The key to
pricing these instruments is determining the joint loss distribution. We have just
derived an equation that describes the evolution of the empirical measure of the limiting large portfolio
of assets. At any future value in time, we can determine the loss in the portfolio by
calculating the total mass of the empirical measure of assets that have not defaulted. 
Thus the portfolio loss $L^N_t$ can be approximated by
\begin{displaymath}
L^N_t = N L_t,\label{portfolioloss}
\end{displaymath}
where $N$ is the number of assets in the portfolio. We note that given the initial
condition (\ref{initialcondition}) we have $L^N_0=0$.
Also, due to the way in which defaults are incorporated into the model, we have
\begin{displaymath}\begin{array}{cc}
0 \leq\  L_t \leq 1, & \textrm{for } t\geq 0\\
P(L_s \geq K) \leq P(L_t \geq K ), & \textrm{for } s\leq t,
\end{array}\end{displaymath}
which ensures that there is no arbitrage in the loss distribution. Both of these
properties are expected for a model of cumulative loss in a portfolio.

\subsection{A connection with filtering}

We note that the SPDE can be viewed as a PDE with a Brownian drift. This is easily seen through 
an interpretation as the Zakai equation for a filtering problem. Let
$(\tO,\tcf,\tbp)$ be a probability space. Under $\tbp$ we
define the signal process $X$ to be a stochastic process satisfying
\begin{eqnarray*}
 dX &=& \mu dt - \srho dM + \sqrt{1-\rho} dW, \;\; t\leq \tau_0 \\
X_t &=& 0, \;\; t>\tau_0
\end{eqnarray*}
where $\tau_0 = \inf\{t: X_t = 0\}$, where $\mu,\rho$ are
constants and $M$ and $W$ are independent Brownian motions and $X_0=x$. 
The observation process $Y$ is taken to be just the market noise,
\[ dY_t = dM_t, \]
then the Zakai equation (see for example \cite{BainCris}) for the conditional distribution of the signal
given the observations is exactly our SPDE.


Thus, by standard filtering theory,
 if we want to compute a functional of the signal we need to
calculate
\[ m_{\psi}(t) = \tbe(\psi(X_t)|\cf^M_t) = \int \psi(y) u(t,x) dx. \]
This means that the probability distribution for the position of a
company given the market noise has a density $u(t,x)$ satisfying
\[ du(t,x) = (-\mu u_x(t,x)+\frac12 u_{xx}(t,x))dt -\srho
u_x(t,x) dM_t, \]
with $u(0,x)=u_0(x)$, that is the initial guess at $X_0$ is the
density $u_0(x)$ and $u(t,0)=0$. 
Thus for the loss function we are interested in computing the
proportion of companies that have defaulted by time $t$ and this can
be found by computing $m_{\psi}(t)$ for $\psi(t) = I_{\{\tau_0<t\}}$. 
If we start from a given fixed point so that $u_0(x)$ is a
delta function at $x$. Then
\[ L_t = m_{\psi}(t) = \tbp^x(\inf_{s\leq t} X_s<0|\cf^M_t). \]

Now the process $X$ can be written as a Brownian motion with drift
\[ X_t = x+\mu t - \srho M_t + \sqrt{1-\rho}W_t, \]
and if we are given $M$, this can be expressed as
\[ X_t = \sqrt{1-\rho} \left(\frac{x+f(t)}{\sqrt{1-\rho}} + W_t\right), \]
where $f(t) = \mu t - \srho M_t$ is a deterministic time dependent drift function 
which is a fixed random path.

Thus to compute the random loss function we set $x'=x/\sqrt{1-\rho}, g(t) = f(t)/\sqrt{1-\rho}$ and write
\begin{eqnarray*}
\tbp^x(\inf_{s\leq t} X_s<0|\cf^M_t) &=& \tbp^0(\inf_{s\leq t}
x' + g(s)+ W_s<0|\cf^M_t) \\
&=& \tbp(\inf_{s\leq t} g(s)+W_s<-x'|\cf^M_t).
\end{eqnarray*}
In the case where we have a general initial distribution $u_0(x)$, the loss function is then
\[ L_t = \int_0^{\infty} u_0(x) \tbp(\inf_{s\leq t} g_s+W_s<-x/\sqrt{1-\rho}|\cf^M_t) dx. \]

Thus we can try to compute this by solving the hitting time problem
for Brownian motion with time dependent drift for a fixed realization
of the market noise. It is straightforward to
use this to simulate a realization of the loss function. 



To derive this SPDE we made some simplifying assumptions. The first of these arose when
specifying the asset processes in (\ref{assetsde}). We had to set the drift and
volatility of all the assets to some common value. For the drift this is not a problem,
because under the risk neutral measure it will be transformed to a value that excludes
arbitrage. The fact that there is only one yield curve means that this value will be the
same for all assets. If our reference portfolio contained entities denominated in more
than one currency this would not be the case and some approximation would have to be
made.

This argument cannot be used for the volatility as it is not affected by a change of
measure. Therefore, it would seem that giving the assets one common value of volatility
is a very restrictive assumption. However, for any given value of the volatility we still have the
freedom to choose the default barrier specific to any one asset. Via the
distance-to-default transformation this freedom manifests itself in our particular choice
of starting value for each process. The effect of changing the barrier and changing the
volatility is very similar. To see this note that default risk is measured by how many
standard deviations away from the barrier our process is. To increase the default risk we
need to reduce this distance which can be done by either increasing the standard
deviation or moving the barrier closer. Although these are clearly not equivalent
transformations they have a very similar effect and so the single volatility assumption
is not as restrictive as it initially appears.

Having a single volatility number also eases calibration as we do not have to estimate
the volatilities of all of the entities within our portfolio. Instead, we will have to
replace it by some `average' market volatility. Not only will this help day-to-day
calibration stability but it means that credit derivative prices will be a function of
one volatility parameter only. This is usually a desirable property from a practitioner's
point of view as it allows one to take a view on that parameter; this cannot be done if
there were a single parameter for each entity within our portfolio.

The major simplification that allowed us to derive our SPDE came when we moved to an
infinite dimensional limit. In this limit, the idiosyncratic noise of the individual
assets is averaged out. In fact, we could have any number of idiosyncratic
components, provided they are independent and uncorrelated, and they would average out to
zero. It is only the correlated components between the assets that remain i.e. the market
risk. Note that this means that if the limiting portfolio was fully diversified, that is had 
no correlation, there would be no noise in the limit and the limit portfolio would 
evolve deterministically!



\section{Numerical solution}\label{section:numerics}

We outline in the following a numerical method for approximating the solution to the
SPDE, which we use in the market pricing examples in the next section.
We start with the SPDE (\ref{martingaleproblem}) in weak form, repeated here for convenience,
\[
\left\langle\phi,\nu_t\right\rangle = \left\langle\phi,\nu_0\right\rangle +
\int_0^t\left\langle \mathcal{A}\phi,\nu_s \right\rangle ds +
\int_0^t\left\langle \srho\phi^{\prime},\nu_s\right\rangle
dM_s
\]
for almost all $t$ and all smooth test functions $\phi \in \bar{C}$.
It follows from Theorem \ref{thm:main} that $\nu_t$ has as one component the density $v$ (describing the non-absorbed element)
satisfying
\begin{equation}
\left(\phi,v(t,\cdot)\right) = \left(\phi,v(0,\cdot)\right) +
\int_0^t\left( \mathcal{A}\phi,v(s,\cdot) \right) ds +
\srho \int_0^t\left(\phi^{\prime},v(s,\cdot) \right)
dM_s,
\end{equation}
where here we write $(\cdot,\cdot)$ for the $L^2$ inner product.
Integrating by parts, noting from Theorem \ref{thm:krylov} that $v(t,\cdot) \in H^1_0$
with dense subspace $\bar{C}$,
\begin{eqnarray*}
\left(\phi,v(t,\cdot)\right) 
 + \int_0^t a(\phi,v(s,\cdot)) ds
= \left(\phi,v(0,\cdot)\right) +
\srho \int_0^t\left(\phi^{\prime},v(s,\cdot) \right) dM_s
\end{eqnarray*}
for all $\phi \in H_0^1$, where
\begin{eqnarray*}
a(\phi,v) = \frac{1}{2} (\phi',v') - \sqrt{\rho} (\phi',v).
\end{eqnarray*}

\subsection{Finite element approximation}

Let $V_h\subset H_0^1([x_0,x_N])$ be the space of piecewise linear functions 
on a grid $x_1<\ldots<x_N$, which are zero at $x_1=0$ and $x_N$ a sufficiently large value (see \ref{subsec:accuracy}).
Denote further by $\{\phi_n: 1\le n\le N\}$ the standard finite element basis
(see e.g.\ \cite{quaval97} for standard finite element theory and approximations to PDEs).
Restricting both the solution 
and test functions to $V_h$,
\begin{eqnarray*}
\left(\phi_n,v_h(t,\cdot)\right) 
 + \int_0^t a(\phi_n,v_h(s,\cdot)) ds
= \left(\phi_n,v_h(0,\cdot)\right) +
\srho \int_0^t\left(\phi_n^{\prime},v_h(s,\cdot) \right) dM_s
\end{eqnarray*}
(for all $1\le n\le N$) defines a semi-discrete finite element
approximation.

Using the stochastic $\theta$-scheme (see \cite{hig00}) for the time 
discretisation of the resulting SDE system,
\begin{eqnarray}
\label{stochtheta}
\left(\phi_n,v_h^{m+1}\right) 
 + \theta \Delta t \, a(\phi_n,v_h^{m+1})
= \left(\phi_n,v_h^{m}\right)
- (1-\theta) \Delta t \, a(\phi_n,v_h^{m}) +
\srho \left(\phi_n^{\prime},v_h^m \right) \sqrt{\Delta t} \Phi_m,
\end{eqnarray}
where $\Phi_m \! \sim \! N(0,1)$, $\Delta t = t_{m+1}-t_m$ is assumed constant and
$v_h^m = \sum_{n=1}^N v_n^m \phi_n$. Thus one gets a linear system
\begin{equation}
\label{linsys}
(M + \theta \Delta t A) v^{m+1} = (M-(1-\theta)\Delta t A) v^m +
\srho \sqrt{\Delta t} \Phi_m D v^m,
\end{equation}
where $v^m=(v_1^m,\ldots,v_N^m)$ and the standard finite element 
matrices are given by
\begin{eqnarray*}
M_{ij} &=& (\phi_i, \phi_j), \quad 1\le i,j\le N,\\
A_{ij} &=& a(\phi_i, \phi_j), \quad 1\le i,j\le N, \\
D_{ij} &=& (\phi_i', \phi_j), \quad 1\le i,j\le N.
\end{eqnarray*}
This gives a pathwise (in $M$, the market factor)
approximation to the SPDE solution via timestepping from an initial
density $v_h(0,\cdot)$, which is found by $L^2$ projection of
$\bar{\nu}_{N_f,t}$ from (\ref{empiricalmeasure}) with $N_f$ firms onto
the finite element space (see e.g.\ \cite{pooley1}, \cite{rannacher1}).

\subsection{Simulating tranche spreads}

For a given (numerical) realisation of the market factor,
we can approximate the loss functional $L_{T_k}$ at time $T_k$ by
\begin{equation}
\label{loss-deltax}
L_{T_k}^{h} = 1- \int_{0}^{x_N} v_h(T_k,x) {\, \rm d}x
\approx 1- h \sum_{n=1}^{N-1} v_n^m
\end{equation}
where $m = T_k/{\Delta t}$.
If we explicitly include the dependence on the Monte Carlo samples
${\Phi} = (\Phi_i)_{1\le i\le I}$ in
$L_{T_k}^{h}(\Phi)$, where $\Phi_i$ as in (\ref{stochtheta})
are drawn independently from a standard normal distribution,
then for $N_{sims}$ simulations with samples
${\Phi^l} = (\Phi_i^l)_{1\le i\le I}$,
$1\le l\le N_{sims}$, we simulate the outstanding tranche notional
(\ref{outstandingnotional}) as
\begin{eqnarray*}
\mathbb{E}^{\mathbb{Q}}[Z_{T_{k}}] &\approx&
\mathbb{E}^{\mathbb{Q}}[\max(d-L_{T_k}^{h},0) - \max(a-L_{T_k}^{h},0)] \\ &\approx&
\frac{1}{N_{sims}} 
\sum_{l=1}^{N_{sims}} \left( \max(d-L_{T_k}^{h}(\Phi^l),0) - \max(a-L_{T_k}^{h}(\Phi^l),0) \right).
\end{eqnarray*}
This gives simulated tranche spreads via (\ref{cdofeeleg}),
(\ref{cdoprotectionleg}) and (\ref{cdospread}).

\subsection{Accuracy and further approximations}
\label{subsec:accuracy}

We now discuss the approximations made previously and further simplifications
made in the numerical implementation of the examples in the next section.

It is necessary for the finite element discretisation to approximate the
semi-infinite boundary value problem for the SPDE by one on a finite domain.
It is expected that if the upper boundary is sufficiently large, dependent on the
initial distances-to-default and model parameters, the probability of crossing this
boundary can be made negligible and zero boundary conditions are appropriate.
We have checked this to be the case for the following numerical simulations
but do not have a theoretical justification at this point.

The derivation of the SPDE and finite element solution assume $H^1$ initial data, however in practice
we want to use a sum of atomic measures (\ref{initialcondition}) corresponding to the distance-to-default of individual firms,
as backed out from CDS spreads. We deal with this by projecting these data onto the finite element
basis (see e.g.\ \cite{pooley1}, \cite{rannacher1}).

The majority of the literature on stochastic finite element methods deals with stochastic
diffusion coefficients (see e.g.\ \cite{debbabode01} and subsequent work)
and we are not aware of results which cover our setting
with stochastic drift.
From standard finite element approximation results for PDEs (see e.g.\ \cite{quaval97}),
one would expect (pathwise) convergence order two in $h$ for solutions in $H^2$, but
Theorem \ref{thm:krylov} suggests weaker regularity at the absorbing boundary,
which we also observe in the numerical solutions. This does not show a measurable
impact on the numerical accuracy in practice.
The weak approximation order of the Euler scheme for SDEs,
and that for the chosen fully implicit scheme for PDEs
($\theta=1$ in (\ref{stochtheta})),
is one (in $\Delta t$).
In this case, the scheme is stable in the mean-square sense of \cite{hig00}.
This is confirmed by numerical experiments, but a rigorous numerical analysis is beyond
the scope of this paper.

A common approximation to the finite element system is to `lump' $M$ in
(\ref{linsys}) in diagonal form,
interpretable as application of a quadrature rule, and ultimately results in $M$ being replaced
by a multiple of the identity matrix. With this approximation, the finite element scheme
becomes identical to a central finite difference approximation.

A further simplification is suggested by the solution (\ref{eq:fullspde}) of the SPDE without absorbing boundary
condition, which decouples the solution into the PDE solution (\ref{eq:pde})
on a doubly-infinite domain, and a
random (normal) offset. This is easy to implement if we apply boundary conditions only at
a discrete set of times.
In analogy to discretely sampled barrier options,
this corresponds to a situation where we observe default not continuously, but only at discrete dates.
The numerical results in the next section were obtained in this way with default monitoring
at payment dates
for computational convenience. This introduces a small shift in the
calibrated parameters compared to the SPDE with continuously absorbing barrier
but the reported results on tranche spreads are almost identical.

The Monte Carlo estimates of outstanding tranche notionals and
subsequently tranche spreads converge per $N_{sims}^{-1/2}$. The variance relative
to the spread is larger for senior tranches due to the rarity of losses in these
tranches, as illustrated by Figure \ref{fig:mcconv}.
\begin{figure}[ht]
\psfrag{log4Nsims}[r][r][0.8]{$\log_4 N_{sims}$}
\psfrag{junior}[d][t][0.8]{$\mathbb{E}^{\mathbb{Q}}[Y_{T}]$, $a=0, d=0.03$}
\includegraphics[width= 0.32 \columnwidth]{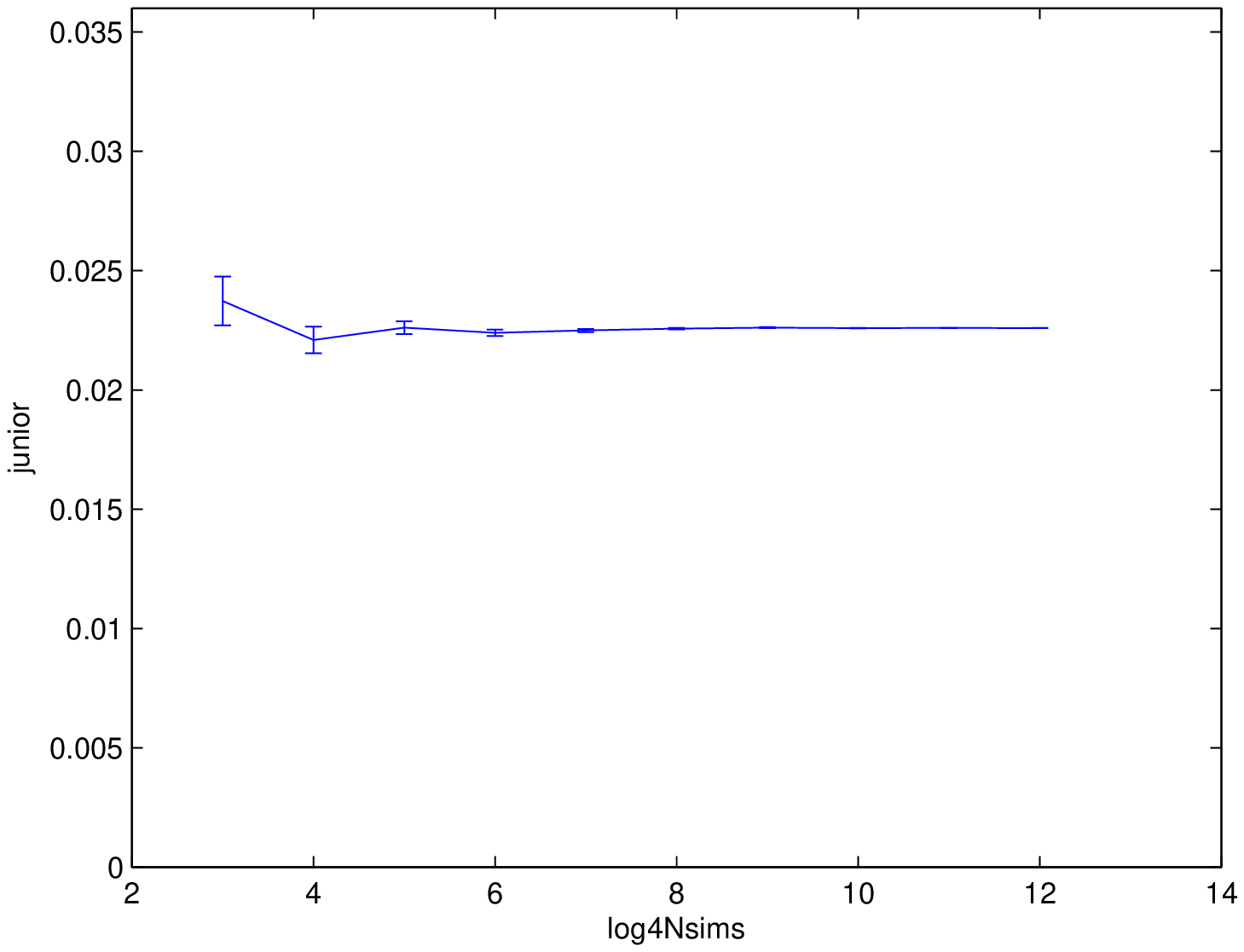}
\hfill
\psfrag{mezz2}[d][t][0.8]{$\mathbb{E}^{\mathbb{Q}}[Y_{T}]$, $a=0.06, d=0.09$}
\includegraphics[width= 0.32 \columnwidth]{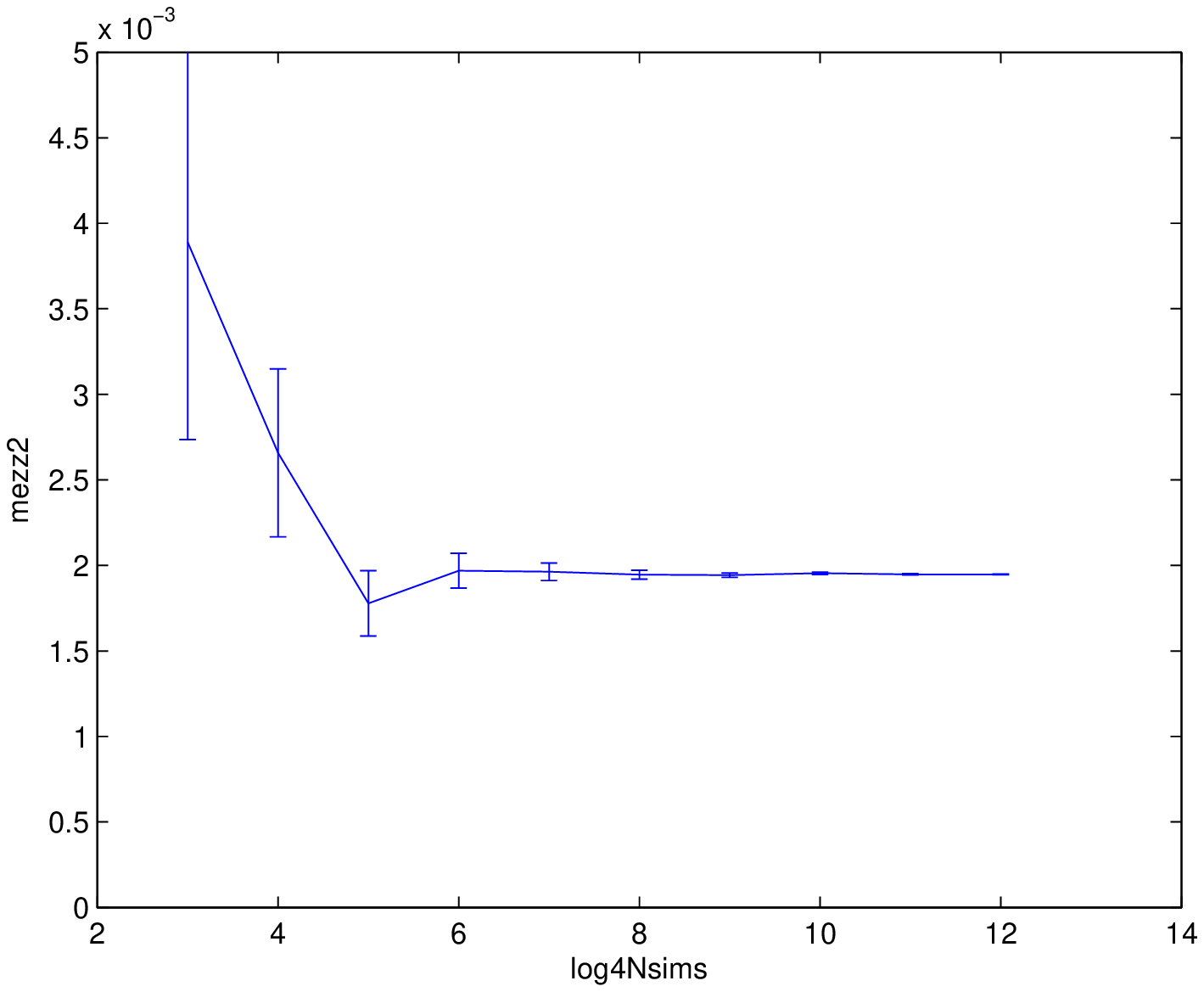}
\hfill
\psfrag{supersenior}[d][t][0.8]{$\mathbb{E}^{\mathbb{Q}}[Y_{T}]$, $a=0.12, d=0.22$}
\includegraphics[width= 0.32 \columnwidth]{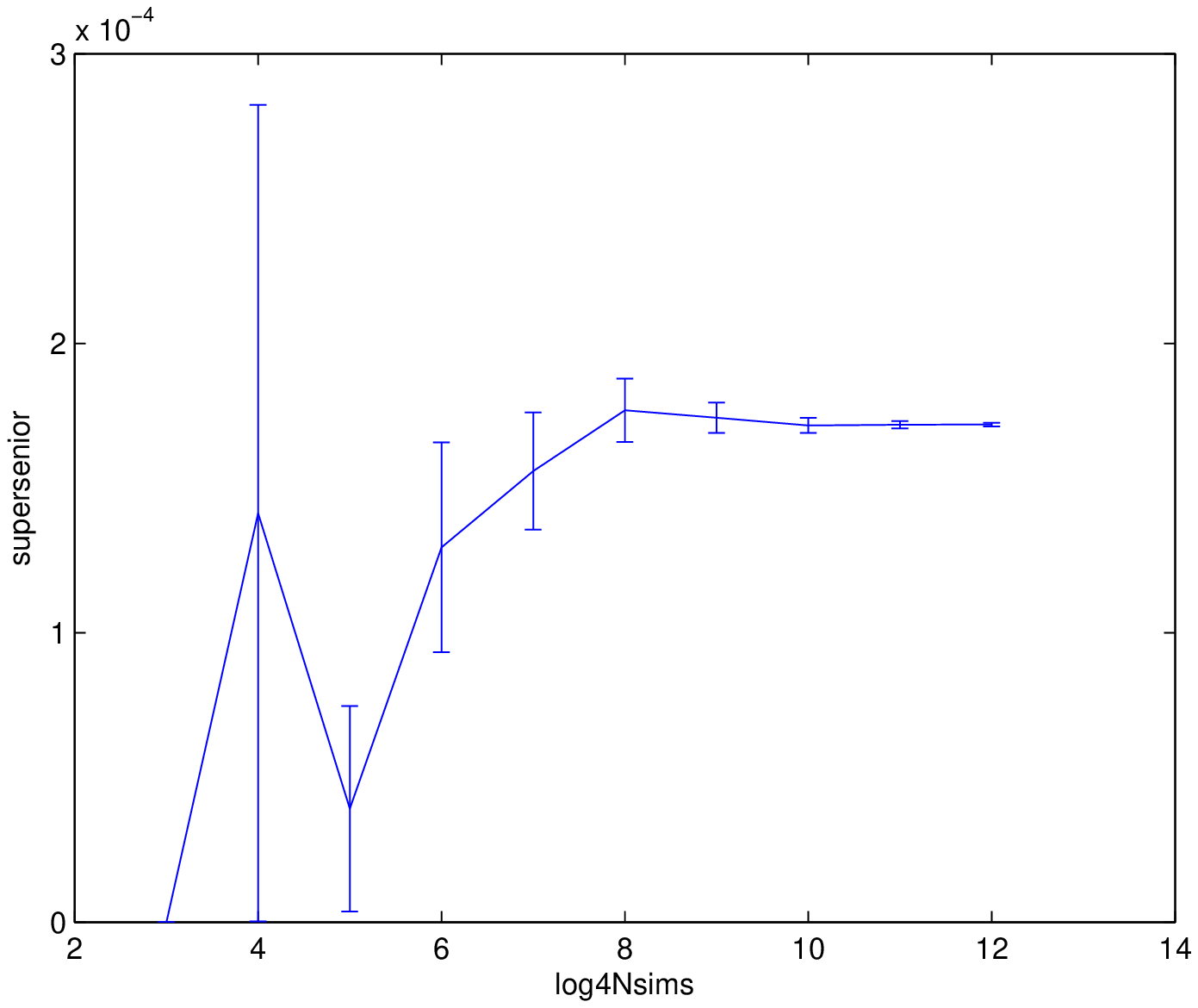}
\caption{Monte Carlo estimators with standard error bars
for expected losses (\ref{trancheloss}) in tranches
$[0,3\%]$, $[6\%,9\%]$, $[12\%,22\%]$, for
$N_{sims} = 16\cdot 4^{k-1}$, $k=1,...,10$, 
and a typical set of parameters, maturity $T=5$.}
\label{fig:mcconv}
\end{figure}
Importance sampling could cure this problem but was not found necessary for the purposes of this
study.

Numerical parameters were in the following adjusted such that the (heuristically) estimated
approximation error
was sufficiently small compared to the effects observed by varying model parameters.

\section{Market pricing examples}\label{section:pricing}

\subsection{Calibration to index tranches}

In this section, we analyse our model's ability to price regular index tranches for all maturities and investigate the implied 
correlation skew. We consider performance pre and post the onset of the credit crunch, illustrating the model's inherent 
ability to cope with a variety of credit environments. 

Throughout the analysis, we infer the initial condition from market spreads for the underlying index constituents, rather 
than allowing it to be a free parameter to be fixed by calibration to index tranches.
This is to be consistent with CDS spreads for the individual constituents.
We do this by backing out the  distance-to-default for each constituent from its five-year CDS spread and then 
aggregating these. 
Note that as we model the distance-to-default as in (\ref{distance-to-default}), different volatilities of the underlying firms can be taken into account by
rescaling.
As a consequence, the initial condition is driven by both the level of constituent spreads and their dispersion. 


We study the ability of our model to price index tranches on two dates: February 22, 2007 and December 5, 2008. These dates are
chosen specifically to investigate the flexibility of the model to cope with different market and spread environments.
February 22, 2007 was pre-crisis when spreads were tight and curves upward sloping; December 5, 2008 was at the height of market 
volatility, when spreads were at their widest and curves frequently inverted.

We set $R=40\%$, the level typically assumed by the market for investment grade names, and for each date, calibrate the model 
to $5$, $7$ and $10$-year index spreads using the volatility, $\sigma$. $r$ is the risk-free rate obtained from the Euro swap curve. 
(N.B. the correlation parameter, $\rho$, does not come into this calibration since index spreads depend only on the expected losses, which are identical to
the sum of default probabilities and hence correlation-independent.) 

Table \ref{22Feb07itraxxspreads} shows the traded and model index spreads for Feb 22, 2007. Since we derive the initial condition 
from constituent spreads, we only have one free parameter, the volatility $\sigma$, for calibrating all three index spreads. 
Increasing $\sigma$ to increase model spreads also causes the initial distance-to-default for each constituent to increase (since 
CDS spreads are fixed), so index and tranche spreads are less sensitive to changes in volatility than they would be if the initial
condition was specified independently.

\begin{table}[ht]
\begin{center}
\begin{tabular}{|c|c|c|c|}
  \hline
  Maturity Date & Fixed Coupon (bp) & Traded Spread (bp) & Model Spread (bp)\\
  \hline
  20/12/2011 & 30 & 21 & 19.6 \\
  20/12/2013 & 40 & 30 & 30.7 \\
  20/12/2016 & 50 & 41 & 41.0 \\
  \hline
\end{tabular}
\caption{The fixed coupons, traded spreads and model spreads for the iTraxx Main Series 6 index on
February 22, 2007. Parameters used for the model spreads are $r=0.042$, $\sigma=0.22$, $R=0.4$.} \label{22Feb07itraxxspreads}
\end{center}
\end{table}

Table \ref{05Dec08itraxxspreads} shows the same results for Dec 5, 2008. In this highly distressed state, we notice that spreads 
are dramatically wider and the curve is inverted with $5$-year $ > $ $7$-year $ > $ $10$-year spreads. Our simple model again does 
a good job of calibrating all three spreads. This is achieved by a smaller distance-to-default for the initial
positions in combination with a lower volatility, triggering more defaults in the near future.
The $5$-year point is a little low, which is a shortcoming of using a purely diffusive driving process: it can be hard to generate 
sufficient short-term losses.
We refer to Section \ref{conclusions} for a discussion of extensions to jump and 
stochastic volatility driven processes.

For the parameters from the calibration in Table \ref{22Feb07itraxxspreads}, Table \ref{22Feb07tranchespreads} illustrates 
the correlation sensitivity of the $5$, $7$ and $10$-year index tranches in the pre-crunch environment. We note that model 
spreads illustrate the behaviour we would anticipate:
\begin{itemize}
\item Equity tranche spreads decline with increasing correlation whilst spreads for other tranches generally increase with correlation. 
As correlation increases, there are less likely to be a few defaults, and so the equity tranche becomes less risky and its spread
decreases. The probability of a greater number of defaults increases with increasing correlation and so spreads on the more 
senior tranches increase with correlation.
\item A notable exception is the $10$-year junior mezzanine tranche ($3\%-6\%$) which behaves more like an equity tranche and has 
declining spreads with increasing correlation. This is because, for the parameters used, the expected index loss is between 
$3\%$ and $6\%$. The risk of this tranche therefore decreases, along with the spread, as correlation increases, making losses 
in this tranche less likely.
\item
The $7$-year junior mezzanine tranche ($3\%-6\%$) spreads indicate the
transition, as maturity increases, from positive to negative correlation 
sensitivity by exhibiting a humped shape.
\item For the $5$ and $7$-year junior mezzanine and $10$-year senior mezzanine tranches, spreads decline with increasing 
correlation for high values of correlation. 
\end{itemize}

\begin{table}[ht]
\begin{center}
\small
\begin{tabular}{l|llllllll}
 & 5 Year \\
 \hline
Tranche & Market & $\rho=0.1$ & $\rho=0.2$ & $\rho=0.3$ & $\rho=0.4$ & $\rho=0.5$& $\rho=0.6$ & $\rho=0.7$  \\
\hline \\
  0\%-3\% & 7.19 \% & 7.55 \% & 4.99 \% & 2.14 \% & -0.71 \% & -3.48 \% & -6.17 \% & -8.78 \% \\
  3\%-6\% & 41 & 15.6 & 55.6 & 86.4 & 106.1 & 116.2 & 119.5 & 117.4 \\
  6\%-9\% &  10.8 & 0.7 & 9.1 & 25 & 40.3 & 54.5 & 65.2 & 71.7 \\
  9\%-12\% & 5 & 0 & 2.2 & 8.2 & 18.8 & 28.6 & 37.2 & 45.4 \\
  12\%-22\% & 1.8 & 0 & 0.2 & 1.7 & 4.9 & 9.8 & 16.1 & 22.5 \\
  22\%-100\% & 0.9 & 0 & 0 & 0 & 0.1 & 0.3 & 0.7 & 1.5 \\
  \hline 
  & 7 Year \\
  \hline
Tranche & Market & $\rho=0.1$ & $\rho=0.2$ & $\rho=0.3$ & $\rho=0.4$ & $\rho=0.5$& $\rho=0.6$ & $\rho=0.7$  \\
\hline \\
   0\%-3\% & 22.1 \% & 27.45 \% & 19.97 \% & 13.79 \% & 8.31 \% &  3.27 \% & -1.47 \% & -6.04 \% \\
  3\%-6\% & 110 &  130.6 & 183.3 & 202.2 & 206 & 201.5 & 191.6 & 177.8 \\
  6\%-9\% &  32.5 & 15.3 & 52.4 & 80.5 & 99.1 & 110.6 & 116.1 & 116.9 \\
  9\%-12\% & 15 & 1.8 & 17.4 & 37.1 & 54.3 & 67.1 & 76.5 & 82.7 \\
  12\%-22\% & 4.9 & 0.1 & 2.3 & 8.9 & 19 & 29.9 & 39.5 & 47.9 \\
  22\%-100\% & 2 & 0 & 0 & 0.1 & 0.4 & 1.1 & 2.3 & 4.1 \\ 
\hline 
& 10 Year \\
  \hline
Tranche & Market & $\rho=0.1$ & $\rho=0.2$ & $\rho=0.3$ & $\rho=0.4$ & $\rho=0.5$& $\rho=0.6$ & $\rho=0.7$  \\
\hline \\
  0\%-3\% & 38 \% & 42.51 \% & 32.51 \% & 24.13 \% & 16.65 \% &  9.71 \% & 3.11 \% & -3.33 \% \\
  3\%-6\% & 302.5 & 375.8 & 354.9 & 331.9 & 308.1 & 283.5 & 258.3 & 231.9 \\
  6\%-9\% &  83 & 101.4 & 147.3 & 166.2 & 173.6 & 174.4 & 170.8 & 163.8 \\
  9\%-12\% & 37 & 24.3 & 64.1 & 90.9 & 107.7 & 117.8 & 122.9 & 124.1 \\
  12\%-22\% & 12.5 & 2 & 13.5 & 29.1 & 44.4 & 57.5 & 68.2 & 76.5 \\
  22\%-100\% & 3.6 & 0 & 0.1 & 0.6 & 1.5 & 3 & 5.1 & 7.7 \\  \hline
\end{tabular}
\end{center}
\caption{Model tranche spreads (bp) for varying values of the correlation
parameter. The equity tranches are quoted as an upfront assuming a 500bp running spread.
The model is calibrated to the iTraxx Main Series 6 index for Feb 22, 2007. Market levels shown are for this date; model parameters are $r=0.042$, $\sigma=0.22$, $R=0.4$. }
\label{22Feb07tranchespreads}
\end{table}

\begin{figure}[!htbp]
\begin{center}
\caption{Implied Correlation Skew for iTraxx Main Series 6 Tranches, Feb 22, 2007.}\label{Fig:CorrSkew07}
		\includegraphics[width=12cm]{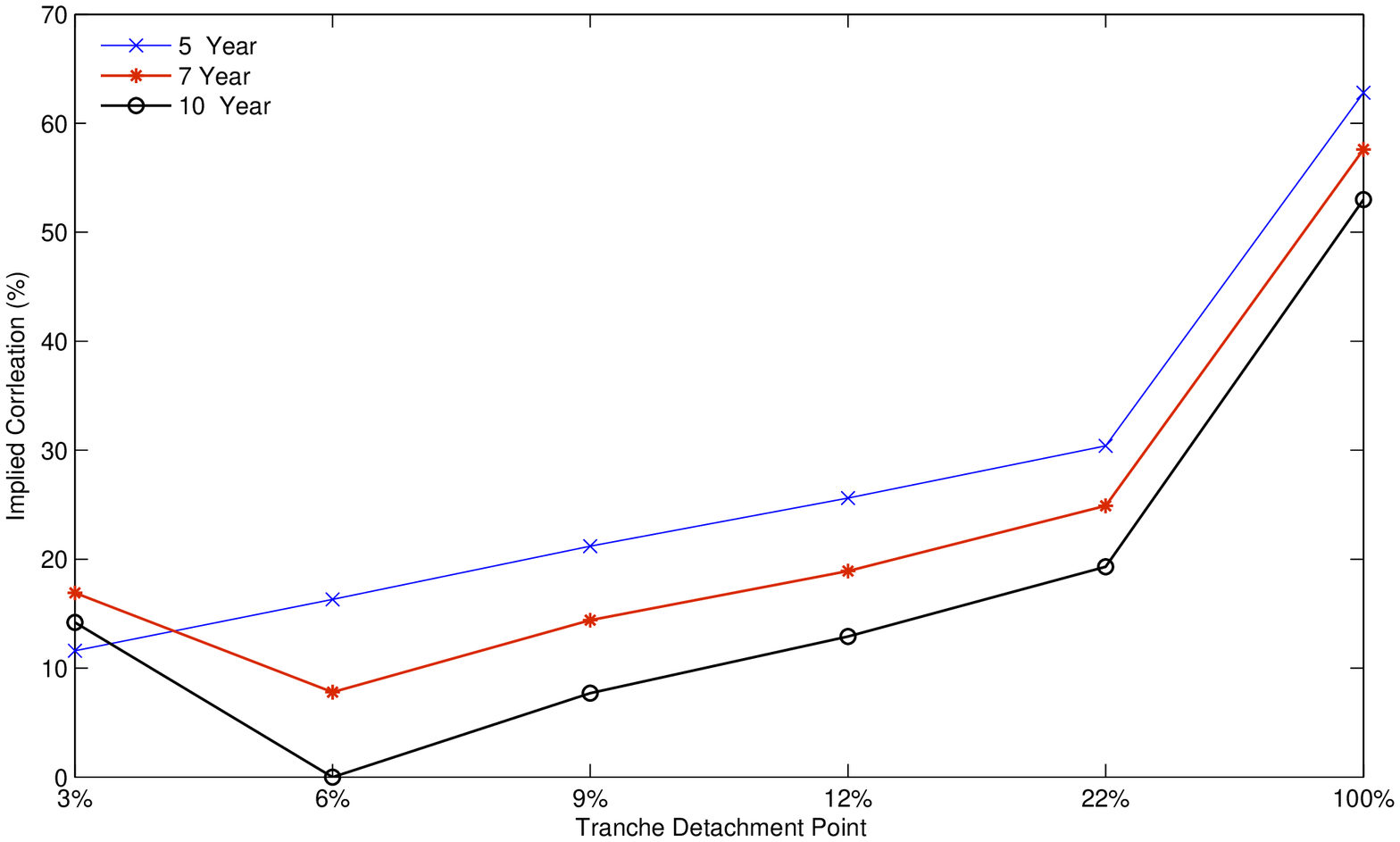}
\end{center}
\begin{center}
\small{The implied correlation for each tranche is the value of correlation that gives a model tranche spread equal to the market tranche spread given in Table \ref{22Feb07tranchespreads}. Model parameters are $r=0.042$, $\sigma=0.22$, $R=0.4$.}
\end{center}
\end{figure}

Figure \ref{Fig:CorrSkew07} illustrates the $5$, $7$ and $10$-year implied correlation skew -- the value of correlation that gives a model spread equal to the market spread for each tranche and maturity. 
\begin{itemize}
\item With the exception of the $0\%-3\%$ tranche, we see similar behaviour and levels for all three maturities. This consistency across the term-structure 
suggests that the dynamics underlying the model are realistic, even in its simple form.
\item
$5$-year implied correlations are generally high relative to the others and $10$-year values relatively low.
To achieve consistency of the correlation parameter across maturities, a driving process with the ability to generate more default events in the short-term would be required, eg a more general Levy process for the market factor.
\item
An anomaly is revealed by the $3\%-6\%$ implied correlations 
and the corresponding row data in Table \ref{22Feb07tranchespreads},
where it is seen that 
the correlation dependence of model tranche
spreads flips from increasing to hump-shaped to decreasing for 
maturities running from 5 to 10 years.
This has the following effect:
for 5 years, there is a unique implied correlation for this tranche;
for 7 years, a second, higher, correlation (just under 1)
also fits this tranche;
for 10 years, only a single high correlation can fit the
market spread. 
Essentially, the implied correlation curves in Figure \ref{Fig:CorrSkew07}
are shifted downwards with increasing maturity.
The alternative 
higher branches, where applicable, are not included in the Figure. 
When a curve crosses zero (in the case of the 10-year $3\%-6\%$ 
tranche),
we have set the implied correlation to zero
(instead of the value of around $0.42$ from the higher branch
which exactly reproduces the
market quote). For pricing and (especially) hedging purposes,
continuous dependence of implied correlations with respect to
maturity and market data is clearly desirable.
The lack of a calibration which is both stable \emph{and} exact underlines the
need for a richer model.
\end{itemize}

\begin{table}[ht]
\begin{center}
\begin{tabular}{|c|c|c|c|}
  \hline
  Maturity Date & Fixed Coupon (bp) & Traded Spread (bp) & Model Spread (bp)\\
  \hline
  20/12/2013 & 120 & 215 & 207 \\
  20/12/2015 & 125 & 195 & 195 \\
  20/12/2018 & 130 & 175 & 176 \\
  \hline
\end{tabular}
\caption{The fixed coupons, traded spreads and model spreads for the iTraxx Main Series 10 index on
December 5, 2008. Parameters used for the model spreads are $r=0.033$, $\sigma=0.136$, $R=0.4$.} \label{05Dec08itraxxspreads}
\end{center}
\end{table}

Table \ref{05Dec08tranchespreads} shows the correlation sensitivity of the Dec 5, 2008 index tranches with parameters from the calibration in Table \ref{05Dec08itraxxspreads}. We notice that relative to Table \ref{22Feb07tranchespreads}, spreads are highly distressed, the index is inverted and tranche spreads are flat to inverted across maturities. As a result, the tranches exhibit very different sensitivity to correlation than before, however there are some common themes and extensions to earlier behaviour:
\begin{itemize}
\item Default probabilities for the index and its constituents are very high. The index expected loss is therefore much greater than before, illustrated by the fact the first three $5$-year tranches and the first four $7$ and $10$-year tranches have declining spreads with increasing correlation. This contrasts with just the equity and $10$-year junior mezzanine tranches in Feb 2007.
\item Much higher levels of $\rho$ are needed to replicate market prices than in pre-crunch times, consistent with the fact that systematic risk is a much greater concern at this time.
\item Too much of our model's portfolio loss distribution lies in the middle tranches: $6\%-22\%$; more weight needs to be in the tail to be able to replicate $22\%-100\%$ tranche values. The same model shortcoming holds for all maturities 
and reflects the need for a more sophisticated driving process. 
\end{itemize}

\begin{table}[ht]
\begin{center}
\small
\begin{tabular}{l|llllllll}
 & 5 Year \\
 \hline
Tranche & Market & $\rho=0.3$ & $\rho=0.4$ & $\rho=0.5$ & $\rho=0.6$ & $\rho=0.7$& $\rho=0.8$ & $\rho=0.9$  \\
\hline \\
  0\%-3\% & 71.5 \%  & 81.88 \% & 75.9 \% &  69.56 \% & 63.02 \% & 56.25 \% & 49.16 \% & 41.65 \%
 \\
  3\%-6\% & 1576.3 & 2275.2 & 1978.5 & 1743.2 & 1546.8 & 1374.6 & 1222.8 & 1090.1 \\
  6\%-9\% & 811.5  & 1273.1 & 1168.2 & 1079.7 & 1001.4 & 931.3 & 864.6 & 796.3 \\
  9\%-12\% & 506.1 & 775.7 & 765.8 & 748.6 & 724.7 & 695.8 & 663.2 & 629.1 \\
  12\%-22\% & 180.3 & 307.8 & 353.3 & 384.7 & 405.5 & 418.1 & 423.4 & 420.5\\
  22\%-100\% & 77.9  & 9.2 & 16.5 & 25 & 34.3 & 44.5 & 55.7 & 68.1 \\
  \hline
  & 7 Year \\
  \hline
Tranche & Market & $\rho=0.3$ & $\rho=0.4$ & $\rho=0.5$ & $\rho=0.6$ & $\rho=0.7$& $\rho=0.8$ & $\rho=0.9$  \\
\hline \\
  0\%-3\% & 72.9 \%  & 84.03 \% & 78.98 \% &  73.26 \% & 66.93 \% & 60 \% & 52.41 \% & 44.13 \% \\
  3\%-6\% & 1473.2  & 2327.3 & 1985.7 & 1715.2 & 1493.4 & 1308 & 1147.8 & 1001.3 \\
  6\%-9\% & 804.2  & 1344.2 & 1199 & 1085.2 & 988.2 & 900.7 & 820.9 & 747.9 \\
  9\%-12\% & 512.4  & 855.4 & 808.4 & 765.3 & 725.3 & 684.8 & 643 & 600.4 \\
  12\%-22\% & 182.6 & 375.4 & 401.7 & 417.6 & 425.6 & 427.4 & 423.1 & 411.8 \\
  22\%-100\% & 75.8  & 14 & 22 & 30.6 & 39.6 & 49.3 & 59.7 & 71.2 \\
  \hline
& 10 Year \\
  \hline
Tranche & Market & $\rho=0.3$ & $\rho=0.4$ & $\rho=0.5$ & $\rho=0.6$ & $\rho=0.7$& $\rho=0.8$ & $\rho=0.9$  \\
\hline \\
  0\%-3\% & 73.8 \%  & 85.13 \% & 80.57 \% &  74.99 \% & 68.51 \% & 61.31 \% & 53.31 \% & 44.22 \% \\
  3\%-6\% & 1385.5 & 2270.8 & 1895.7 & 1611.1 & 1385.8 & 1195.3 & 1032 & 889.6 \\
  6\%-9\% & 824.7  & 1332.2 & 1164.2 & 1033.7 & 925.5 & 833.5 & 749.8 & 669.7 \\
  9\%-12\% & 526.1 & 870.8 & 798.8 & 740.7 & 689.3 & 640.5 & 592.1 & 543.1 \\
  12\%-22\% & 174.1 & 406.1 & 414.9 & 417.5 & 415.6 & 409.8 & 400.2 & 385.3 \\
  22\%-100\% & 76.3 & 18.3 & 26.1 & 34 & 42.1 & 50.6 & 59.7 & 69.8 \\  
  \hline
\end{tabular}
\caption{Model tranche spreads (bp) for varying values of the correlation
parameter. The equity tranches are quoted as an upfront assuming a 500bp running spread.
The model is calibrated to the iTraxx Main Series 10 index for Dec 5, 2008. Market levels shown are for this date; model parameters are $r=0.033$, $\sigma=0.136$, $R=0.4$. }\label{TrancheResults08}
\label{05Dec08tranchespreads}
\end{center}
\end{table}

\subsection{Forward starting CDO contracts}

These contracts are obligations to buy or sell protection on a specified tranche for a
specified spread at some specified time in the future. Although these instruments
are traded infrequently, their
pricing and hedging is an active research topic. We will look at two
types of forward starting CDO: one that resets the cumulative loss at the forward start
date and one that does not. For discussion purposes we will refer to these as resetting
and non-resetting respectively but it should be borne in mind that these are not standard
market terms.

\subsubsection{Non-resetting forward CDO tranche}
For a non-resetting forward CDO tranche defined over the time interval $[T,T^*]$, the
cumulative losses incurred up to time $T$ count towards the total loss in the tranche for
all $t$ with $T<t$. This feature makes pricing straightforward and analogous to a forward
CDS contract.

Consider a portfolio with $m$ entities in the reference portfolio. We define the total loss on the portfolio at time t by
\begin{equation}
L_t = \sum_{i=1}^{m} L_i 1_{\lbrace \tau_{i} \leq t \rbrace}.
\end{equation}
If the forward tranche has attachment point \emph{a} and detachment point \emph{d} then
the outstanding tranche notional, $Z_t$, is given as
\begin{equation}
Z_t = [d - L_t]^{+} - [a - L_t]^{+}.
\end{equation}
The value of the forward tranche contract is again given by the difference between the
fee leg and the protection leg. So far the setup has been the same as the standard CDO
tranche. The only difference when pricing this forward contract is the fact that now we
are only interested in the payment dates $T_i,\, i=1,\ldots,n$ where $T < T_1 < \ldots, <
T_n\leq T^*$. Using these payment dates the present value of the coupon payments given a
forward spread $s$ is
\begin{equation}\label{eqnfeeleg}
s V^{fee} = s \sum_{i=1}^{n} \frac{\delta_{i}}{b(T_i)}
\mathbb{E}^{\mathbb{Q}}[Z_{T_{i}}].
\end{equation}
The protection leg is given by
\begin{equation}\label{eqnprotleg}
V^{prot} = \sum_{i=1}^{n} \frac{1}{b(T_{i})}\mathbb{E}^{\mathbb{Q}}[Z_{T_{i-1}} -
Z_{T_{i}}].
\end{equation}
Today, the value of the forward starting contract is zero and hence the forward
break-even spread is given by
\begin{equation} \label{breakevenspread}
s=\frac{V^{prot}}{V^{fee}}.
\end{equation}

\subsubsection{Resetting forward CDO tranche}
With this contract, the cumulative loss up to time $T$ is ignored and the value of the
tranche is dependent only on the further loss incurred after time $T$. If the forward
tranche has attachment \emph{a} and detachment \emph{d} then this is equivalent to a
non-resetting forward tranche with attachment $(L_T+a)$ and detachment $(L_T+d)$. Using
the same payment dates as the non-resetting forward contract we define the effective
forward loss at time $T_i$ by
\begin{equation}
\hat{L}_{T_i}=L_{T_i}-L_T,
\end{equation}
which gives the forward tranche notional as
\begin{equation}
Z_t = [d - \hat{L}_t]^{+} - [a - \hat{L}_t]^{+}.
\end{equation}
With these new definitions the forward break-even spread can be calculated as before
using (\ref{eqnfeeleg}), (\ref{eqnprotleg}) and (\ref{breakevenspread}).

\subsection{Forward pricing results for the pre-crunch state early 2007}

We value resetting and non-resetting forward CDO contracts for a range of
forward starting dates $T$. The data used is for the European iTraxx Main Series 6 index from February 22 2007. The index fixed coupons and traded spreads are shown in table \ref{22Feb07itraxxspreads} and we use a constant risk-free rate of 4.2\% obtained from the Euro swap curve. The tenor of the forward contracts is always five years i.e.
$T^*-T=5$. The forward dates we use are 0 years i.e. the spot spread and the 1, 3 and 5
year forward starting dates. The forward break-even spreads for the non-resetting and
resetting forwards are shown in table \ref{forward-cdos-2007}.

\begin{table}[ht]
\begin{center}
\begin{tabular}{l|llll|}
$\rho = 0.1$ & \multicolumn{4}{c|}{non-resetting} \\ \hline
& $T=0$ & $T=1$ & $T=3$ & $T=5$ \\
\hline
0\%-3\% &  7.55 \% 
&  24.06 \% & 51.32 \% & 64.31 \% \\
3\%-6\% & 15.6 &   71.0 & 374.2 & 911 \\
6\%-9\% &  0.7 &  4.7 & 58.7 & 231 \\
9\%-12\% &  0 &  0.7 & 9.7 & 54 \\
12\%-22\% &  0 &       0 & 0.7 & 4.5 \\
22\%-100\% &  0 &       0 & 0 & 0 \\ \hline
\end{tabular}
\begin{tabular}{lll}
\multicolumn{3}{c}{resetting} 
\\ \hline
$T=1$ & $T=3$ & $T=5$ \\
\hline
24.05 \% & 50.26 \% & 60.95 \% \\
70.9 & 312.3 & 484.0 \\
4.7 & 43.5 & 79.5 \\
0.7 & 6.6 & 10.3 \\
0 & 0.3 & 0.5 \\
0 & 0 & 0\\ \hline
\end{tabular}\\
\begin{tabular}{l|llll|}
$\rho = 0.3 $ & \multicolumn{4}{c|}{}
\\ \hline
0\%-3\% & 2.14 \% & 13.87 \% & 31.87 \% & 36.77 \% \\
3\%-6\% & 86.4 & 174.2 & 441.9 & 704.3 \\
6\%-9\% & 25.0 & 62.6 & 188.5 & 358.6 \\
9\%-12\% & 8.2& 24.1 & 88.1 & 197.6 \\
12\%-22\% & 1.7 & 5.0 & 24.3 & 63.2 \\
22\%-100\% & 0 & 0.1 & 0.4 & 1.3 \\ \hline
\end{tabular}
\begin{tabular}{lll}
\\ \hline
   13.87 \% & 33.04 \% & 42.44 \% \\
    174.1 & 402.9 & 545.3 \\
    62.6 & 164.1 & 231.1 \\
    24.1 & 74.4 & 107.9 \\
    5.0 & 19.9 & 25.8 \\
    0 & 0.2 & 0.2 \\ \hline
\end{tabular}\\ 
\begin{tabular}{l|llll|}
$\rho = 0.5$ & \multicolumn{4}{c|}{}
\\ \hline
0\%-3\%& -3.48 \%\hspace{0mm} &  5.18 \%\hspace{0.5mm} & 17.96 \% & 20.16 \% \\
3\%-6\%&  116.2 &  193.9 & 400.8 & 532.2 \\
6\%-9\%&   54.5 & 101.6 & 228.2 & 341.4 \\
9\%-12\%&   28.6 & 59.8 & 142.9 & 237.0 \\
12\%-22\%&   9.8 & 22.0 & 62.0 & 118.5 \\
22\%-100\%&    0.3 & 0.8 & 2.9 & 6.4
\end{tabular}
\begin{tabular}{lll}
\\
\hline
    5.19 \%\hspace{1.7mm} & 20.11 \% & 28.31 \% \\
    193.8 & 384.5 &  507.8 \\
    101.5 & 211.2 & 283.7 \\
    59.7 & 128.9 & 177.0 \\
    22.0 & 54.4 & 73.9 \\
    0.8 & 2.3 & 2.1
\end{tabular}
\caption{The non-resetting and resetting forward spreads (bp) for varying values of the
correlation
parameter. The equity tranches are quoted as an upfront assuming a 500bp running spread.
The model is calibrated to the iTraxx Main Series 6 index for {\bf 22 Feb 2007}. All forwards
have a tenor of 5 years.}\label{forward-cdos-2007}
\end{center}
\end{table}

\begin{figure}[ht]
\psfrag{t}[r][r][0.8]{$t$}
\psfrag{l}[d][t][0.8]{$\frac{{\rm d}}{{\rm d} t} \mathbb{E}^{\mathbb{Q}}[Y_{t}]$}
\includegraphics[width= 0.47 \columnwidth]{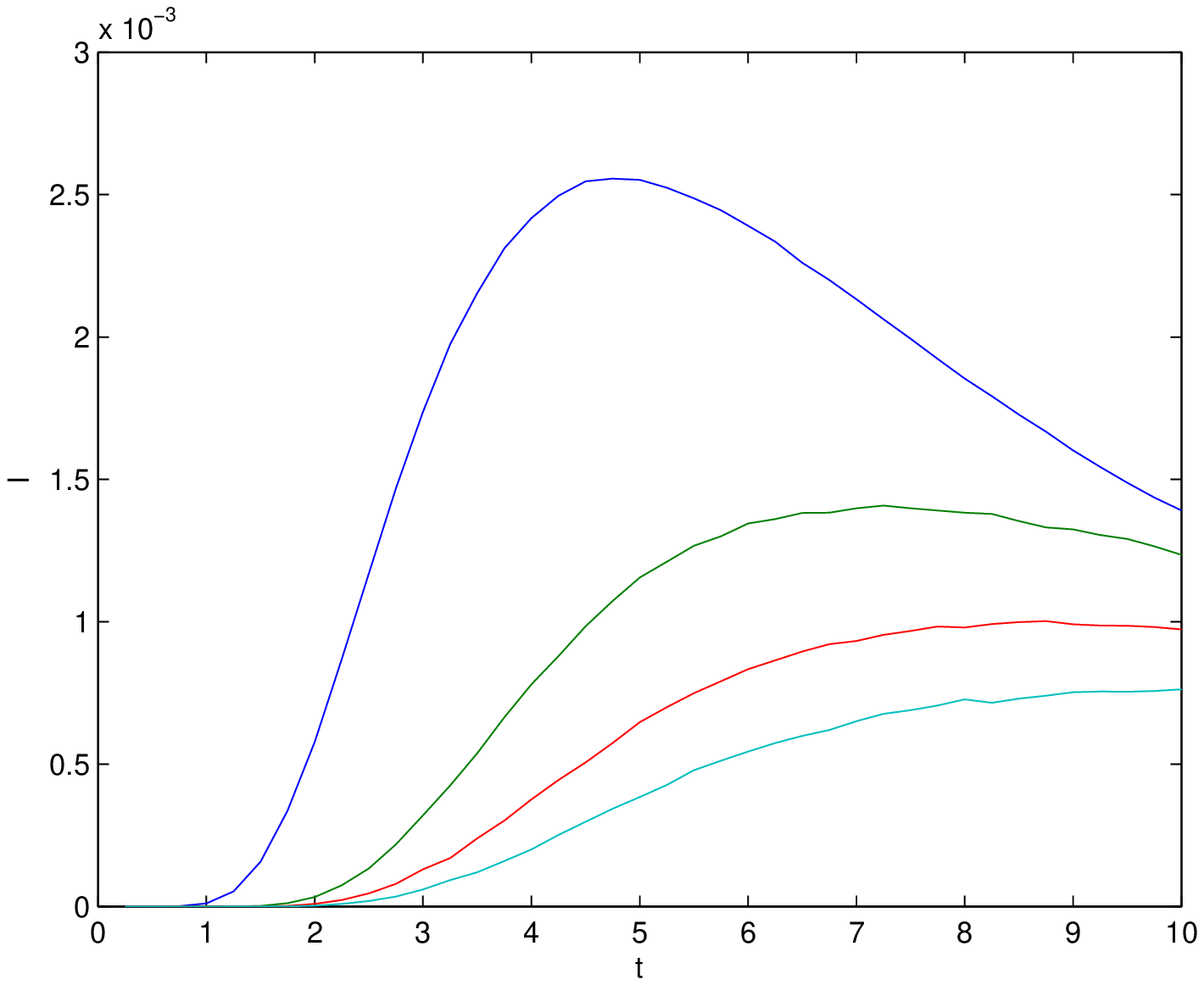}
\hfill
\includegraphics[width= 0.47 \columnwidth]{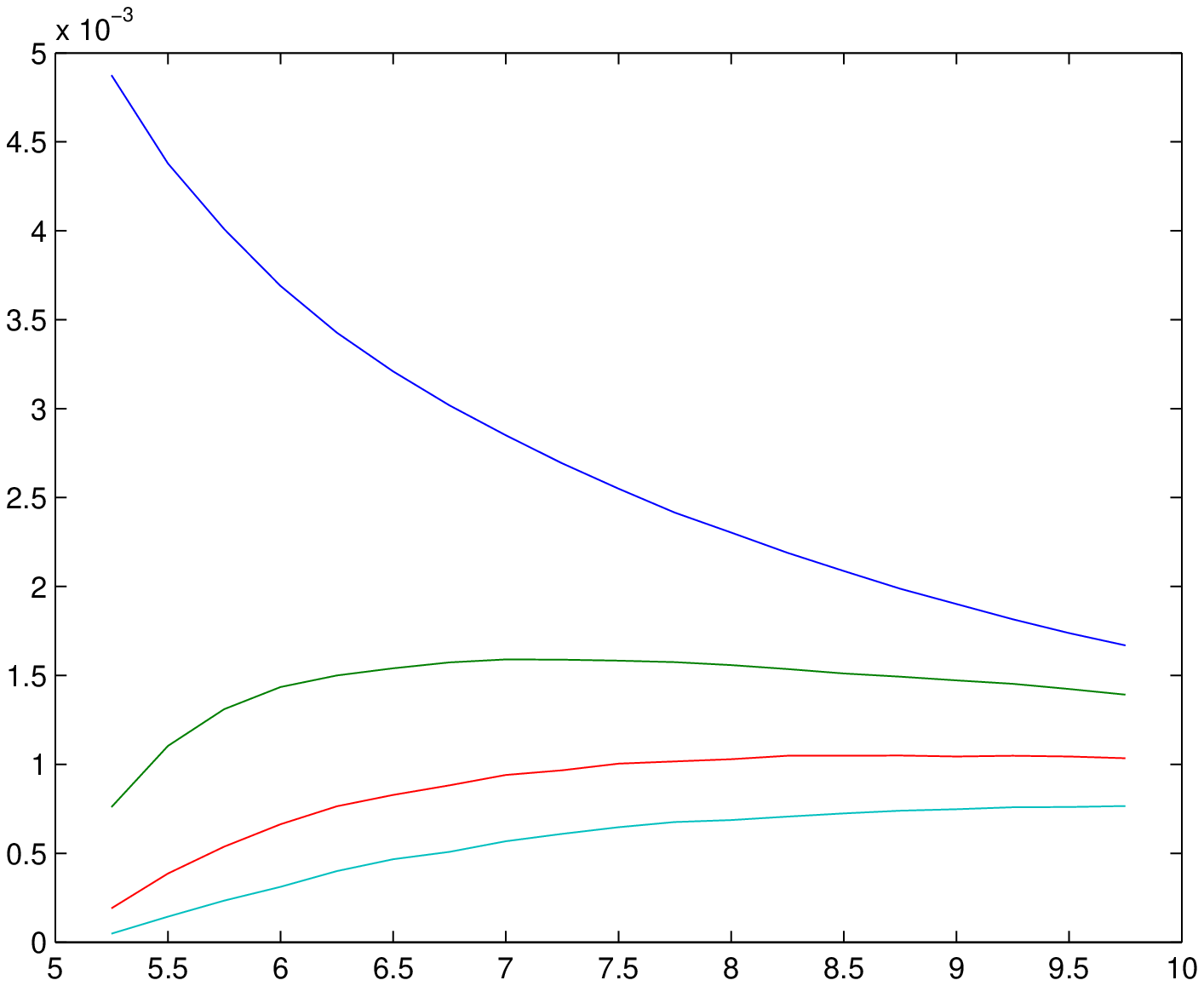}
\caption{Rate of expected losses in tranches
$[0,3\%]$, $[3\%,6\%]$, $[6\%,9\%]$, $[9\%,12\%]$ in 2007, for correlation $\rho=0.5$,
for non-resetting (left) and resetting (right) losses.}
\label{fig:ratetranchelosses2007}
\end{figure}

\subsubsection{Non-resetting forward CDO tranche}

First we focus on the non-resetting tranches. From table \ref{forward-cdos-2007} we observe
the following points:
\begin{itemize}
\item As the forward start date increases, the break-even forward spread increases for
all tranches.
\item For the junior mezzanine (3-6\%) tranche, as the forward start date increases the
spread sensitivity to correlation changes sign. The sensitivities of all other tranches
are single signed.
\end{itemize}
Both of these observations can be explained by the fact that losses in the portfolio are
cumulative. As time passes, the total loss in the portfolio accumulates and so the
attachment and detachment points of non-resetting forward tranches effectively move down
the capital structure. In other words, forward equity tranches start behaving like very
narrow spot equity tranches, forward junior mezzanine tranches start behaving like spot
equity tranches and so on. On the forward start date, investors will require additional
compensation for holding these now riskier tranches and so the break-even forward spread
increases.

This also gives the reason why the correlation sensitivity of the junior mezzanine
tranche changes sign. For a start date sufficiently far into the future, the tranche is
expected to be an equity tranche which has a negative correlation sensitivity.

\subsubsection{Resetting forward CDO tranche}

Now we turn to the resetting tranches which from a dynamic modelling point of view can be
considered the more interesting of the forward contracts. From table \ref{forward-cdos-2007} we
observe the following points:
\begin{itemize}
\item For the forward tranches with a start date in 1 years time the break-even spreads
are the same as for the non-resetting forward tranches.

\item As the forward start date increases, the break-even forward spread date generally increases.

\end{itemize}
Addressing these observations in order, the reason for the first point is the
nature of the structural model. Because of the diffusive nature of the asset processes,
the probability of defaults occurring in the short term is very low. In Section 3 this
was highlighted as one of the major downsides for this type of model. The consequence of
that property here is that the cumulative loss within the first year is negligible and so
we have $\hat{L}_t \approx L_t$. This in turn leads to the same break-even spreads for
both types of forward contract.

The second point can be explained simply by the potential for a decrease in credit
quality due to the natural diffusion of the asset processes. This is also present in the
non-resetting tranche prices but there the increase in spreads is dominated by the move
down the capital structure.
An exception is the super senior (22-100\%) tranche for high correlation. We come back to
this when discussing the distressed state where this effect is more pronounced.

\subsection{Forward pricing results for the distressed state late 2008}

We again value resetting and non-resetting forward CDO contracts for a range of
forward starting dates $T$. The data used is for the European iTraxx Main Series 6 index from December 5 2008. The index fixed coupons and traded spreads are shown in table \ref{05Dec08itraxxspreads} and we use a constant risk-free rate of 3.3\% obtained from the Euro swap curve. The contract details are as before.
Forward break-even spreads for the non-resetting and
resetting forwards are shown in tables \ref{forward-cdos-2008}.

\begin{table}[ht]
\begin{center}
\begin{tabular}{l|llll|}
$\rho = 0.3$ & \multicolumn{4}{c|}{non-resetting} \\ \hline
& $T=0$ & $T=1$ & $T=3$ & $T=5$ \\
\hline
0\%-3\% & 81.88 \% &   82.40 \%& 69.26 \% & 56.96 \% \\
3\%-6\% & 2275.2 &   3106.5 & 2658.7 & 1989.3 \\
6\%-9\% & 1273.1 &   1708.0 & 1771.4 & 1401.2 \\
9\%-12\% & 775.7 &   1045.9 & 1205.5 & 1024.9 \\
12\%-22\% & 307.8 &    425.3 & 570.1 & 546.6 \\
22\%-100\% & 9.2 &    13.9 & 24.1 & 29.4  \\ \hline
\end{tabular}
\begin{tabular}{lll}
\multicolumn{3}{c}{resetting} 
\\ \hline
$T=1$ & $T=3$ & $T=5$ \\
\hline
   83.21 \% & 77.69 \% & 68.98 \% \\
   2479.9 & 1830.2 & 1206.6 \\
   1348.2 & 901.9 & 522.1 \\
    811.7 & 487.7 & 238.9 \\
    313.3 & 143.9 & 51.4 \\
    8.0 & 1.8 & 0.2 \\ \hline
\end{tabular}
\\ 
\begin{tabular}{l|llll|}
$\rho = 0.5$ & \multicolumn{4}{c|}{}
\\ \hline
0\%-3\% & 69.56 \% & 66.22 \% & 47.22 \% & 34.18 \% \\
3\%-6\% &  1743.2 & 2090.3 & 1580.0 & 1119.3 \\
6\%-9\%  &  1079.7 & 1341.4 & 1149.0 & 869.2 \\
9\%-12\%  &   748.6 & 938.6 & 890.4 & 690.9 \\
12\%-22\%  &   384.7 & 495.5 & 532.2 & 455.6 \\
22\%-100\%  &   25 & 33.4 & 44.8 & 44.8 \\ \hline
\end{tabular}
\begin{tabular}{lll}
\\ \hline
  70.09 \% & 63.07 \% & 53.86 \% \\
  1858.6 & 1441.5 & 1033.5 \\
  1141.9 & 845.4 & 565.4 \\
   787.1 & 549.8 & 333.5 \\
   400.0 & 241.7 & 123.6 \\
   23.3 & 8.6 & 2.5 \\ \hline
\end{tabular}
\\ 
\begin{tabular}{l|llll|}
$\rho = 0.7$ & \multicolumn{4}{c|}{}
\\ \hline
0\%-3\% & 56.25 \% & 49.41 \% & 29.70 \% & 16.99 \% \\
3\%-6\% & 1374.6 &   1481.6 & 1040.3 & 734.7 \\
6\%-9\%  & 931.3 &   1072.1 & 827.8 & 602.8 \\
9\%-12\%  &  695.8 &    818.8 & 689.6 & 516.9 \\
12\%-22\%  & 418.1 &    513.8 & 478.1 & 384.8 \\
22\%-100\%  & 44.5 &     55.9 & 64.1 & 57.9 
\end{tabular}
\begin{tabular}{lll}
\\ \hline
   55.67 \% & 47.50 \% & 38.76 \% \\
   1428.3 & 1146.6 & 864.3 \\
    973.7 & 766.2 & 549.2 \\
    731.3 & 554.3 & 379.7 \\
    443.0 & 310.4 & 188.1 \\
    42.7 & 20.8 & 9.5
\end{tabular}
\caption{The non-resetting and resetting forward spreads (bp) for varying values of the 
correlation
parameter. The equity tranches are quoted as an upfront assuming a 500bp running spread.
The model is calibrated to the iTraxx Main Series 6 index for {\bf 5 Dec 2008}. All forwards
have a tenor of 5 years.}\label{forward-cdos-2008}
\end{center}
\end{table}

\begin{figure}[ht]
\psfrag{t}[r][r][0.8]{$t$}
\psfrag{l}[d][t][0.8]{$\frac{{\rm d}}{{\rm d} t} \mathbb{E}^{\mathbb{Q}}[Y_{t}]$}
\includegraphics[width= 0.47 \columnwidth]{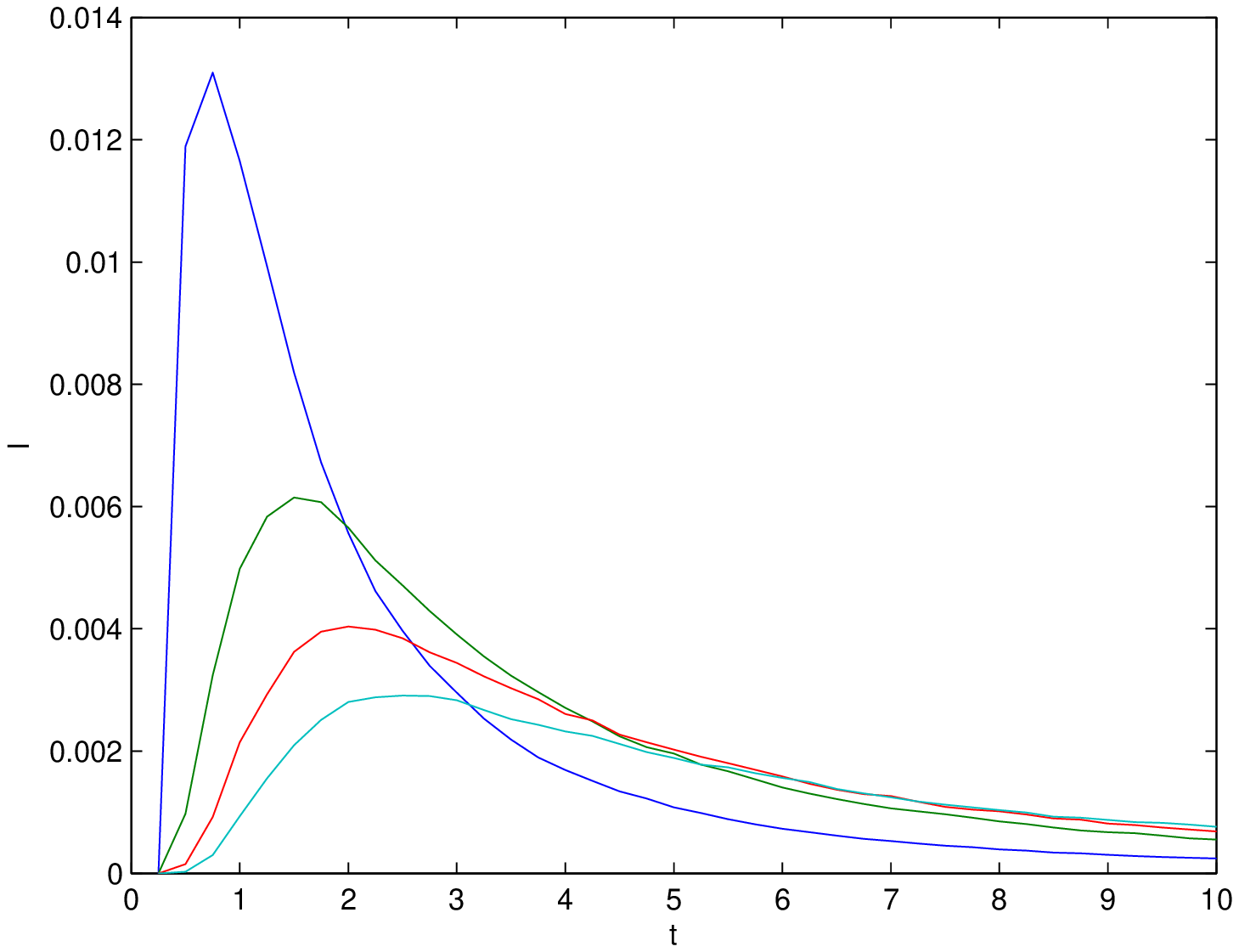}
\hfill
\includegraphics[width= 0.47 \columnwidth]{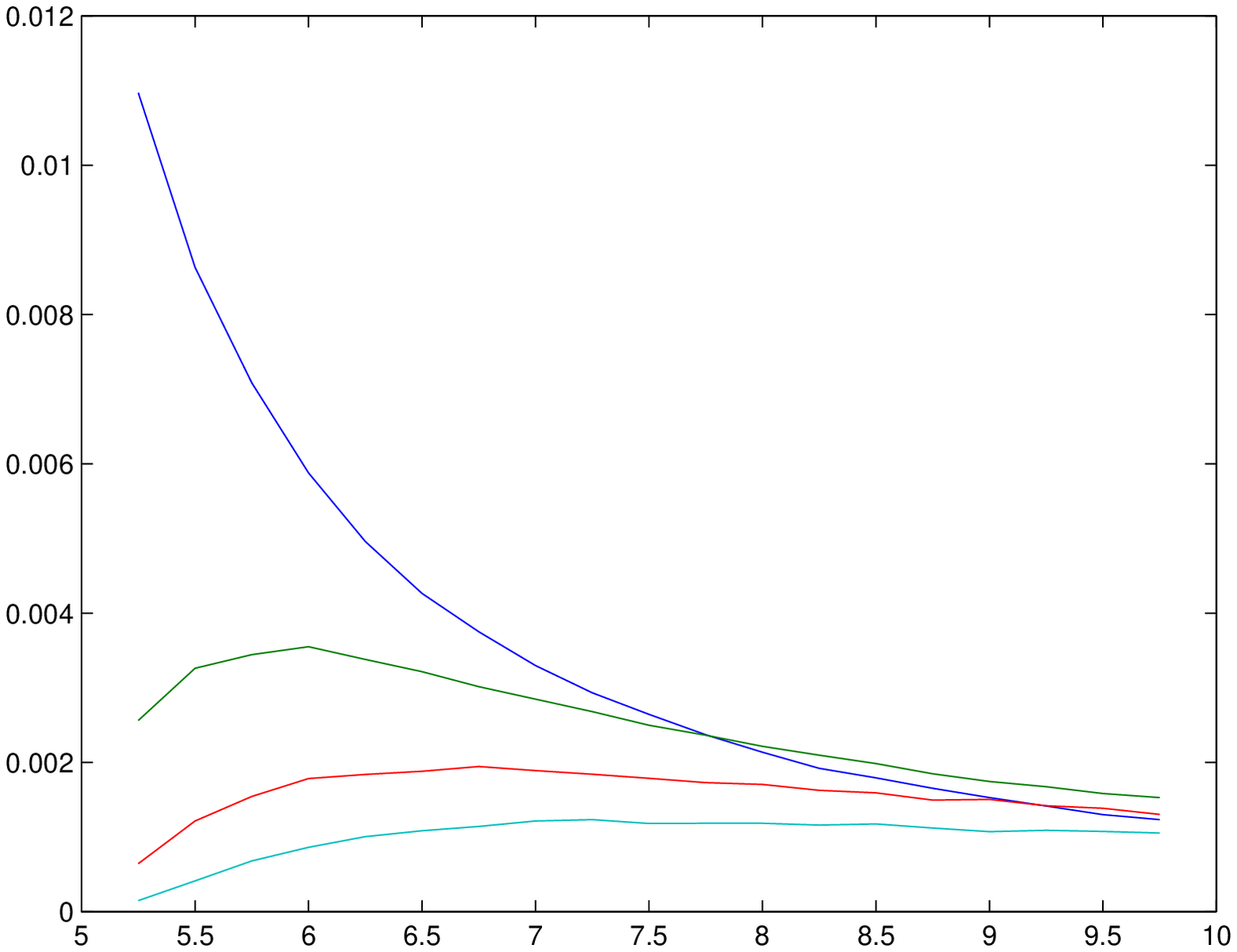}
\caption{Rate of expected losses in tranches
$[0,3\%]$, $[3\%,6\%]$, $[6\%,9\%]$, $[9\%,12\%]$ in 2008, for fixed correlation $\rho=0.5$,
for non-resetting (left) and resetting (right) losses.}
\label{fig:ratetranchelosses2008}
\end{figure}

Comparing figures \ref{fig:ratetranchelosses2008} with \ref{fig:ratetranchelosses2007},
the expected
loss rate peaks at the short end in the distressed environment of 2008.
The effect on the resetting forward starting CDO, which is basically a standard CDO moved
into the future, is that the tranche spreads decrease with the forward start date. This
is in contrast to the 2007 environment where the risk is generally perceived to increase
with the forward start date.

The behaviour is more involved for the non-resetting tranches, where the outstanding
tranche notional decays leading up to the forward start date. The net
effect here is that typically the spreads decrease, with respect to the forward start date,
for the junior tranches, increase for the senior tranches, and have a hump-shaped
term-structure in the mezzanine range.

\section{Conclusions}
\label{conclusions}

We have illustrated the ability of our simple model to crudely calibrate to the index term-structure in wildly different market 
environments, and have shown that the correlation sensitivity of tranche spreads demonstrates the behaviour expected. 
More importantly, using just two parameters and without making them time-dependent, we have shown that our very simple 
structural evolution model displays realistic term-structure dynamics. Using just the volatility parameter, it is able 
to calibrate well to all three index spreads and correlation sensitivities of the various tranches are fairly stable 
across maturities. This is an improvement on the majority of pricing models which lack a coherent means of 
incorporating dynamics. 

The next stage, which has not been the focus here, is to extend the framework so that it can 
calibrate to all tranches with a single set of parameters. This will involve moving away from a simple 
Brownian Motion driving the process, and may include a more general stochastic volatility or Levy or jump-diffusion process.
Jumps in the market factor are conceptually easy to include and result in a jump
process driving the SPDE drift.
Similarly, a single stochastic volatility factor affecting all firms will result
in a stochastic term driving the SPDE diffusion.
Contagion may be incorporated by making model parameters, notably the correlation,
loss dependent.
These extensions allowe the loss distribution 
process to become more skewed, allocating more weight to the tail and increasing super senior tranche spreads, as well as generally allowing more flexibility to
match observed data.


\begin{thebibliography}{}

\bibitem{Aldous1985}
Aldous, D. (1985) Exchangeability and related topics, \emph{{E}cole d'{E}te {S}t {F}lour
  1983, Springer Lecture Notes in Mathematics}, 1117, 1--198.

\bibitem{BainCris} Bain, A. and Crisan, D. (2009) \emph{Fundamentals of stochastic filtering}, Springer.

\bibitem{bayrak} Bayraktar, E., \& Yang, B. (2009) Multi-scale time--changed birth processes for pricing 
multi-name credit derivatives, \emph{Appl. Math. Fin.}, 16(5), 429-449.

\bibitem{brigo1} Brigo, D., Pallavicini, A., \& Torresetti, R. (2006) Calibration of CDO tranches with 
the dynamical generalized-Poisson loss model. Working paper.

\bibitem{blackcox1} Black, F., \& Cox, J. (1976) Valuing corporate securities: some effects of bond
            indenture provisions, \emph{J. Finance}, 31, 351--367.



\bibitem{schoutens1} Cariboni, J., \& Schoutens, W. (2004)
Pricing credit default swaps under L\'{e}vy models, UCS report 2004--2007, K.U. Leuven.

\bibitem{crepey} Carmona, R., \& Crep{\'e}y, S. (2010) 
Particle methods for the
estimation of credit portfolio loss distributions. To appear in
\emph{Int. J. Th. Appl. Fin.}.


\bibitem{carmona} Carmona, R., Fouque, J.-P., \& Vestal, D. (2009) Interacting particle systems for the
computation of rare credit portfolio losses, \emph{Fin. Stoch.}, 13(4), 613--633.

\bibitem{debbabode01} Deb, M.K., Babuska, I.M., \& Oden, J.T. (2001) Solution of stochastic
partial differential equations using Galerkin finite element techniques
\emph{Comp. Methods Appl. Mech. Engrg.}, 190, 6359--6372.

\bibitem{graz1} Di Graziano, G., \& Rogers, C. (2006)
A dynamic approach to the modelling of correlation credit derivatives using Markov chains. Working paper.

\bibitem{duffie2} Duffie, D., \& Lando, D. (2001) Term structure of credit spreads with incomplete
accounting information, \emph{Econometrica}, 69, 633--664.


\bibitem{errais1} Errais, E., Giesecke, K., \& Goldberg, L. (2010) Affine point processes and portfolio 
credit risk, \emph{SIAM J. Fin. Math.}, 1, 642--665.

\bibitem{Ethier1986}
Ethier, S.N. and Kurtz, T.G. (1986)
\emph{Markov Processes: Characterization and Convergence},
Wiley.


\bibitem{filip} Filipovic, D., Overbeck, L., \& Schmidt, T. (2011) Dynamic CDO term structure 
modelling, \emph{Math. Fin.}, 21(1), 53--71.

\bibitem{finger1} Finger, C., Finkelstein, V., Pan, G., Lardy, J., Thomas, T., \&
Tierney, J. (2002) Creditgrades$^{\textrm{TM}}$, technical document, RiskMetrics Group,
Inc. Modelling

\bibitem{fouque1} Fouque, J., Wignall B., \&  Zhou, X. (2008) Modeling correlated defaults: first passage model
under stochastic volatility, \emph{J. Comp. Fin.}, 11(3), 43--78.




\bibitem{haworth1} Haworth, H., \& Reisinger, C. (2007) Modelling basket credit
default swaps with default contagion, \emph{J. Credit Risk}, 3(4), 31--67.

\bibitem{hig00} Higham, D.J. (2000) Mean-square and asymptotic stability of the stochastic theta method, \emph{SIAM J. Num. Anal.}, 
38(3), 753--769.

\bibitem{hilberink1} Hilberink, B., \& Rogers, L.C.G. (2002) Optimal capital structure and endogenous
default, \emph{Fin. Stoch.}, 6, 237--263.

\bibitem{hull2} Hull, J., \& White, A. (2001) Valuing credit default swaps II: modeling default
correlations, \emph{J. Deriv.}, 8, 12--22.

\bibitem{hull1} Hull, J., Predescu M., \& White, A. (2010) The valuation of correlation-dependent credit
derivatives using a structural model, \emph{J. Credit Risk}, 6(3), 99--132.

\bibitem{hull3} Hull, J., \& White, A. (2007) Forwards and European options on CDO tranches, \emph{J. Credit Risk}, 3(2), 63--73.

\bibitem{Iyengar1985}
Iyengar, S. (1985)
\newblock Hitting lines with two-dimensional Brownian motion,
\newblock \emph{SIAM J. Appl. Math.}, 45, 983--989.

\bibitem{jackson1} Jackson, K., \& Zhang, W. (2009) Valuation of forward starting
CDOs, \emph{Int. J. Comp. Math.}, 86(6), 955--963.




\bibitem{LJthesis} Jin, L. (2010) Particle systems and SPDEs with application to credit modelling. \emph{D.Phil. Thesis} University
of Oxford.

\bibitem{joshi1} Joshi, M., \& Stacey, A. (2006) Intensity gamma, a new approach to pricing portfolio 
credit derivatives. \emph{Risk}. 

\bibitem{Karat1991}
Karatzas, I., \& Shreve, S.E. (1991)
\newblock \emph{Brownian motion and stochastic calculus, second edition}.
\newblock Springer.


\bibitem{Kry94} Krylov, N.V. (1994) A $W^n_2$-theory of the Dirichlet problem for SPDEs in general
smooth domains, \emph{Probab. Theory Relat. Fields}, 98, 389--421.

\bibitem{ku1999} Kurtz, T.G., \& Xiong, J., (1999) Particle
representations For a class of non-linear SPDEs, \emph{Stoch. Proc. Appl.}, 83,
103--126.





\bibitem{Lot99} Lototsky, S.V. (1999) Dirichlet problem for stochastic parabolic equations in smooth domains,
\emph{Stochastics Stochastics Rep.}, 68, 145--175.

\bibitem{merton1} Merton, R. (1974) On the pricing of corporate debt: the risk structure of interest
rates, \emph{J. Finance}, 29, 449--470.


\bibitem{Metzler2010}
Metzler, A. (2010)
\newblock On the first passage problem for correlated {B}rownian motion,
\newblock \emph{Statistics and Probability Letters}, 80, 277--284.


\bibitem{mortensen1} Mortensen, A. (2006) Semi-analytical valuation of basket credit derivatives in 
intensity-based models, \emph{J. Derivatives}, 13(4), 8--26.


\bibitem{okane2} O'Kane, D. (2008) \emph{Modelling single-name and multi-name credit derivatives}, Wiley Finance. 


\bibitem{papasir} Papageorgiou, E., \& Sircar, R. (2009) Multiscale intensity models and name grouping for valuation of 
multi-name credit derivatives, \emph{Appl. Math. Fin.}, 16(4), 353--383. 

\bibitem{pooley1} Pooley, D.M.,
Vetzal, K.R., \& Forsyth, P.A. (2003)
Remedies for non-smooth payoffs in option pricing,
\emph{J. Comp. Fin.}, 6, 25--40.


\bibitem{quaval97} Quarteroni, A., \& Valli, A. (1997) \emph{Numerical approximation of
partial differential equations}, Springer.


\bibitem{rannacher1} Rannacher, R. (1984)
Finite element solution of diffusion problems with irregular data,
\emph{Numerische Mathematik}, 43, 309--327.

\bibitem{Revuz2005}
Revuz, D. and Yor, M. (2005) \emph{ Continuous Martingales and Brownian Motion}, Third
  Edition, Springer.


\bibitem{schon2} Sch{\"o}nbucher, P.J. (2005) Portfolio losses and the term structure of loss 
transition rates: a new methodology for the pricing of portfolio credit derivatives. Working paper.

\bibitem{sirzar} Sircar, R., \& Zariphopoulou, T. (2010) Utility valuation of credit 
derivatives and application to CDOs, \emph{Quant. Finance}, 10(2), 195--208.




\bibitem{zhou2} Zhou, C. (1997) A jump-diffusion approach to modelling credit risk and valuing
defaultable securities, Federal Reserve Board, Washington.

\bibitem{zhou3} Zhou, C. (2001a) The term structure of credit spreads with jump risk, \emph{J.
Bank. Fin}, 25, 2015--2040.

\bibitem{zhou1} Zhou, C. (2001b) An analysis of default correlations and multiple
defaults, \emph{Rev. Fin. Studies}, 14, 555--576.

\end{thebibliography}
\end{document}